\documentclass[12pt]{article}
\usepackage{amsmath}
\usepackage{amssymb}
\usepackage{amsthm}
\usepackage{graphicx}
\usepackage{enumitem}
\usepackage{natbib}
\usepackage{hyperref}
\hypersetup{
    colorlinks=true,
    linkcolor=blue,
    citecolor=blue,
    filecolor=magenta,      
    urlcolor=green,
}
\usepackage{bm}
\usepackage{tikz-cd}
\usepackage[linesnumbered,ruled,vlined]{algorithm2e}
\usepackage{subcaption}
\usepackage{tabularx}
\usepackage{longtable}
\usepackage{multirow}
\newcommand{\blind}{1}

\newtheorem{thm}{Theorem}
\newtheorem{lem}{Lemma}

\newtheorem{prop}{Proposition}

\newtheorem{exm}{Example}

\newtheorem{defi}{Definition}

\DeclareMathOperator*{\argmin}{arg\,min}

\def\bco{\iffalse}

\def\Var{{\rm Var}}

\def\f{Fr\'echet }
\def\T{{ \mathrm{\scriptscriptstyle T} }}

\newcommand{\single}{\renewcommand{\baselinestretch}{1.2}\normalsize}
\newcommand{\double}{\renewcommand{\baselinestretch}{1.9}\normalsize}

\begin{document}
\def\thefootnote{*}\footnotetext{These authors contributed equally to this work}

\def\spacingset#1{\renewcommand{\baselinestretch}%
{#1}\small\normalsize} \spacingset{1}


\if1\blind
{
  \title{\bf Deep \f Regression}
  \author{Su I Iao$^*$, Yidong Zhou$^*$, Hans-Georg M\"{u}ller\\
  Department of Statistics, University of California, Davis}
  \date{May 9, 2025}
  \maketitle
} \fi

\if0\blind
{
  \bigskip
  \bigskip
  \bigskip
  \begin{center}
    {\LARGE\bf Deep \f Regression}
\end{center}
  \medskip
} \fi

\bigskip
\begin{abstract}
Advancements in modern science have led to the increasing availability of non-Euclidean data in metric spaces. This paper addresses the challenge of modeling relationships between non-Euclidean responses and multivariate Euclidean predictors. We propose a flexible regression model capable of handling high-dimensional predictors without imposing parametric assumptions. Two primary challenges are addressed: the curse of dimensionality in nonparametric regression and the absence of linear structure in general metric spaces. The former is tackled using deep neural networks, while for the latter we demonstrate the feasibility of mapping the metric space where responses reside to a low-dimensional Euclidean space using manifold learning. We introduce a reverse mapping approach, employing local \f regression, to map the low-dimensional manifold representations back to objects in the original metric space. We develop a theoretical framework, investigating the convergence rate of deep neural networks under dependent sub-Gaussian noise with bias. The convergence rate of the proposed regression model is then obtained by expanding the scope of local \f regression to accommodate multivariate predictors in the presence of errors in predictors. Simulations and case studies show that the proposed model outperforms existing methods for non-Euclidean responses, focusing on the special cases of probability distributions and networks.
\end{abstract}

\noindent%

\newpage
\spacingset{1.9} 
\section{Introduction}
\label{sec:intro}
Regression analysis serves as the cornerstone of statistical methodology, providing a fundamental framework for modeling complex relationships between response and predictor variables. The demand for comprehensive and versatile regression techniques has soared, driven by the prevalence of non-Euclidean responses. These responses encompass a broad spectrum of data types, ranging from networks \citep{mull:22:11}, covariance matrices \citep{dryd:09} and trees \citep{nye:17} to distributions \citep{pete:22} and other complex structures.

Real-world applications across various domains underscore the need for advanced regression methods. For instance, in transportation science, understanding the evolution of transportation networks in response to factors like weather conditions and public health data is of paramount interest \citep{mull:21:1}. These and related scenarios are covered by a general approach where one views the responses as random objects situated in metric spaces \citep{mull:16:7}.  Challenges that we address in this paper include the absence of a linear structure and the potential high dimensionality of multivariate predictors.

Regression with metric space-valued responses is thus of paramount interest. Early explorations include regression for non-Euclidean data using Euclidean embedding with distance matrices \citep{fara:14} and a Nadaraya-Watson kernel regression approach \citep{hein:09}. A \f regression approach was introduced more recently \citep{mull:19:6,mull:22:8,scho:22}. It extends linear and nonparametric regression techniques to accommodate metric space-valued responses through conditional \f means. However, these regression models either rely on strict assumptions akin to classical linear models for scalar responses or suffer from the curse of dimensionality in analogy to traditional nonparametric regression.

In practical applications, the need often arises to flexibly model metric space-valued responses with potentially high-dimensional predictors, so that these limitations prevent application of these methods for real-world challenges. Several recent works considered extending sufficient dimension reduction approaches \citep{ying:22, zhan:21:1}, single index models \citep{mull:23:3}, and principal component regression \citep{han:23} to \f regression to address the high-dimensionality of the predictors. These methods, while capable of reducing predictor dimensionality, are confined to the construction of a universal kernel or strong assumptions similar to those for classical single-index and linear models.

To address these limitations, we harness deep learning \citep{lecu:15}, taking advantage of its capability of handling high-dimensional predictors without compromising model flexibility. Deep learning is known for its capacity to automatically extract features from complex data, but has not been developed so far for non-Euclidean responses. The proposed regression model utilizes manifold learning \citep{ghoj:23} for the observed metric space-valued responses, generating low-dimensional representations. Manifold learning is a nonlinear dimensionality reduction technique that has proven useful for dimensionality reduction of infinite-dimensional functional data \citep{mull:12:1}, where data are assumed to lie on a low-dimensional manifold that is unknown. Indeed, many data residing in metric spaces, including parametric families of probability distributions or stochastic block models for networks, are expected to have such a low-dimensional manifold structure; we exploit such a structure for the proposed statistical modeling. 

We deploy a deep neural network that takes the observed predictors as input and produces the low-dimensional manifold representation as output. One key innovation of the proposed approach is that we demonstrate that local \f regression can be utilized to map the low-dimensional manifold representation to the original metric space so that the final output resides in the metric space where the random objects are situated.  This map, which has not yet been considered in the manifold learning literature for metric space-valued data, is crucial for the implementation of a regression model in statistics, as it is mandatory for interpretability to predict the random objects residing in metric spaces, rather than their low-dimensional representations. This comprehensive framework makes it possible to handle high-dimensional predictors while concurrently accommodating the complexities intrinsic to metric space-valued responses.

Our key contributions are as follows. First, we utilize deep learning for \f regression, highlighting its efficacy as a powerful tool for the regression of metric space-valued responses. Second, our model can handle high-dimensional predictors for metric space-valued data that allow a manifold representation. Third, we investigate the convergence rate of deep neural networks in the presence of dependent sub-Gaussian noise accompanied by bias. This analysis is conducted under a well-established general composition assumption on the regression function \citep{schm:20}, which necessitates an extension of the existing theoretical framework due to the presence of dependence and bias. Although \citet{bos:23} introduced the convergence rate of deep neural networks with dependent and non-central response variables, their findings are limited to local dependence, while the dependence considered here does not rely on any specific form and is not restricted to local dependence. Fourth, we demonstrate how local \f regression can be used to map the low-dimensional representation back to the original metric space, for which we establish convergence in the presence of errors in predictors. Fifth, we examine the asymptotic properties of the proposed regression model, including its rate of convergence. Lastly, we showcase results from various simulation settings and two real-data applications, demonstrating the superior performance of the proposed regression model compared to existing approaches.

The structure of the paper is as follows. In Section~\ref{sec:pre}, we introduce notation and provide background on deep neural networks and local \f regression. The proposed regression model is introduced in Section~\ref{sec:met}, with the asymptotic convergence rate established in Section~\ref{sec:the}. Simulation results for distributional data are presented in Section~\ref{sec:sim}. The proposed framework is illustrated in Section~\ref{sec:app} using networks arising from the New York yellow taxi records. Finally, we conclude with a brief discussion in Section \ref{sec:dis}. Auxiliary results and proofs as well as additional simulations for networks and a second data application for age-at-death distributions of human mortality are provided in the Supplementary Material.

\section{Preliminaries}
\label{sec:pre}
\subsection{Notation} 
Let $\mathbb{R}$ represent the set of real numbers and $\mathbb{R}_{+}$ denote the set of all positive real numbers. The set of all positive natural numbers is denoted as $\mathbb{N}_{+}$. Throughout this work, plain letters refer to (random) scalars and bold letters refer to (random) vectors. For a vector $\bm{v}$, we define $\bm{v}^{\oplus0}=1, \bm{v}^{\oplus1}=\bm{v}, \bm{v}^{\oplus2}=\bm{v}\bm{v}^\T$. We let $\|\cdot\|$ denote a norm and set $\|\cdot\|=\|\cdot\|_2$ (the Euclidean norm for vectors) by default unless stated otherwise. Additionally, $\|\cdot\|_0$ denotes the number of nonzero entries of a matrix or vector, and $\|\cdot\|_{\infty}$ represents the sup-norm for vectors, matrices, or functions. The notation $a \lesssim b$ implies the existence of a positive constant $C$ such that $a \leq Cb$. Moreover, $a\asymp b$ if both $a\lesssim b$ and $b\lesssim a$. Finally, $a \wedge b=\min \{a, b\}, a \vee b=\max \{a, b\}$.

\subsection{Deep neural networks}
\label{subsec:dnn}
We provide an overview of deep neural networks (DNNs) as function approximations. The focus lies on the nonparametric regression model with $p$-dimensional covariates in the unit hypercube, represented by observations $(\bm{X}_i, U_i)\in [0,1]^p\times \mathbb{R}$. Consider the following regression model
\[U_i = g_0(\bm{X}_i) + \epsilon_i, \quad i=1,\ldots,n,\]
where $g_0$ is an unknown function and $\epsilon_i$ is independent and identically distributed (i.i.d.) noise with zero mean and finite variance.

A DNN with $L$ hidden layers and layer width $\mathbf{p}=(p_0, \ldots, p_L, p_{L+1})^\T$ is defined recursively, forming a composite function $g: \mathbb{R}^{p_0} \mapsto \mathbb{R}^{p_{L+1}}$,
\begin{align}\label{eq:DNN}
&g(\bm{x})=W_L g_L(\bm{x})+\bm{b}_L \nonumber,\nonumber\\
&g_L(\bm{x})=\sigma\{W_{L-1} g_{L-1}(\bm{x})+\bm{b}_{L-1}\},\ \ldots,\ g_1(\bm{x})=\sigma(W_0 \bm{x}+\bm{b}_0),
\end{align}
where $W_l\in \mathbb{R}^{p_{l+1}\times p_{l}}$ is a weight matrix and $\bm{b}_l\in\mathbb{R}^{p_{l+1}}$ is a shift vector for $l=0,\ldots,L$. The rectifier linear unit (ReLU) activation function $\sigma=\max\{0, x\}$ \citep{nair:10} operates component-wise, that is, $\sigma\{(x_1,\ldots,x_{p_l})^\T\}=(\sigma(x_1),\ldots,\sigma(x_{p_l}))^\T$. We focus on the ReLU activation function due to its widespread use, empirical success, and theoretical support \citep{lecu:15}. Observe that $p_0=p$ and $p_{L+1}=1$ in the case of nonparametric regression with $p$-dimensional covariates.

Given $L \in \mathbb{N}_{+}$, $\mathbf{p} \in \mathbb{N}_{+}^{L+2}$, $s \in \mathbb{N}_{+}$ and a constant $D>0$, in analogy to \citet{schm:20}, we consider a class of sparse neural networks
\begin{align}\label{eq:DNN Class}
\mathcal{G}(L, s, \mathbf{p}, D)=\{&g \text{ of the form \eqref{eq:DNN}: }\max_{l=0,\ldots,L} \{ \|W_l\|_{\infty},\|\bm{b}_l\|_{\infty}\} \leq 1, \nonumber\\
&\sum_{l=1}^L\|W_l\|_0+\|\bm{b}_l\|_0 \leq s,\|g\|_{\infty} \leq D\}.
\end{align}
We focus on DNNs with bounded weight matrices and shift vectors. This choice is motivated by empirical observations where the learned matrices and vectors tend to remain relatively small, especially when initialized with modest-sized matrices and vectors. Additionally, the consideration of sparse neural networks arises from the practical issue of overfitting deep feedforward neural networks with fully connected layers. This concern is addressed by pruning weights, reducing the total number of nonzero parameters, and sparsely connecting layers \citep{han:15, srin:17}.

To estimate the regression function $g_0$, the basic paradigm is to minimize the empirical risk
\[\hat{g} = \argmin_{g\in \mathcal{G}} \frac{1}{n}\sum_{i=1}^n\{U_i - g(\bm{X}_i)\}^2.\]
Certain restrictions on the regression function $g_{0}$ are necessary to study the asymptotic properties of the estimator $\hat{g}$. Consider the H\"{o}lder class of smooth functions, where the ball of $\beta$-H\"{o}lder functions with radius $M>0$ is
\begin{align*}
    \mathcal{H}_p^\beta(\mathbb{D}, M)=&\{g: \mathbb{D}\subset \mathbb{R}^p \mapsto \mathbb{R}: \nonumber\\
    &\sum_{\bm{\alpha}:|\bm{\alpha}|<\beta}\|\partial^{\bm{\alpha}} g\|_{\infty}+\sum_{\bm{\alpha}:|\bm{\alpha}|=\lfloor\beta\rfloor} \sup _{\bm{x}, \bm{y} \in \mathbb{D}, \bm{x} \neq \bm{y}} \frac{|\partial^{\bm{\alpha}} g(\bm{x})-\partial^{\bm{\alpha}} g(\bm{y})|}{\|\bm{x}-\bm{y}\|_{\infty}^{\beta-\lfloor\beta\rfloor}} \leq M\}.
\end{align*}
Here $\lfloor\beta\rfloor$ denotes the largest integer strictly smaller than $\beta$, $\partial^{\bm{\alpha}}=\partial^{\alpha_1} \cdots \partial^{\alpha_p}$ with $\bm{\alpha}=(\alpha_1, \ldots, \alpha_p)^\T$ and $|\bm{\alpha}|=\sum_{k=1}^p \alpha_k$. Let $q \in \mathbb{N}$, $\bm{\beta}=(\beta_0, \ldots, \beta_q)^\T \in \mathbb{R}_{+}^{q+1}$ and $\mathbf{d}=$ $(d_0, \ldots, d_{q+1})^\T \in \mathbb{N}_{+}^{q+2}, \,\tilde{\mathbf{d}}=(\tilde{d}_0, \ldots, \tilde{d}_q)^\T \in \mathbb{N}_{+}^{q+1}$ with $\tilde{d}_j \leq d_j,\, j=0, \ldots, q$. We assume that the regression function $g_0$ belongs to a composite smoothness function class:
\begin{align}\label{eq:holder}
\mathcal{H}(q, \bm{\beta}, \mathbf{d}, \tilde{\mathbf{d}}, M)= & \{g=g_q \circ \cdots \circ g_0: g_j=(g_{j 1}, \ldots, g_{j d_{j+1}})^\T \text { and } \nonumber \\
& g_{jk} \in \mathcal{H}_{\tilde{d}_j}^{\beta_j}([a_j, b_j]^{\tilde{d}_j}, M), \text { for some }|a_j|,|b_j| \leq M\} .
\end{align}
Functions in this class are characterized by two kinds of dimensions, $\mathbf{d}$ and $\tilde{\mathbf{d}}$, where the latter represents the intrinsic dimension of the function. For example, if
\begin{align}\label{eq:composite}
    g(\bm{x})= g_{31}(&g_{21}[g_{11}\{g_{01}(x_1) + g_{02}(x_2)\} + g_{12}\{g_{03}(x_3) + g_{04}(x_4)\}] \nonumber\\
    + &g_{22}[g_{13}\{g_{05}(x_5)+ g_{06}(x_6)\} + g_{14}\{g_{07}(x_7) + g_{08}(x_8)\}]),\quad\bm{x}\in[0,1]^{8}
\end{align}
and $g_{i j}$ are twice continuously differentiable, then smoothness $\bm{\beta}=(2,2,2,2)^\T$, dimensions $\bm{d}=$ $(8,8,4,2,1)^\T$ and $\tilde{\bm{d}}=(1,1,1,1)^\T$. The composite smoothness class $\mathcal{H}(q, \bm{\beta}, \mathbf{d}, \tilde{\mathbf{d}}, M)$ contains a rich set of classical smoothness classes and has been widely adopted \citep{baue:19, schm:20, kohl:21}.

Writing 
\begin{equation}
\label{eq:kappan}
    \kappa_n=\max _{j=0, \ldots, q} n^{-\tilde{\beta}_j /(2 \tilde{\beta}_j+\tilde{d}_j)}\text{ with }\tilde{\beta}_j=\beta_j \prod_{k=j+1}^q(\beta_k \wedge 1),
\end{equation} 
we require the following common assumption on the structure of the DNN model \citep{schm:20, zhon:22}.
\begin{enumerate}[label=(C\arabic*)]
    \item $D\geq \max\{M,1\}$, $L=O(\log n)$, $s=O(n \kappa_n^2 \log n)$ and $n \kappa_n^2 \lesssim \min_{l=1, \ldots, L}p_l \leq \max_{l=1, \ldots, L}p_l \lesssim$ $n$.\label{itm:d1}
\end{enumerate}
Assumption \ref{itm:d1} characterizes the flexibility of the neural network family $\mathcal{G}(L, s, \bm{p}, D)$ as per \eqref{eq:DNN Class}. While more flexible neural networks (with increased depth and width) are capable of achieving smaller approximation errors \citep{anth:99, schm:20}, they often introduce larger estimation errors. Thus, Assumption \ref{itm:d1} provides a balance between approximation and estimation errors.

\subsection{Local \f regression}
Let $(\Omega, d)$ be a totally bounded metric space with distance $d: \Omega\times\Omega\mapsto [0,\infty)$ and $\mathcal{T}$ be a closed interval in $\mathbb{R}$. Consider a random pair $(\bm{Z}^0, Y)$ with a joint distribution $F$ supported on the product space $\mathcal{T}^r\times\Omega$. We denote the marginal distributions of $\bm{Z}^0$ and $Y$ as $F_{\bm{Z}^0}$ and $F_Y$, respectively, and the conditional distribution of $Y$ given $\bm{Z}^0$ as $F_{Y|\bm{Z}^0}$. The \f mean and \f variance of random objects in metric spaces \citep{frec:48}, as generalizations of usual notions of mean and variance, are defined as
\[y_{\oplus}=\argmin_{y\in\Omega}E\{d^2(Y, y)\},\quad V_{\oplus}=E\{d^2(Y, y_{\oplus})\},\]
where the existence and uniqueness of the minimizer depend on the structural properties of the underlying metric space and are guaranteed for Hadamard spaces.

The conditional \f mean of $Y$ given $\bm{Z}^0=\bm{z}$, corresponding to the regression function, is
\begin{equation}
\label{eq:cfm}
v(\bm{z})=\argmin_{y\in\Omega} Q(y, \bm{z}),\quad Q(\cdot, \bm{z})=E\{d^2(Y, \cdot)|\bm{Z}^0=\bm{z}\}.
\end{equation}
The lack of an algebraic structure in general metric spaces poses challenges in modeling and estimating the regression function $v(\bm{z})$. To address this challenge, previous work \citep{mull:19:6} suggested leveraging the algebraic structure within the predictor space $\mathcal{T}^r$. Let $K$ be an $r$-dimensional kernel corresponding to a symmetric $r$-dimensional probability density function and write $K_h(\bm{z})=h^{-r}K(H^{-1}\bm{z})$ with $H=hI_r$ for a suitably chosen bandwidth $h$. Extending local linear regression to metric space-valued responses, local \f regression (LFR) \citep{mull:19:6} models conditional \f means $v(\bm{z})$ as weighted \f means,
\begin{equation}
\label{eq:lfr}
v_h(\bm{z})=\argmin_{y\in\Omega}Q_h(y, \bm{z}),\quad Q_h(\cdot, \bm{z})=E\{w(\bm{Z}^0, \bm{z}, h)d^2(Y, \cdot)\}.
\end{equation}
Here 
\[w(\bm{Z}^0, \bm{z}, h)=\frac{1}{\mu_0-\mu_1^\T\mu_2^{-1}\mu_1}K_h(\bm{Z}^0-\bm{z})\{1-\mu_1^\T\mu_2^{-1}(\bm{Z}^0-\bm{z})\},\]
where $\mu_j=E\{K_h(\bm{Z}^0-\bm{z})(\bm{Z}^0-\bm{z})^{\oplus j}\}$ for $j=0, 1, 2$. For an i.i.d. sample $\{(\bm{Z}^0_i, Y_i)\}_{i=1}^n$, the corresponding empirical version is
\begin{equation}
    \label{eq:lfrtilde}
    \tilde{v}_h(\bm{z})=\argmin_{y\in\Omega}\tilde{Q}_h(y, \bm{z}),\quad \tilde{Q}_h(\cdot, \bm{z})=\frac{1}{n}\sum_{i=1}^n\tilde{w}(\bm{Z}^0_i, \bm{z}, h)d^2(Y_i, \cdot).
\end{equation}
Here 
\begin{equation}
    \label{eq:wtilde}
    \tilde{w}(\bm{Z}^0_i, \bm{z}, h)=\frac{1}{\tilde{\mu}_0-\tilde{\mu}_1^\T\tilde{\mu}_2^{-1}\tilde{\mu}_1}K_h(\bm{Z}^0_i-\bm{z})\{1-\tilde{\mu}_1^\T\tilde{\mu}_2^{-1}(\bm{Z}^0_i-\bm{z})\},
\end{equation}
where $\tilde{\mu}_j=n^{-1}\sum_{i=1}^nK_h(\bm{Z}^0_i-\bm{z})(\bm{Z}^0_i-\bm{z})^{\oplus j}$ for $j=0, 1, 2$.

\section{Methodology}
\label{sec:met}
Consider a totally bounded metric space $(\Omega, d)$ with metric $d:\Omega\times\Omega\mapsto [0,\infty)$. Let $(\bm{X}, Y)$ be a random pair in $\mathbb{R}^p\times\Omega$, where $\bm{X}$, without loss of generality, is assumed to lie in the unit hypercube $[0, 1]^p$. The conditional \f mean $E_\oplus(Y|\bm{X})=\argmin_{y\in\Omega}E\{d^2(Y, y)|\bm{X}\}$, representing the regression function, is considered to lie in a subset $\mathcal{M}\subset \Omega$, where $\mathcal{M}$ is a manifold isomorphic to a subspace of $\mathbb{R}^r$ with $r$ denoting its intrinsic dimension. Note that the metric space-valued response $Y\in\Omega$, subject to random noise, is not constrained to lie on the manifold. Suppose we observe $n$ independent realizations $\{(\bm{X}_i, Y_i)\}_{i=1}^n$ of $(\bm{X}, Y)$. 

There exists a bijective representation map of the manifold $\mathcal{M}$,
\begin{equation}
    \label{eq:psi}
    \bm{\psi}:\mathcal{M}\mapsto\mathbb{R}^r.
\end{equation}
Algorithms such as ISOMAP \citep{tene:00}, t-SNE \citep{van:08}, UMAP \citep{mcin:18}, Laplacian eigenmaps \citep{belk:03}, or diffusion maps \citep{coif:06} are widely used to estimate the representation map $\bm{\psi}$. In this work, we select the ISOMAP algorithm due to its reliable performance in our simulations; see Section S.11 of the Supplementary Material for detailed comparisons with alternative algorithms. A description of the ISOMAP algorithm is provided in Section S.2 of the Supplementary Material, and additional details can be found in \citet{tene:00}. Manifold representations have proved useful across many scenarios, including latent parametric families of probability distributions or stochastic block models for networks, where a low-dimensional manifold may closely align with parameters of the underlying model \citep{mull:22:12}.

We propose Deep \f regression (DFR) as a unified framework to model the relationship between the metric space-valued response $Y$ and a possibly high-dimensional multivariate predictor $\bm{X}$, as summarized in Figure~\ref{fig:diag}. Here the regression function
\begin{equation}
    \label{eq:m}
    m=v\circ \bm{g}_0:[0, 1]^p\mapsto\Omega
\end{equation}
exhibits a composite structure, 
where $\bm{g}_0=(g_{01}, \ldots, g_{0r})^\T$ and $v$ correspond to a nonparametric regression step using DNNs and local \f regression, respectively. When the intrinsic dimension is $r>1$, we fit one DNN for each coordinate.
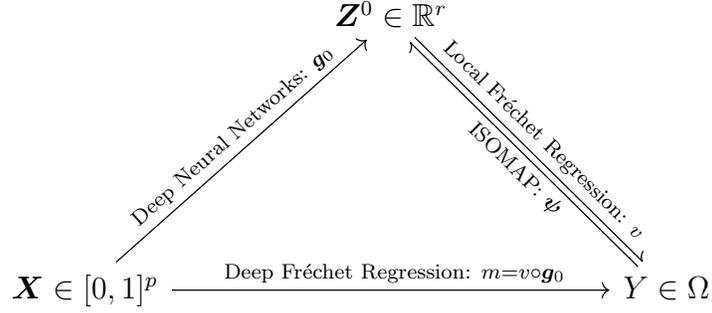
\begin{figure}[t]
\single
\centering
    \begin{tikzcd}
	&& {\bm{Z}^0\in\mathbb{R}^r} \\
	\\
	\\
        \\
	{\bm{X}\in[0,1]^p} &&&& {E_\oplus(Y|\bm{X})}
	\arrow["{ \text{ Local \f Regression: } v}", marking, allow upside down, rightharpoonup, shift left=0.25ex, from=1-3, to=5-5, sloped, pos=0.5, overlay]
	\arrow["{\text{ Deep \f Regression: }m = v \circ \bm{g}_0}", marking, allow upside down, no tail, from=5-1, to=5-5, sloped, pos=0.5, overlay]
	\arrow["{\text{ ISOMAP: }\bm{\psi}}", marking, allow upside down, shift right=0.25ex, leftharpoondown, swap, from=1-3, to=5-5, sloped, pos=0.5, overlay]
	\arrow["{\text{ Deep Neural Networks: }\bm{g}_0}", marking, allow upside down, no tail, from=5-1, to=1-3, sloped, pos=0.5, overlay]
    \end{tikzcd}
    \caption{Schematic diagram for the deep \f regression $m= v \circ \bm{g}_0$, where $\bm{Z}^0=\bm{\psi}\{E_\oplus(Y|\bm{X})\}$ denotes the low-dimensional representation of the regression function $E_\oplus(Y|\bm{X})$, $v$ the local \f regression and $\bm{g}_0$ deep neural networks.}
    \label{fig:diag}
\end{figure}

The proposed framework includes the ISOMAP representation $\bm{\psi}$,  which is unknown and must be estimated from the data, yielding $\hat{\bm{\psi}}$. We postulate that the manifold can be well identified at the sample points through the ISOMAP algorithm.
\begin{enumerate}[label=(C\arabic*)]
    \setcounter{enumi}{1}
    \item \label{itm:iso1} Consider $\bm{Z}^0_i=(Z_{i1}^0, \ldots, Z_{ir}^0)^\T=\bm{\psi}\{E_\oplus(Y_i | \bm{X}_i)\}$ and $\bm{Z}_i=(Z_{i1}, \ldots, Z_{ir})^\T=\hat{\bm{\psi}}(Y_i)$. For each $j=1, \ldots, r$, there exists a function $\pi_j:[0,1]^{p\times n}\mapsto\mathbb{R}$ such that
    \begin{equation*}
        Z_{ij}-Z_{ij}^0=\pi_j(\bm{X}_1,\ldots,\bm{X}_n)(\epsilon_{ij}+u_{nj}),\quad i=1,\ldots,n,
    \end{equation*}
    where the bias $u_{nj}\to 0$ as $n\to\infty$, and $\{\epsilon_{ij}\}_{i=1}^n$ are i.i.d. (mean zero) sub-Gaussian random variables with parameter 1 and independent of $\{\bm{X}_i\}_{i=1}^n$, i.e., $P(|\epsilon_{ij}|\geq t)\leq 2\exp\{-t^2\}$. Furthermore, there exists a constant $C_{\pi_j}$ such that 
    \begin{equation*}
        C_{\pi_j} = \sup_{\{\bm{x}_1,\ldots,\bm{x}_n\}}|\pi_j (\bm{x}_1,\ldots,\bm{x}_n) | < \infty.
    \end{equation*}
\end{enumerate}
Assumption \ref{itm:iso1} is designed to be flexible, as it accounts for both the random error in $Y_i$ and the estimation error induced by the ISOMAP algorithm. Specifically, it allows for dependence between the estimated low-dimensional representations $Z_{ij}$ and includes potential bias $u_{nj}$ introduced by ISOMAP. The function $\pi_j$ quantifies the dependence among $Z_{ij}$ across different $i$, with an upper bound $C_{\pi_j}$ that is specific to each coordinate. This assumption is consistent with similar considerations in \citet{kuri:22}, where deep neural networks are employed for analyzing time series data. The sub-Gaussian noise $\epsilon_{ij}$ reflects the randomness arising from both the inherent variability in $Y_i$ and the approximation error in the manifold estimation process.

We employ DNNs that take the observed predictors $\bm{X}_i$ as input and produce the estimated low-dimensional representations $\bm{Z}_i$ as output. For each $j=1, \ldots, r$, by Assumption \ref{itm:iso1}, using  $Z_{ij}^0=g_{0j}(\bm{X}_i)$ we may model the relationship between $Z_{ij}$ and $\bm{X}_i$ as
\begin{equation}\label{eq:our DNN}
    Z_{ij} = g_{0j}(\bm{X}_i) + \pi_j(\bm{X}_1,\ldots,\bm{X}_n)(\epsilon_{ij}+u_{nj}).
\end{equation}
Unlike the prevalent focus on independent Gaussian/sub-Gaussian noise without bias in the existing literature on DNNs \citep{baue:19, schm:20, kohl:21}, the noise $\pi_j(\bm{X}_1,\ldots,\bm{X}_n)(\epsilon_{ij}+u_{nj})$ considered in \eqref{eq:our DNN} exhibits dependence across different $i$ and is accompanied by a vanishing bias $u_{nj}$, characterized by the randomness of a sub-Gaussian variable. This generalization necessitates a nontrivial expansion of the current theoretical framework for DNNs \citep{schm:20}, which is provided below in Theorem \ref{thm:DNN}. The estimation of $g_{0j}$ is performed by minimizing the empirical risk 
\begin{equation}\label{eq:DNN Ln}
    \hat{g}_j = \argmin_{g\in \mathcal{G}} \frac{1}{n}\sum_{i=1}^n \{Z_{ij} - g(\bm{X}_i)\}^2.
\end{equation}

\begin{algorithm}[t]
\single
	\SetAlgoLined
	\KwIn{data $\{(\bm{X}_i, Y_i)\}_{i=1}^n$, and a new predictor level $\bm{X}$.}
	\KwOut{prediction $\hat{m}(\bm{X})$.}
        ISOMAP: 
        $$\{\hat{\bm{\psi}}(Y_i)\}_{i=1}^n\longleftarrow\text{estimated low-dimensional representations of }\{Y_i\}_{i=1}^n,$$
        and denote $\hat{\bm{\psi}}(Y_i)$ by $\bm{Z}_i=(Z_{i1}, \ldots, Z_{ir})^\T$\;
        DNNs: for $j=1, \ldots, r$,
        $$\hat{g}_{j}\longleftarrow \argmin_{g\in \mathcal{G}} \frac{1}{n}\sum_{i=1}^n \{Z_{ij} - g(\bm{X}_i)\}^2,$$
        where $\hat{g}_j$ is the $j$th deep neural network trained using sample $\{(\bm{X}_i, Z_{ij})\}_{i=1}^n$\;
        $\{\hat{\bm{Z}}_i\}^n_{i=1}\longleftarrow \{(\hat{g}_1(\bm{X}_i),\ldots,\hat{g}_r(\bm{X}_i))^\T\}_{i=1}^n$\;
        LFR: 
        $$\hat{m}(\bm{X})=\hat{v}_h\circ \hat{\bm{g}}(\bm{X})=\hat{v}_h(\hat{\bm{Z}})\longleftarrow\argmin_{y\in\Omega}\frac{1}{n}\sum_{i=1}^n\hat{w}(\hat{\bm{Z}}_i, \hat{\bm{Z}}, h)d^2(Y_i, y),$$
        where $\hat{\bm{Z}}=\hat{\bm{g}}(\bm{X})=(\hat{g}_1(\bm{X}),\ldots,\hat{g}_r(\bm{X}))^\T$.
	\caption{Deep \f Regression}
	\label{alg:dfr}
\end{algorithm}

The proposed framework also deploys local \f regression using the  fitted value of DNNs $\hat{\bm{Z}}_i = \hat{\bm{g}}(\bm{X}_i)=\{\hat{g}_1(\bm{X}_i),\ldots,\hat{g}_r(\bm{X}_i)\}^\T$ in lieu of the unobservable $\bm{Z}^0_i$ as predictors. This leads to an errors-in-variables version of  local \f regression,
\begin{equation}
    \label{eq:lfrhat}
    \hat{v}_h(\bm{z})=\argmin_{y\in\Omega}\hat{Q}_h(y, \bm{z}),\quad \hat{Q}_h(\cdot, \bm{z})=\frac{1}{n}\sum_{i=1}^n\hat{w}(\hat{\bm{Z}}_i, \bm{z}, h)d^2(Y_i, \cdot),
\end{equation}
where 
\begin{equation}
    \label{eq:what}
    \hat{w}(\hat{\bm{Z}}_i, \bm{z}, h)=\frac{1}{\hat{\mu}_0-\hat{\mu}_1^\T\hat{\mu}_2^{-1}\hat{\mu}_1}K_h(\hat{\bm{Z}}_i-\bm{z})\{1-\hat{\mu}_1^\T\hat{\mu}_2^{-1}(\hat{\bm{Z}}_i-\bm{z})\},
\end{equation}
with $\hat{\mu}_j=n^{-1}\sum_{i=1}^nK_h(\hat{\bm{Z}}_i-\bm{z})(\hat{\bm{Z}}_i-\bm{z})^{\oplus j}$ for $j=0, 1, 2$. The implementation of the minimization of $\hat{Q}_h(\cdot, \bm{z})$ for various commonly encountered metric spaces is detailed in Section \ref{supp:lfr:implement} of the Supplementary Material.

Consequently, the proposed estimate of the regression function $m$ as per \eqref{eq:m} that relates the predictor $\bm{X}$ to the metric space-valued response $Y$ is 
\begin{equation}
    \label{eq:mhat}
    \hat{m}=\hat{v}_h\circ\hat{\bm{g}}
\end{equation}
with $\hat{\bm{g}}=(\hat{g}_1, \ldots, \hat{g}_r)^\T$ and $\hat{v}_h$ defined in \eqref{eq:DNN Ln} and \eqref{eq:lfrhat}. This comprehensive framework is designed to effectively handle potentially high-dimensional predictors while simultaneously addressing the inherent complexities associated with metric space-valued responses. Further details are provided in Algorithm \ref{alg:dfr}.

\section{Asymptotic Properties}
\label{sec:the}
We first derive the convergence rate of DNN estimates \eqref{eq:DNN Ln}, extending the existing theoretical framework for DNNs to accommodate dependent noise and bias. We then delve into the convergence rate of the local \f regression estimate \eqref{eq:lfrhat}, considering multivariate predictors contaminated by errors induced by DNNs. Finally, the convergence rate of the DFR estimate as per \eqref{eq:mhat} is established by combining the convergence rates of DNNs and local \f regression. For DNN, we require \ref{itm:d1} and the following common assumption \citep{schm:20, kohl:21, zhon:22}.
\begin{enumerate}[label=(C\arabic*)]
    \setcounter{enumi}{2}
    \item The target function $g_{0j}$ belongs to the H\"{o}lder class $\mathcal{H} (q,\bm{\beta},\mathbf{d},\tilde{\mathbf{d}},M)$ defined in \eqref{eq:holder} for $j = 1,\ldots, r$.\label{itm:d2}
\end{enumerate}
The H\"{o}lder class in Assumption \ref{itm:d2} encompasses a rich class of smooth functions \citep{schm:20, kohl:21}; see Section~\ref{subsec:dnn} for more details. The following result formalizes the convergence rate of DNN estimates \eqref{eq:DNN Ln}.

\begin{thm}
\label{thm:DNN}
    Consider the nonparametric regression model \eqref{eq:our DNN}. Under \ref{itm:d1}--\ref{itm:d2}, for any $i=1,\ldots,n$ and $j = 1,\ldots,r$, there exists an estimator $\hat{g}_j$ in \eqref{eq:DNN Ln} such that
    $$E[\{\hat{g}_j(\bm{X}_i) - g_{0j}(\bm{X}_i)\}^2] = O(\kappa_n^2 \log^3 n+|u_{nj}|)$$
    and
    $$E[\{\hat{g}_j(\bm{X}) - g_{0j}(\bm{X})\}^2] = O(\kappa_n^2 \log^3 n+|u_{nj}|),$$
    where  $\bm{X}$ is a new realization of the predictor that is independent of the sample $\{(\bm{X}_i, Y_i)\}_{i=1}^n$ and $\kappa_n$ is defined in \eqref{eq:kappan}. It immediately follows that
    $$E(\|\hat{\bm{Z}}_i - \bm{Z}_i^0\|^2)= O(\kappa_n^2 \log^3 n+u^2_{n})$$
    and 
    $$E(\|\hat{\bm{Z}} - \bm{Z}^0\|^2)= O(\kappa_n^2 \log^3 n+u^2_{n}),$$
    where $\hat{\bm{Z}}=(\hat{g}_1(\bm{X}),\ldots,\hat{g}_r(\bm{X}))^\T$, $\bm{Z}^0 = (g_{01}(\bm{X}),\ldots,g_{0r}(\bm{X}))^\T$, $u_n^2=\max_{j=1, \ldots, r}|u_{nj}|$ and $g_{0j}$ are the components of $\bm{g}_0$ as defined in \eqref{eq:m}.
\end{thm}

The proof of this result is provided in the Supplementary Material. It involves decomposing both in-sample and out-of-sample errors into approximation error, ISOMAP-induced error, neural network complexity and the difference between the empirical risk of any estimator and the global minimum over the neural network space. Bounds on these components are used to control the total error, as formalized in Lemma \ref{lem:lemma4 extension} in the Supplementary Material, which extends Lemma 4 of \citet{schm:20} to accommodate our setting with dependent sub-Gaussian noise. This leads to Theorem \ref{thm:DNN}, following similar arguments as in Theorem 1 and Corollary 1 of \citet{schm:20}, but adjusted for the noise structure and complexity of the proposed DFR model. In contrast to the prevailing emphasis on DNNs with independent Gaussian/sub-Gaussian noise, Theorem \ref{thm:DNN} establishes the convergence rate of DNNs in the presence of dependent sub-Gaussian noise, including an asymptotically vanishing bias. We expect this extension to also be useful for other work in deep learning.

Note that the convergence rates in Theorem \ref{thm:DNN} depend on the bias term $u_{nj}$, smoothness $\bm{\beta}$, and intrinsic dimension $\tilde{\mathbf{d}}$ of functions $g_{0j}$ as per \eqref{eq:holder}, rather than the predictor dimension $p$. Consequently, the proposed method circumvents the curse of dimensionality, maintaining faster convergence rates particularly when the intrinsic dimension $\tilde{\mathbf{d}}$ remains relatively low. For instance, disregarding the bias term momentarily, consider the true function $g_{0j}$ structured in \eqref{eq:composite}. For such a target function, traditional smoothing methods are fraught with a slow convergence rate of order $n^{-2/(p+4)} = n^{-1/6}$. In contrast, the proposed DNN approach attains a convergence rate of order $n^{-2/5}\log^{3/2} n$, reaching the optimal one-dimensional nonparametric regression rate up to a polylogarithmic factor. This improved convergence rate effectively mitigates the curse of dimensionality inherent in nonparametric regression.

Next, we derive the convergence rate for errors-in-variables multivariate local \f regression, for which we require the following assumptions.
\begin{enumerate}[label=(C\arabic*)]
    \setcounter{enumi}{3}
    \item \label{itm:l1} The minimizers $v(\bm{z})$, $v_h(\bm{z})$, $\tilde{v}_h(\bm{z})$ and $\hat{v}_h(\bm{z})$ for local \f regression as per \eqref{eq:cfm}, \eqref{eq:lfr}, \eqref{eq:lfrtilde}, and \eqref{eq:lfrhat} exist and are unique, the last two almost surely. Additionally, for any $\varepsilon>0$,
    \begin{enumerate}[label = (\roman*)]
        \item $\inf_{d\{v(\bm{z}), y\}>\varepsilon}[Q(y, \bm{z})-Q\{v(\bm{z}), \bm{z}\}]>0,$
        \item $\liminf_{h\to0}\inf_{d\{v_h(\bm{z}), y\}>\varepsilon}[Q_h(y, \bm{z})-Q_h\{v_h(\bm{z}), \bm{z}\}]>0,$
        \item $P(\inf_{d\{\tilde{v}_h(\bm{z}), y\}>\varepsilon}[\tilde{Q}_h(y, \bm{z})-\tilde{Q}_h\{\tilde{v}_h(\bm{z}), \bm{z}\}]>0)\to1,$
        \item $P(\inf_{d\{\hat{v}_h(\bm{z}), y\}>\varepsilon}[\hat{Q}_h(y, \bm{z})-\hat{Q}_h\{\hat{v}_h(\bm{z}), \bm{z}\}]>0)\to1.$
    \end{enumerate}
    \item \label{itm:l2} Let $B_\delta(y)\subset\Omega$ be the ball of radius $\delta$ centered at $y$ and $N\{\varepsilon, B_\delta(y), d\}$ be its covering number using balls of size $\varepsilon$ (see Definition \ref{def:covering} in Section \ref{supp:back} of the Supplementary Material). Then for any $y\in\Omega$, 
    \begin{equation*}
        \int_0^1[1+\log N\{\delta\varepsilon, B_{\delta}(y), d\}]^{1/2}d\varepsilon=O(1)\quad\text{as }\delta\to0.
    \end{equation*}
    \item \label{itm:l3} There exists $\eta_1, \eta_2, \eta_3>0$, $C_1, C_2, C_3>0$ and $\gamma_1, \gamma_2, \gamma_3>1$ such that
    \begin{enumerate}[label = (\roman*)]
    \item $\inf_{d\{v(\bm{z}), y\}<\eta_1}[Q(y, \bm{z})-Q\{v(\bm{z}), \bm{z}\}-C_1d\{v(\bm{z}), y\}^{\gamma_1}]\geq0,$
    \item $\liminf_{h\to0}\inf_{d\{v_h(\bm{z}), y\}<\eta_2}[Q_h(y, \bm{z})-Q_h\{v_h(\bm{z}), \bm{z}\}-C_2d\{v_h(\bm{z}), y\}^{\gamma_2}]\geq0,$
    \item $P(\inf_{d\{\tilde{v}_h(\bm{z}), y\}<\eta_3}[\tilde{Q}_h(y, \bm{z})-\tilde{Q}_h\{\tilde{v}_h(\bm{z}), \bm{z}\}-C_3d\{\tilde{v}_h(\bm{z}), y\}^{\gamma_3}]\geq0)\to1$.
    \end{enumerate}
\end{enumerate}
Assumption \ref{itm:l1} is a standard requirement to ensure the consistency of M-estimators, as outlined in \citet{well:23}. We note that the existence and uniqueness of the minimizers in \ref{itm:l1} is guaranteed for the case of Hadamard spaces, the curvature of which is bounded above by 0 \citep{stur:03}. Assumption \ref{itm:l2} entails conditions on the metric entropy of $\Omega$ to control its size. In Proposition \ref{prop:lfrpc}, it constrains the bracketing integral used to derive the tail bound, following Theorem 2.14.16 in \citet{well:23}. The curvature assumption in \ref{itm:l3} regulates the behavior of $(Q_h - Q)$, $(\tilde{Q}_h - Q_h)$, and $(\hat{Q}_h - \tilde{Q}_h)$ near their minima, which makes it possible to derive convergence rates. These assumptions stem from empirical process theory and are frequently employed in the literature concerning statistical analysis for metric space-valued data \citep{mull:19:6, scho:22, mull:22:8, zhan:21:1}. 

We now provide some statistically relevant examples of metric spaces, for which we demonstrate in Section \ref{supp:prop} of the Supplementary Material that these metric spaces satisfy  Assumptions \ref{itm:l1}--\ref{itm:l3}; see Propositions \ref{prop:wass}, \ref{prop:mat} and \ref{prop:cov}.

\begin{exm}
\label{exm:was}
    Let $\Omega$ be the space of probability distributions on a closed interval of $\mathbb{R}$ with finite second moments, equipped with the Wasserstein metric $d_{\mathcal{W}}$ where
    \[d_{\mathcal{W}}^2(\mu, \nu)=\int_0^1\{F_{\mu}^{-1}(p)-F_{\nu}^{-1}(p)\}^2dp\]
    for any two probability distributions $\mu, \nu\in\Omega$. The Wasserstein space $(\Omega, d_{\mathcal{W}})$ satisfies Assumptions \ref{itm:l1}--\ref{itm:l3} with $\gamma_1=\gamma_2=\gamma_3=2$.
\end{exm}
\begin{exm}
\label{exm:mat}
    	Let $\Omega$ be the space of graph Laplacians of undirected weighted networks with a fixed number of nodes $m$ and bounded edge weights equipped with the Frobenius metric $d_F$. The space $(\Omega, d_F)$ satisfies Assumptions \ref{itm:l1}--\ref{itm:l3} with $\gamma_1=\gamma_2=\gamma_3=2$.
\end{exm}
\begin{exm}
\label{exm:cov}
    Let $\Omega$ be the space of $m$-dimensional covariance matrices with bounded diagonal entries or correlation matrices equipped with the Frobenius metric $d_F$. The space $(\Omega, d_F)$ satisfies Assumptions \ref{itm:l1}--\ref{itm:l3} with $\gamma_1=\gamma_2=\gamma_3=2$.
\end{exm}

We next derive the convergence rate of multivariate local \f regression, extending previous results for local \f regression \eqref{eq:lfrtilde} with one-dimensional predictors. The proof relies on M-estimation techniques arising from empirical process theory. The kernel and distributional assumptions \ref{itm:k1} and \ref{itm:p1} listed in the Appendix are standard for local regression. In the following, $r$ is as in \eqref{eq:psi} the intrinsic dimension of the manifold representation \eqref{eq:psi}.

\begin{prop}
\label{prop:lfrpc}
If \ref{itm:k1}--\ref{itm:p1}, \ref{itm:l1} (i)--(ii), \ref{itm:l2} and \ref{itm:l3} (i)--(ii) hold, then
\begin{equation}\label{lfrpc:bias}
    d\{v_h(\bm{z}), v(\bm{z})\}=O(h^{2/(\gamma_1-1)})
\end{equation}
as $h\to0$. If furthermore $nh^r\to\infty$, then
\begin{equation}\label{lfrpc:variance}
    d\{\tilde{v}_h(\bm{z}), v_h(\bm{z})\}=O_p\{(nh^r)^{-1/\{2(\gamma_2-1)\}}\}.
\end{equation}
\end{prop}

In general, the convergence rates are determined by the local geometry near the minima as quantified in \ref{itm:l3}. For metric spaces in Examples \ref{exm:was}, \ref{exm:mat} and \ref{exm:cov}, $\gamma_1=\gamma_2=2$ and convergence rates of the bias and stochastic deviation are $d\{v_h(\bm{z}), v(\bm{z})\}=O(h^2)$ and $d\{\tilde{v}_h(\bm{z}), v_h(\bm{z})\}=O_p\{(nh^r)^{-1/2}\}$, respectively. With $h=n^{-1/(4+r)}$, the multivariate local \f regression estimate achieves the rate $d\{\tilde{v}_h(\bm{z}), v(\bm{z})\}=O_p(n^{-2/(4+r)})$, corresponding to the well-known optimal rate for standard local linear regression.

The low-dimensional representations $\bm{Z}^0_i$ in \eqref{eq:lfrtilde} are unobservable, necessitating the use of \eqref{eq:lfrhat} in practice, where the estimated representations $\hat{\bm{Z}}_i$ obtained through DNNs are substituted. Consequently, errors are introduced in the predictors of local \f regression, which then requires addressing the errors-in-variables problem. The following result characterizes the impact of these errors on the estimation.

\begin{prop}
\label{prop:eiv}
Suppose \ref{itm:k1}, \ref{itm:l1}, \ref{itm:l3} hold. If $\|\hat{\bm{Z}}_i-\bm{Z}^0_i\|=O_p(\zeta_n)$ for $i=1, \ldots, n$, $\|\hat{\bm{Z}}-\bm{Z}^0\|=O_p(\zeta_n)$, $nh^r\to\infty$, and $h^{-r-2}\zeta_n\to0$, then it holds for $\tilde{v}_h(\cdot)$ and $\hat{v}_h(\cdot)$ as per \eqref{eq:lfrtilde} and \eqref{eq:lfrhat} that
\[d\{\hat{v}_h(\hat{\bm{Z}}), \tilde{v}_h(\bm{Z}^0)\}=O_p\{(h^{-r-1}\zeta_n)^{1/(\gamma_3-1)}\}.\]
\end{prop}

For metric spaces in Examples \ref{exm:was}, \ref{exm:mat} and \ref{exm:cov}, $\gamma_3=2$ and the convergence rate provided in Proposition \ref{prop:eiv} reduces to $O_p\{(h^{-r-1}\zeta_n)$. Combining Theorems \ref{thm:DNN} and Propositions \ref{prop:lfrpc} and \ref{prop:eiv}, we have the following theorem.
\begin{thm}
\label{thm:pc}
If \ref{itm:d1}--\ref{itm:l3}, \ref{itm:k1}--\ref{itm:p1} hold and furthermore $nh^r\to\infty$, $h^{-r-2}\zeta_n\to0$, one has
\begin{equation}
    d\{\hat{m}(\bm{X}), m(\bm{X})\}=O_p\{h^{2/(\gamma_1-1)}+(nh^r)^{-1/\{2(\gamma_2-1)\}}+(h^{-r-1}\zeta_n)^{1/(\gamma_3-1)}\},
\end{equation}
where $\bm{X}$ is a new predictor independent of the sample $\{(\bm{X}_i, Y_i)\}_{i=1}^n$ and $\zeta_n=\kappa_n\log^{3/2}n+u_n$.
\end{thm}

The initial two terms $h^{2/(\gamma_1-1)}, (nh^r)^{-1/\{2(\gamma_2-1)\}}$ in Theorem \ref{thm:pc} give  the convergence rate of multivariate local \f regression, as outlined in Proposition \ref{prop:lfrpc}, representing the optimal rate when $\bm{Z}_i^0$, $\bm{\psi}$, and $\bm{g}_0$ are known. The final term, $(h^{-r-1}\zeta_n)^{1/(\gamma_3-1)}$, originates from Theorem \ref{thm:DNN} and Proposition \ref{prop:eiv}, reflecting errors introduced during the estimation of low-dimensional representations of metric space-valued responses, DNNs, and the subsequent errors-in-variables problem.

If one assumes $E\{\sup_{i=1,\ldots,n}(\hat{\bm{Z}}_i - \bm{Z}_i^0)^2\}=O(\tau_n^2)$, the dependence of the convergence rate presented in Proposition \ref{prop:eiv} on the intrinsic dimension of metric space-valued responses $r$ would vanish. Specifically, the convergence rate in Proposition \ref{prop:eiv} would be further improved to $(h^{-1}\tau_n)^{1/(\gamma_3-1)}$, leading to an enhanced convergence rate of the proposed DFR model as $h^{2/(\gamma_1-1)}+(nh^r)^{-1/{2(\gamma_2-1)}}+(h^{-1}\tau_n)^{1/(\gamma_3-1)}$. The proof is provided in Section \ref{subsec:remark2} of the Supplementary Material. While the primary focus of the current work is to devise a general framework for handling multivariate predictors and metric space-valued responses, an intriguing problem for future work is to derive such a supremum rate for DNNs when dealing with dependent sub-Gaussian noise coupled with bias.  

\section{Implementation and Simulations}
\label{sec:sim}
\subsection{Implementation details}
The optimization problem in \eqref{eq:DNN Ln} was addressed within \texttt{Pytorch} framework, where we opted for the Adam optimizer \citep{king:14}, a widely used choice in deep learning known for its robust performance across diverse models and datasets. The weight matrices $W_l$ and shift vectors $\bm{b}_l$ of the function $g\in\mathcal{G}$ as per \eqref{eq:DNN Class} were initialized using the default random initialization provided by \texttt{Pytorch}.

ISOMAP and DNNs involve tuning parameters such as the number of nearest neighbors $k$ in the Dijkstra algorithm \citep{dijk:59}, the number of hidden layers $L$, the number of neurons $p_l$ in each hidden layer, the dropout rate, and the learning rate. For simplicity, we maintained the same number of neurons in each hidden layer. The dropout rate \citep{sriv:14} is the rate of randomly excluding neurons during training. The learning rate \citep{good:17} determines the step size for gradient descent in the Adam algorithm. Tuning these parameters can be accomplished through a grid search, assessing empirical risk over a held-out validation set after each training iteration. We set aside 20\% of the training set as the validation set in each run. The training process stops when the empirical risk on the validation set ceases to reliably improve.

DNNs offer high expressive capacity but are susceptible to overfitting, especially with smaller sample sizes. To address this, we applied several regularization techniques. One key method is early stopping, which monitors the validation error during training and halts the process when the validation error begins to rise, thus avoiding excessive fitting to the training data \citep{good:17}. We also utilized dropout \citep{sriv:14}, a technique that randomly deactivates a subset of neurons during training, reducing the model's capacity and mitigating overfitting. For sample sizes $n = 100, 200, 500, 1000$, the runtime of the proposed DFR model on a MacBook Pro with an Apple M2 chip was 0.26, 0.48, 1.09, and 2.29 minutes for the case of distributional objects, and 0.24, 0.44, 1.07, and 2.22 minutes for network objects, respectively.

\subsection{Simulations for probability distributions}
\label{sec:simu:dist}
To assess the performance of the proposed DFR model, we report the results of simulations for various settings. The random objects we consider first are univariate probability distributions equipped with the Wasserstein metric; see Example \ref{exm:was}. The Wasserstein space is a geodesic metric space related to optimal transport \citep{vill:03} and is increasingly applied in various research domains, including population pyramids \citep{bigo:17, deli:11}, financial returns \citep{zhan:22}, and multi-cohort studies \citep{zhou:23}, among others \citep{pete:22}. Additional simulations for networks can be found in Section \ref{supp:network} of the Supplementary Material. Despite the simulated networks not residing on a 2-dimensional manifold, the DFR model exhibits remarkable robustness as the sample size increases.

Consider the true regression function, represented as a quantile function, to be
\[m(\bm{X}) = E(\eta|\bm{X})+E(\sigma|\bm{X}) \Phi^{-1}(\cdot),\]
which corresponds to a Gaussian distribution with mean and standard deviation depending on $\bm{X}$. The distribution parameters of $m(\bm{X})$ are generated conditionally on $\bm{X}$, where the mean $\eta$ and the standard deviation $\sigma$ are assumed to follow a normal distribution and a Gamma distribution, respectively.

We consider sample sizes $n=100,200,500,1000$, with $Q=500$ Monte Carlo runs. In each Monte Carlo run, predictors $\bm{X}_i \in \mathbb{R}^9$, $i=1, \ldots, n$, are independently sampled as
\begin{align*}
        &X_{i1} \sim U(-1,0),\ X_{i2} \sim U(0,1),\ X_{i3} \sim U(1,2)\\&
        X_{i4} \sim N(0,1),\ X_{i5} \sim N(-10,3),\ X_{i6} \sim N(10,3)\\&
        X_{i7} \sim \mathrm{Bernoulli}(0.6),\ X_{i8} \sim \mathrm{Bernoulli}(0.7),\ X_{i9} \sim \mathrm{Bernoulli}(0.3).
\end{align*}
For each $i$, the mean and standard deviation of the Gaussian distribution are generated conditionally on $\bm{X}_i$ as
\begin{align*}
         &\eta_i|\bm{X}_i\sim N(\mu, 0.5^2), \sigma_i|\bm{X}_i \sim\text{Gamma}(\theta^2, \theta^{-1})\text{ where }\\
         &\mu=3 X_{i8}\{\sin (\pi X_{i1})+\cos(\pi X_{i2})\}  + X_{i7}(5X_{i4}^2+X_{i5}) ,\\&
        \theta=[3+0.5 X_{i8}\{ \sin (\pi X_{i1}) +\cos(\pi X_{i2})\} + X_{i7}|5X_{i4}^2 + X_{i5}|].
\end{align*}
The corresponding distributional response is then constructed as $Y_i = \eta_i + \sigma_i \Phi^{-1}$. 

To simulate practical scenarios where direct observations of probability distributions are not available, but rather independent data samples generated by the corresponding distribution are obtained, we independently sample 100 observations $\{y_{ij}\}_{j=1}^{100}$ from each distributional response $Y_i$. Consequently, one must initially estimate the distributional response $Y_i$ from the random sample $\{y_{ij}\}_{j=1}^{100}$, introducing a bias to the regression model. This practical consideration aligns with prior research \citep{zhou:23}, where the empirical measure in place of the unobservable distributional response $Y_i$ was adopted. For the implementation of local \f regression using empirical measures, we follow the algorithm outlined in \citet{zhou:23}. The bandwidths for the local \f regression in each Monte Carlo run are chosen as $10\%$ of the range of the intermediate estimates $\hat{\bm{Z}}_i = \hat{\bm{g}}(\bm{X}_i)$.

For the $q$th Monte Carlo run out of $Q=500$ runs, with $\hat{m}_q(\cdot)$ denoting the fitted regression function, the quality of the estimation is quantified by the mean squared prediction error (MSPE),
$$\mathrm{MSPE}_q = \frac{1}{100}\sum_{i=1}^{100}d_{\mathcal{W}}^2\{\hat{m}_{q}(\bm{X}_i^{\text{test}}), m(\bm{X}_i^{\text{test}})\},$$
where $\{\bm{X}_i^{\text{test}}\}_{i=1}^{100}$ denote out-of-sample predictors and the metric $d_{\mathcal{W}}$ is the Wasserstein metric for probability distributions. The average quality of the estimation over the $Q=500$ Monte Carlo runs is quantified  by the average mean squared prediction error (AMSPE)
\[\mathrm{AMSPE}=\frac{1}{Q}\sum_{q=1}^Q\mathrm{MSPE}_q.\]

\begin{figure}[t]
\single
\centering
\begin{subfigure}{0.72\textwidth}
    \centering
    \includegraphics[width=\linewidth]{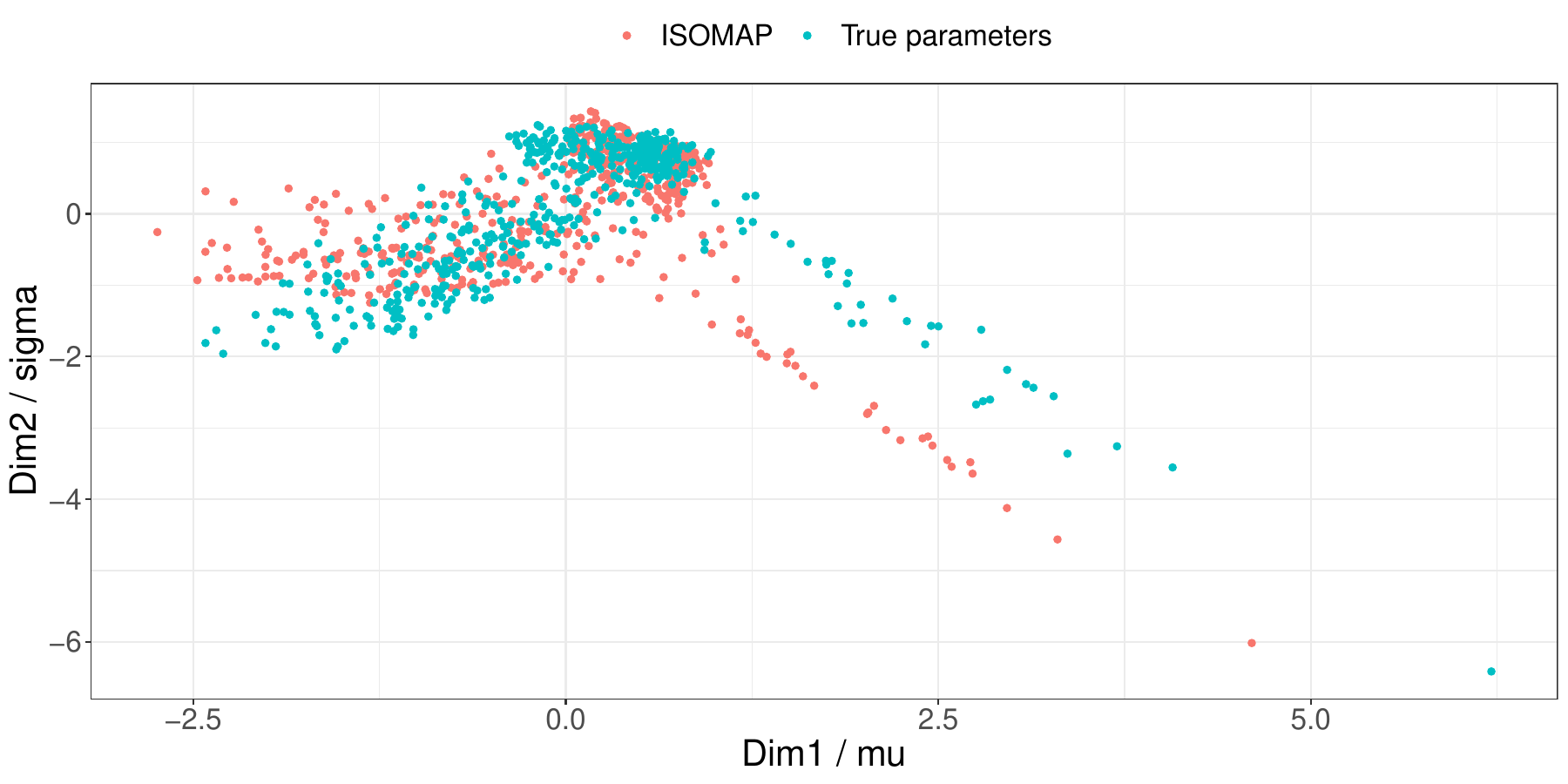}
    \caption{}
\end{subfigure}%

\begin{subfigure}{0.72\textwidth}
    \centering
    \includegraphics[width=\linewidth]{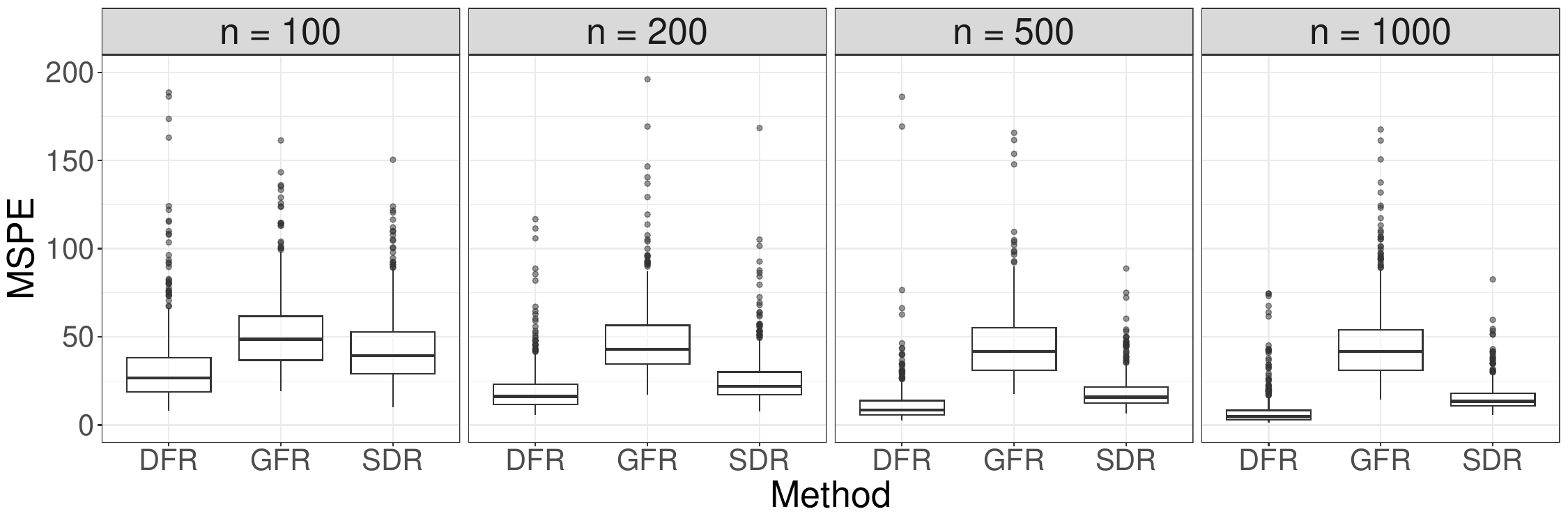}
    \caption{}
\end{subfigure}
\caption{(a) Two-dimensional representations of probability distributions $\{Y_i\}_{i=1}^n$ using ISOMAP, along with their true parameters (mean and standard deviation) for sample size $n=500$. (b) Boxplot of mean squared prediction errors for $Q=500$ Monte Carlo runs,
comparing  deep \f regression (DFR), global \f regression (GFR) \citep{mull:19:6} and sufficient dimension reduction (SDR) \citep{zhan:21:1}.}
\label{fig:dis}
\end{figure}

To assess the validity of the manifold assumption, Figure~\ref{fig:dis}(a) presents a scatter plot depicting the two-dimensional representations of distributional responses $\{Y_i\}_{i=1}^n$ alongside their true parameters (mean and standard deviation) for a sample size of $n=500$. The plot reveals a distinctive horseshoe shape in the two-dimensional representations, implying that the distributional responses adhere to a lower-dimensional manifold. Intriguingly, the two-dimensional representations closely align with the two-dimensional parameters of the distributional responses. This observation suggests that for distributional responses ISOMAP effectively identifies suitable low-dimensional representations corresponding to the latent parameter space.

We conducted comparisons with global \f regression (GFR) \citep{mull:19:6} and local \f regression with sufficient dimension reduction (SDR) \citep{zhan:21:1}. Figure~\ref{fig:dis}(b) summarizes MSPEs for all Monte Carlo runs and various sample sizes using the proposed DFR model, GFR and SDR. The MSPE decreases with increasing sample size, indicating the convergence of the DFR model to the target. Notably, the DFR model exhibits superior performance over GFR and SDR across all sample sizes, including the case  $n=100$. In Table \ref{tab:simu1}, we present the AMSPE for different sample sizes. The DFR model is seen to fare best across different methods. Despite the conventional need for large sample sizes in the performance of deep neural networks, the DFR model surprisingly demonstrates robustness for small sample sizes, potentially owing to its flexibility in accommodating complex regression relationships.

\begin{table}[t]
\single
\centering
\caption{Average mean squared prediction error of deep \f regression (DFR), global \f regression (GFR) \citep{mull:19:6} and sufficient dimension reduction (SDR) \citep{zhan:21:1} for distributional responses.}
\label{tab:simu1}
\begin{tabular}{c|ccc}
\hline
$n$ & DFR & GFR & SDR\\
\hline
100 & 34.000 & 52.573 & 43.598 \\
200 & 20.976 & 48.140 & 26.128 \\
500 & 12.544 & 45.712 & 18.742 \\
1000 & 7.874 & 46.418 & 15.719 \\
\hline
\end{tabular}
\end{table}

\section{Data Application}
\label{sec:app}
Yellow and green taxi trip records in New York City (NYC), including pick-up and drop-off dates, locations, trip distances, payment methods, and passenger counts, are available at \url{https://www.nyc.gov/site/tlc/about/tlc-trip-record-data.page}. Additionally, we collect NYC weather history, including daily average temperature, humidity, wind speed, pressure, and total precipitation from \url{https://www.wunderground.com/history/daily/us/ny/new-york-city/KLGA/date}. The objective is to predict transport networks constructed from taxi trip records using relevant predictors. Given the potential influence of weather conditions on travel plans, the variability in travel patterns across different days of the week, and the impact of daily trip features on the taxi system, we construct a fifteen-dimensional predictor set including daily weather information, indicators for days of the week or holiday, and daily trip features averaged over each day; see Table~\ref{tab:predictor}.

\begin{table}[t]
\single
\centering
\caption{Predictors of New York taxi network data.}
\begin{tabular}{p{0.1\linewidth} | p{0.225\linewidth} | p{0.635\linewidth}}
\hline 
Category & Variables & Explanation \\
\hline 
\multirow{5}{*}{Weather} & 1. Temp & daily average temperature \\
& 2. Humidity & daily average humidity \\
& 3. Wind & daily average windspeed \\
& 4. Pressure & daily average barometric pressure \\
& 5. Precipitation & daily total precipitation \\
\hline 
\multirow{2}{*}{Day} & 6. Mon to Thur & indicator for Monday to Thursday\\
& 7. Sun or Holiday & indicator for Sunday or holidays\\
\hline
\multirow{8}{*}{Trip}& 8. Passenger Count & daily average number of passengers \\
& 9. Trip Distance & daily average trip distance \\
& 10. Fare Amount & daily average fare amount \\
& 11. Tip Amount & daily average tip amount \\
& 12. Tolls Amount & daily average tolls amount\\
& 13. Credit Card & average of credit card indicators for the type of payment \\
& 14. Cash & average of cash indicators for the type of payment \\
& 15. Dispute & average of dispute indicators for the type of payment \\
\hline 
\end{tabular}
    \label{tab:predictor}
\end{table}

We analyze yellow taxi trip records in Manhattan, excluding islands, and divide the 66 taxi zones into 13 regions based on preprocessing procedures outlined in \citet{mull:22:11}. While \citet{mull:22:11} focused on the effect of COVID-19 on the transport network, here we investigate the pre-COVID transport network. Therefore, we limit our analysis to a period of 1092 days from January 1, 2017, to December 31, 2019, excluding 3 outlier days. The main interest regarding these traffic records is the transport network that represents daily passenger movement between the 13 regions. To this end, we build daily undirected networks with nodes standing for the 13 regions and edge weights representing the number of passengers traveling between the regions. Each of these networks is uniquely associated with a $13 \times 13$ graph Laplacian. The two-dimensional representations of 1092 graph Laplacians using MDS are shown in Figure \ref{fig:isomap_taxi}, which suggests a clear separation among Monday to Thursday, Friday or Saturday, and Sunday or holiday in the second dimension. Additionally, as the first dimension increases, the total ridership also increases. The scatter plot exhibits a horseshoe shape, indicating the plausibility of the manifold assumption. 

\begin{figure}[t]
    \single
    \centering
    \includegraphics[width=0.8\linewidth]{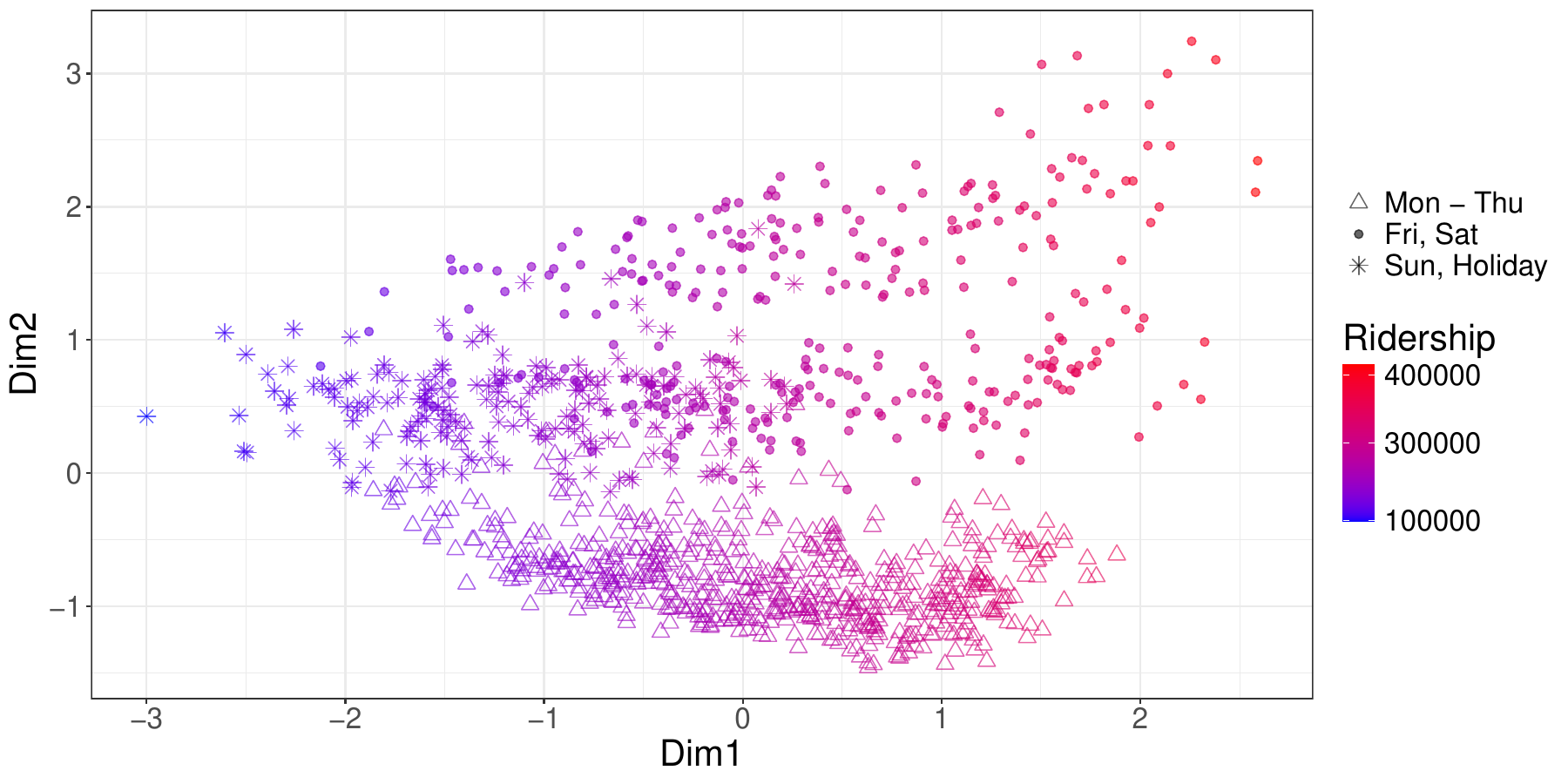}
   \caption{Two-dimensional representations of 1092 taxi networks using MDS with respect to the Frobenius metric.}
   \label{fig:isomap_taxi}
\end{figure}

The proposed DFR model was applied to model the relationship between daily graph Laplacians and the fifteen-dimensional predictors. The MSPE was calculated using $10$-fold cross-validation, averaged over 100 runs. The proposed method achieves better prediction performance, resulting in a $45\%$ and $55\%$ improvement in prediction accuracy compared to GFR \citep{mull:22:11} and SDR \citep{zhan:21:1}, respectively.

\begin{figure}[t]
    \single
    \centering
    \includegraphics[width=0.8\linewidth]{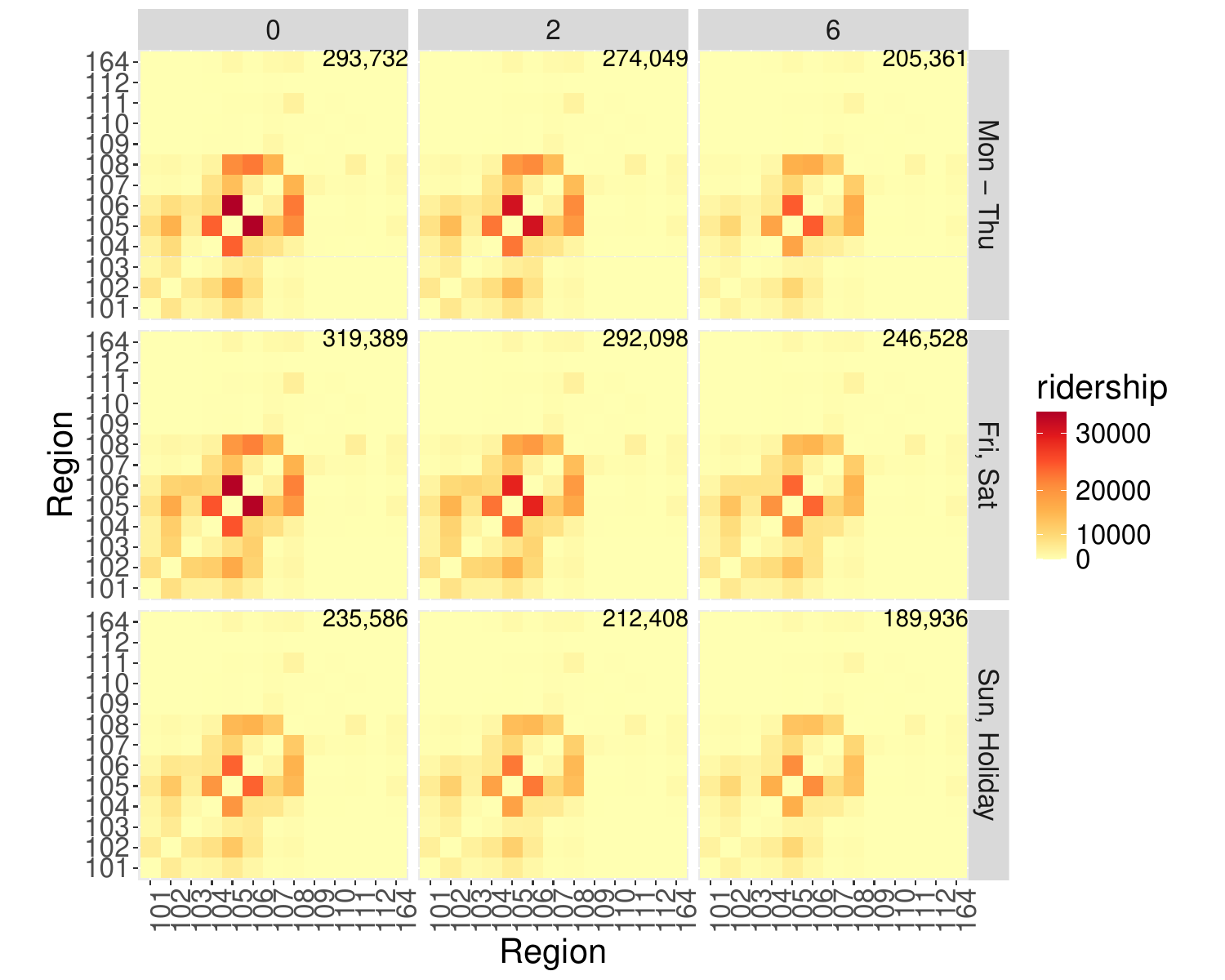}
    \caption{Predicted networks represented as heatmaps at different levels of total precipitation on Monday to Thursday, Friday or Saturday, and Sunday or holiday. The left, middle, and right columns show, respectively, the predicted networks at $0$, $2$, and $6$ inches of total precipitation. The three rows depict the predicted networks on Monday to Thursday, Friday or Saturday, and Sunday or holiday, respectively. The value in the top right corner of each panel represents the total ridership of the entire transport network.}
    \label{fig:preci}
\end{figure}

To further investigate the effect of daily total precipitation on taxi traffic, we predicted networks with varying total precipitation levels, while all other covariates were fixed at their median levels. Figure \ref{fig:preci} shows the predicted networks at total precipitation levels of 0, 2 and 6 inches on Monday to Thursday, Friday or Saturday, and Sunday or holiday. We observed that edge weights decrease on Sunday or holiday compared with those on Monday to Saturday, regardless of total precipitation levels, indicating reduced commute needs on these days. Across different precipitation levels, the taxi traffic from Monday to Thursday is more concentrated in areas 105, 106, 107, and 108, which are primarily residential areas with popular destinations such as Penn Station and Grand Central Terminal. Taxi traffic on Friday and Saturday is more diffuse, spreading into areas 102, 103, and 104, which are home to popular bars, restaurants, chain stores, and high-end art galleries and museums. As the total precipitation increases, taxi traffic decreases, irrespective of the day of the week. When the total precipitation reaches 6 inches, corresponding to a rainstorm, the remaining taxi traffic for all three categories of days is primarily concentrated in areas 105, 106, 107, and 108. This may be related to unavoidable activities such as work-related events, while discretionary travel plans, primarily observed in areas 102, 103, and 104, are more likely to be canceled due to adverse weather conditions.

\section{Discussion}
\label{sec:dis}
The proposed deep \f regression model introduces a novel approach for analyzing metric space-valued responses using multivariate predictors. Extensive simulations and applications across diverse scenarios, including networks and probability distributions, highlight the versatility of the model. The proposed approach is flexible and requires no specific model assumptions. Key contributions include an in-depth exploration of the convergence rate of deep neural networks under dependent sub-Gaussian noise with bias and an extension of the scope of local \f regression to handle multivariate predictors, including the convergence rate in an errors-in-variables regression framework. The proposed approach can be easily adapted to various manifold learning methods, such as t-SNE, UMAP, Laplacian eigenmaps, and diffusion maps; however, these alternative manifold learning methods did not perform as well as ISOMAP in our simulations.

While ISOMAP performed well in simulations and applications, a potential limitation is that errors introduced by the algorithm may propagate through the regression model, affecting final predictions. This issue is addressed in Assumption~\ref{itm:iso1}, which accounts for both dependence and bias from the manifold learning step. The bias propagates through the deep neural networks and local \f regression, influencing the convergence rate in Theorem \ref{thm:pc}. Specifically, the term $u_n$ in the third component $(h^{-r-1}\zeta_n)^{1/(\gamma_3-1)}$ with $\zeta_n = \kappa_n\log^{3/2} n + u_n$ captures the bias induced by ISOMAP. Since it appears that there is no established asymptotic theory for ISOMAP, the exact nature of the term $u_n$ is not known and will be a topic for future research.

The original ISOMAP algorithm \citep{tene:00} performs well in both simulations and real-data applications but may face computational challenges due to the need for shortest-path distance computation via Dijkstra's algorithm and metric multidimensional scaling (MDS) through eigenvalue decomposition. For large-scale datasets, the landmark-ISOMAP algorithm \citep{silv:02} offers a more efficient alternative by selecting $n_0$ landmark points from the total $n$ data points and computing an $n_0 \times n$ geodesic distance matrix. Landmark-MDS is then used to obtain the Euclidean embedding of all points. Similarly, training DNNs on large datasets can be computationally demanding. A common approach to mitigate this is using asynchronous stochastic gradient descent \citep{rech:11, zhen:17}, an efficient optimization method widely employed in industry \citep{chil:14, good:17}.

Another limitation of our approach lies in the assumption of the existence of a low-dimensional representation of random objects. However, our simulations and applications demonstrate the broad applicability of the proposed method, showcasing robust performance even when the intrinsic dimension of metric space-valued responses is not low. Open problems for future research include the supremum rate of DNN convergence when dealing with dependent sub-Gaussian noise with bias.

\appendix
\section*{Appendix A. Kernel and distributional assumptions}
\begin{enumerate}[label=(A\arabic*)]
    \item \label{itm:k1} The $r$-dimensional kernel $K(\cdot)$ satisfies $\int_{\mathbb{R}^r}K(\bm{u})u_{j}^{4}d\bm{u}<\infty$ and $\int_{\mathbb{R}^r}K^2(\bm{u})u_{j}^{6}d\bm{u}<\infty$ for $j=1, \ldots, r$ and $\bm{u}=(u_1, \ldots, u_r)^\T$, and is Lipschitz continuous with compact support $[-1, 1]^r$.
    \item \label{itm:p1} Both the marginal density $f_{\bm{Z}^0}(\cdot)$ of $\bm{Z}^0$ and the conditional densities $f_{\bm{Z}^0| Y}(\cdot, y)$ of $\bm{Z}^0$ given $Y=y$ exist and are twice continuously differentiable, the latter for all $y\in\Omega$, and $\sup_{\bm{z}, y}|(\partial^2 f_{\bm{Z}^0| Y}/\partial \bm{z}^2)(\bm{z}, y)|<\infty$. In addition, for any open set $U\subset\Omega$, $\int_UdF_{Y| \bm{Z}^0}(\bm{z}, y)$ is continuous as a function of $\bm{z}$.
\end{enumerate}

\section*{Supplementary Materials}
Online supplementary material provides auxiliary results and proofs, additional simulation studies for network data, and a second data application involving age-at-death distributions in human mortality. Code for implementing the proposed regression model is available at \url{https://github.com/SUIIAO/Deep-Frechet-Regression}.

\section*{Acknowledgments}
\label{sec:akd}
This research was supported in part by NSF grant DMS-2310450.

\single
\bibliographystyle{chicago}
\bibliography{collection}

\begin{thebibliography}{}

\bibitem[\protect\citeauthoryear{Anthony, Bartlett, and Bartlett}{Anthony et~al.}{1999}]{anth:99}
Anthony, M., P.~L. Bartlett, and P.~L. Bartlett (1999).
\newblock {\em Neural Network Learning: Theoretical Foundations}, Volume~9.
\newblock Cambridge University Press.

\bibitem[\protect\citeauthoryear{Bauer and Kohler}{Bauer and Kohler}{2019}]{baue:19}
Bauer, B. and M.~Kohler (2019).
\newblock {On deep learning as a remedy for the curse of dimensionality in nonparametric regression}.
\newblock {\em Annals of Statistics\/}~{\em 47\/}(4), 2261--2285.

\bibitem[\protect\citeauthoryear{Belkin and Niyogi}{Belkin and Niyogi}{2003}]{belk:03}
Belkin, M. and P.~Niyogi (2003).
\newblock {Laplacian eigenmaps for dimensionality reduction and data representation}.
\newblock {\em Neural Computation\/}~{\em 15\/}(6), 1373--1396.

\bibitem[\protect\citeauthoryear{Bhattacharjee and M{\"u}ller}{Bhattacharjee and M{\"u}ller}{2023}]{mull:23:3}
Bhattacharjee, S. and H.-G. M{\"u}ller (2023).
\newblock Single index {F}r{\'e}chet regression.
\newblock {\em Annals of Statistics\/}~{\em 51\/}(4), 1770--1798.

\bibitem[\protect\citeauthoryear{Bigot, Gouet, Klein, and L{\'o}pez}{Bigot et~al.}{2017}]{bigo:17}
Bigot, J., R.~Gouet, T.~Klein, and A.~L{\'o}pez (2017).
\newblock {Geodesic {PCA} in the {W}asserstein space by convex {PCA}}.
\newblock {\em Annales de l{'}Institut Henri Poincar\'{e} B: Probability and Statistics\/}~{\em 53\/}(1), 1--26.

\bibitem[\protect\citeauthoryear{Bos and Schmidt-Hieber}{Bos and Schmidt-Hieber}{2024}]{bos:23}
Bos, T. and J.~Schmidt-Hieber (2024).
\newblock A supervised deep learning method for nonparametric density estimation.
\newblock {\em Electronic Journal of Statistics\/}~{\em 18\/}(2), 5601--5658.

\bibitem[\protect\citeauthoryear{Chen and M\"{u}ller}{Chen and M\"{u}ller}{2012}]{mull:12:1}
Chen, D. and H.-G. M\"{u}ller (2012).
\newblock {Nonlinear manifold representations for functional data}.
\newblock {\em Annals of Statistics\/}~{\em 40\/}(1), 1--29.

\bibitem[\protect\citeauthoryear{Chen and M{\"u}ller}{Chen and M{\"u}ller}{2022}]{mull:22:8}
Chen, Y. and H.-G. M{\"u}ller (2022).
\newblock {Uniform convergence of local {F}r\'echet regression, with applications to locating extrema and time warping for metric-space valued trajectories}.
\newblock {\em Annals of Statistics\/}~{\em 50\/}(3), 1573--1592.

\bibitem[\protect\citeauthoryear{Chen, Zhou, Chen, Gajardo, Fan, Zhong, Dubey, Han, Bhattacharjee, Zhu, Iao, Kundu, Petersen, and M{\"u}ller}{Chen et~al.}{2023}]{chen:20}
Chen, Y., Y.~Zhou, H.~Chen, A.~Gajardo, J.~Fan, Q.~Zhong, P.~Dubey, K.~Han, S.~Bhattacharjee, C.~Zhu, S.~I. Iao, P.~Kundu, A.~Petersen, and H.-G. M{\"u}ller (2023).
\newblock {\em {frechet: Statistical Analysis for Random Objects and Non-Euclidean Data}}.
\newblock R package version 0.3.0.

\bibitem[\protect\citeauthoryear{Chilimbi, Suzue, Apacible, and Kalyanaraman}{Chilimbi et~al.}{2014}]{chil:14}
Chilimbi, T., Y.~Suzue, J.~Apacible, and K.~Kalyanaraman (2014).
\newblock {Project Adam: Building an efficient and scalable deep learning training system}.
\newblock In {\em 11th USENIX Symposium on Operating Systems Design and Implementation (OSDI 14)}, pp.\  571--582.

\bibitem[\protect\citeauthoryear{Coifman and Lafon}{Coifman and Lafon}{2006}]{coif:06}
Coifman, R.~R. and S.~Lafon (2006).
\newblock Diffusion maps.
\newblock {\em Applied and Computational Harmonic Analysis\/}~{\em 21\/}(1), 5--30.

\bibitem[\protect\citeauthoryear{Cox and Cox}{Cox and Cox}{2001}]{cox:01}
Cox, T.~F. and M.~A.~A. Cox (2001).
\newblock {\em {Multidimensional Scaling.}}
\newblock CRC Press.

\bibitem[\protect\citeauthoryear{De~Silva and Tenenbaum}{De~Silva and Tenenbaum}{2002}]{silv:02}
De~Silva, V. and J.~B. Tenenbaum (2002).
\newblock Global versus local methods in nonlinear dimensionality reduction.
\newblock In {\em Advances in Neural Information Processing Systems}, Volume~15.

\bibitem[\protect\citeauthoryear{Delicado}{Delicado}{2011}]{deli:11}
Delicado, P. (2011).
\newblock {Dimensionality reduction when data are density functions}.
\newblock {\em Computational Statistics \& Data Analysis\/}~{\em 55\/}(1), 401--420.

\bibitem[\protect\citeauthoryear{Dijkstra}{Dijkstra}{1959}]{dijk:59}
Dijkstra, E.~W. (1959).
\newblock A note on two problems in connexion with graphs.
\newblock {\em Numerische Mathematik\/}~{\em 1}, 269--271.

\bibitem[\protect\citeauthoryear{Dryden, Koloydenko, and Zhou}{Dryden et~al.}{2009}]{dryd:09}
Dryden, I.~L., A.~Koloydenko, and D.~Zhou (2009).
\newblock {Non-{E}uclidean statistics for covariance matrices, with applications to diffusion tensor imaging}.
\newblock {\em Annals of Applied Statistics\/}~{\em 3\/}(3), 1102--1123.

\bibitem[\protect\citeauthoryear{Dubey, Chen, and M{\"u}ller}{Dubey et~al.}{2024}]{mull:22:12}
Dubey, P., Y.~Chen, and H.-G. M{\"u}ller (2024).
\newblock Metric statistics: Exploration and inference for random objects with distance profiles.
\newblock {\em Annals of Statistics\/}~{\em 52\/}(2), 757--792.

\bibitem[\protect\citeauthoryear{Fan and Gijbels}{Fan and Gijbels}{1996}]{fan:96}
Fan, J. and I.~Gijbels (1996).
\newblock {\em {Local Polynomial Modelling and its Applications}}.
\newblock London: Chapman \& Hall.

\bibitem[\protect\citeauthoryear{Faraway}{Faraway}{2014}]{fara:14}
Faraway, J.~J. (2014).
\newblock {Regression for non-{E}uclidean data using distance matrices}.
\newblock {\em Journal of Applied Statistics\/}~{\em 41\/}(11), 2342--2357.

\bibitem[\protect\citeauthoryear{Fr\'{e}chet}{Fr\'{e}chet}{1948}]{frec:48}
Fr\'{e}chet, M. (1948).
\newblock {Les \'{e}l\'{e}ments al\'{e}atoires de nature quelconque dans un espace distanci\'{e}}.
\newblock {\em Annales de l'Institut Henri Poincar\'{e}\/}~{\em 10\/}(4), 215--310.

\bibitem[\protect\citeauthoryear{Ghojogh, Crowley, Karray, and Ghodsi}{Ghojogh et~al.}{2023}]{ghoj:23}
Ghojogh, B., M.~Crowley, F.~Karray, and A.~Ghodsi (2023).
\newblock {\em Elements of Dimensionality Reduction and Manifold Learning}.
\newblock Springer New York.

\bibitem[\protect\citeauthoryear{Goodfellow, Bengio, and Courville}{Goodfellow et~al.}{2016}]{good:17}
Goodfellow, I., Y.~Bengio, and A.~Courville (2016).
\newblock {\em Deep Learning}.
\newblock MIT Press.

\bibitem[\protect\citeauthoryear{Han, Pool, Tran, and Dally}{Han et~al.}{2015}]{han:15}
Han, S., J.~Pool, J.~Tran, and W.~Dally (2015).
\newblock Learning both weights and connections for efficient neural network.
\newblock In {\em Advances in Neural Information Processing Systems}, Volume~28.

\bibitem[\protect\citeauthoryear{Hein}{Hein}{2009}]{hein:09}
Hein, M. (2009).
\newblock Robust nonparametric regression with metric-space valued output.
\newblock In {\em Advances in Neural Information Processing Systems}, pp.\  718--726.

\bibitem[\protect\citeauthoryear{Kingma and Ba}{Kingma and Ba}{2015}]{king:14}
Kingma, D.~P. and J.~Ba (2015).
\newblock Adam: A method for stochastic optimization.
\newblock In {\em Proceedings of the 3rd International Conference on Learning Representations}.

\bibitem[\protect\citeauthoryear{Kohler and Langer}{Kohler and Langer}{2021}]{kohl:21}
Kohler, M. and S.~Langer (2021).
\newblock On the rate of convergence of fully connected deep neural network regression estimates.
\newblock {\em Annals of Statistics\/}~{\em 49\/}(4), 2231--2249.

\bibitem[\protect\citeauthoryear{Kurisu, Fukami, and Koike}{Kurisu et~al.}{2025}]{kuri:22}
Kurisu, D., R.~Fukami, and Y.~Koike (2025).
\newblock Adaptive deep learning for nonparametric time series regression.
\newblock {\em Bernoulli\/}~{\em 31\/}(1), 240--270.

\bibitem[\protect\citeauthoryear{LeCun, Bengio, and Hinton}{LeCun et~al.}{2015}]{lecu:15}
LeCun, Y., Y.~Bengio, and G.~Hinton (2015).
\newblock Deep learning.
\newblock {\em Nature\/}~{\em 521\/}(7553), 436--444.

\bibitem[\protect\citeauthoryear{McInnes and Healy}{McInnes and Healy}{2018}]{mcin:18}
McInnes, L. and J.~Healy (2018).
\newblock {UMAP: Uniform manifold approximation and projection for dimension reduction}.
\newblock {\em arXiv preprint arXiv:1802.03426\/}.

\bibitem[\protect\citeauthoryear{M\"{u}ller}{M\"{u}ller}{2016}]{mull:16:7}
M\"{u}ller, H.-G. (2016).
\newblock {{P}eter {H}all, {F}unctional {D}ata {A}nalysis and {R}andom {O}bjects}.
\newblock {\em Annals of Statistics\/}~{\em 44\/}(5), 1867--1887.

\bibitem[\protect\citeauthoryear{Nair and Hinton}{Nair and Hinton}{2010}]{nair:10}
Nair, V. and G.~E. Hinton (2010).
\newblock Rectified linear units improve restricted {B}oltzmann machines.
\newblock In {\em Proceedings of the 27th International Conference on Machine Learning}, pp.\  807--814.

\bibitem[\protect\citeauthoryear{Nye, Tang, Weyenberg, and Yoshida}{Nye et~al.}{2017}]{nye:17}
Nye, T.~M., X.~Tang, G.~Weyenberg, and R.~Yoshida (2017).
\newblock Principal component analysis and the locus of the {F}r{\'e}chet mean in the space of phylogenetic trees.
\newblock {\em Biometrika\/}~{\em 104\/}(4), 901--922.

\bibitem[\protect\citeauthoryear{Petersen and M{\"u}ller}{Petersen and M{\"u}ller}{2019}]{mull:19:6}
Petersen, A. and H.-G. M{\"u}ller (2019).
\newblock {Fr\'{e}chet regression for random objects with Euclidean predictors}.
\newblock {\em Annals of Statistics\/}~{\em 47\/}(2), 691--719.

\bibitem[\protect\citeauthoryear{Petersen, Zhang, and Kokoszka}{Petersen et~al.}{2022}]{pete:22}
Petersen, A., C.~Zhang, and P.~Kokoszka (2022).
\newblock {Modeling probability density functions as data objects}.
\newblock {\em Econometrics and Statistics\/}~{\em 21}, 159--178.

\bibitem[\protect\citeauthoryear{Recht, Re, Wright, and Niu}{Recht et~al.}{2011}]{rech:11}
Recht, B., C.~Re, S.~Wright, and F.~Niu (2011).
\newblock Hogwild!: A lock-free approach to parallelizing stochastic gradient descent.
\newblock In {\em Advances in Neural Information Processing Systems}, Volume~24.

\bibitem[\protect\citeauthoryear{Schmidt-Hieber}{Schmidt-Hieber}{2020}]{schm:20}
Schmidt-Hieber, J. (2020).
\newblock {Nonparametric regression using deep neural networks with ReLU activation function}.
\newblock {\em Annals of Statistics\/}~{\em 48\/}(4), 1875--1897.

\bibitem[\protect\citeauthoryear{Sch{\"o}tz}{Sch{\"o}tz}{2022}]{scho:22}
Sch{\"o}tz, C. (2022).
\newblock {Nonparametric regression in nonstandard spaces}.
\newblock {\em Electronic Journal of Statistics\/}~{\em 16\/}(2), 4679--4741.

\bibitem[\protect\citeauthoryear{Song and Han}{Song and Han}{2023}]{han:23}
Song, D. and K.~Han (2023).
\newblock Errors-in-variables {F}r\'echet regression with low-rank covariate approximation.
\newblock In {\em Advances in Neural Information Processing Systems}, Volume~36, pp.\  80575--80607.

\bibitem[\protect\citeauthoryear{Srinivas, Subramanya, and Venkatesh~Babu}{Srinivas et~al.}{2017}]{srin:17}
Srinivas, S., A.~Subramanya, and R.~Venkatesh~Babu (2017).
\newblock Training sparse neural networks.
\newblock In {\em Computer Vision and Pattern Recognition}, pp.\  138--145.

\bibitem[\protect\citeauthoryear{Srivastava, Hinton, Krizhevsky, Sutskever, and Salakhutdinov}{Srivastava et~al.}{2014}]{sriv:14}
Srivastava, N., G.~Hinton, A.~Krizhevsky, I.~Sutskever, and R.~Salakhutdinov (2014).
\newblock Dropout: a simple way to prevent neural networks from overfitting.
\newblock {\em Journal of Machine Learning Research\/}~{\em 15\/}(1), 1929--1958.

\bibitem[\protect\citeauthoryear{Stellato, Banjac, Goulart, Bemporad, and Boyd}{Stellato et~al.}{2020}]{stel:20}
Stellato, B., G.~Banjac, P.~Goulart, A.~Bemporad, and S.~Boyd (2020).
\newblock {OSQP}: An operator splitting solver for quadratic programs.
\newblock {\em Mathematical Programming Computation\/}~{\em 12\/}(4), 637--672.

\bibitem[\protect\citeauthoryear{Sturm}{Sturm}{2003}]{stur:03}
Sturm, K.-T. (2003).
\newblock {Probability measures on metric spaces of nonpositive curvature}.
\newblock {\em Heat Kernels and Analysis on Manifolds, Graphs, and Metric Spaces (Paris, 2002). Contemp. Math., 338. Amer. Math. Soc., Providence, RI\/}~{\em 338}, 357--390.

\bibitem[\protect\citeauthoryear{Tenenbaum, De~Silva, and Langford}{Tenenbaum et~al.}{2000}]{tene:00}
Tenenbaum, J.~B., V.~De~Silva, and J.~C. Langford (2000).
\newblock {A global geometric framework for nonlinear dimensionality reduction}.
\newblock {\em Science\/}~{\em 290\/}(5500), 2319--2323.

\bibitem[\protect\citeauthoryear{Tucker, Wu, and M{\"u}ller}{Tucker et~al.}{2023}]{mull:21:1}
Tucker, D.~C., Y.~Wu, and H.-G. M{\"u}ller (2023).
\newblock Variable selection for global {F}r{\'e}chet regression.
\newblock {\em Journal of the American Statistical Association\/}~{\em 118\/}(542), 1023--1037.

\bibitem[\protect\citeauthoryear{Van~der Maaten and Hinton}{Van~der Maaten and Hinton}{2008}]{van:08}
Van~der Maaten, L. and G.~Hinton (2008).
\newblock Visualizing data using {t-SNE}.
\newblock {\em Journal of Machine Learning Research\/}~{\em 9\/}(86), 2579--2605.

\bibitem[\protect\citeauthoryear{Van~der Vaart and Wellner}{Van~der Vaart and Wellner}{2023}]{well:23}
Van~der Vaart, A. and J.~Wellner (2023).
\newblock {\em Weak Convergence and Empirical Processes: With Applications to Statistics}.
\newblock Springer New York.

\bibitem[\protect\citeauthoryear{Villani}{Villani}{2003}]{vill:03}
Villani, C. (2003).
\newblock {\em {Topics in Optimal Transportation}}.
\newblock American Mathematical Society.

\bibitem[\protect\citeauthoryear{Ying and Yu}{Ying and Yu}{2022}]{ying:22}
Ying, C. and Z.~Yu (2022).
\newblock {F}r{\'e}chet sufficient dimension reduction for random objects.
\newblock {\em Biometrika\/}~{\em 109\/}(4), 975--992.

\bibitem[\protect\citeauthoryear{Zhang, Kokoszka, and Petersen}{Zhang et~al.}{2022}]{zhan:22}
Zhang, C., P.~Kokoszka, and A.~Petersen (2022).
\newblock {Wasserstein} autoregressive models for density time series.
\newblock {\em Journal of Time Series Analysis\/}~{\em 43\/}(2), 30--52.

\bibitem[\protect\citeauthoryear{Zhang, Xue, and Li}{Zhang et~al.}{2024}]{zhan:21:1}
Zhang, Q., L.~Xue, and B.~Li (2024).
\newblock Dimension reduction for {F}r{\'e}chet regression.
\newblock {\em Journal of the American Statistical Association\/}~{\em 119\/}(548), 2733--2747.

\bibitem[\protect\citeauthoryear{Zheng, Meng, Wang, Chen, Yu, Ma, and Liu}{Zheng et~al.}{2017}]{zhen:17}
Zheng, S., Q.~Meng, T.~Wang, W.~Chen, N.~Yu, Z.-M. Ma, and T.-Y. Liu (2017).
\newblock Asynchronous stochastic gradient descent with delay compensation.
\newblock In {\em Proceedings of the 34th International Conference on Machine Learning}, pp.\  4120--4129.

\bibitem[\protect\citeauthoryear{Zhong, Mueller, and Wang}{Zhong et~al.}{2022}]{zhon:22}
Zhong, Q., J.~Mueller, and J.-L. Wang (2022).
\newblock Deep learning for the partially linear {C}ox model.
\newblock {\em Annals of Statistics\/}~{\em 50\/}(3), 1348--1375.

\bibitem[\protect\citeauthoryear{Zhou and M\"{u}ller}{Zhou and M\"{u}ller}{2022}]{mull:22:11}
Zhou, Y. and H.-G. M\"{u}ller (2022).
\newblock {Network regression with graph {L}aplacians}.
\newblock {\em Journal of Machine Learning Research\/}~{\em 23\/}(320), 1--41.

\bibitem[\protect\citeauthoryear{Zhou and M{\"u}ller}{Zhou and M{\"u}ller}{2024}]{zhou:23}
Zhou, Y. and H.-G. M{\"u}ller (2024).
\newblock Wasserstein regression with empirical measures and density estimation for sparse data.
\newblock {\em Biometrics\/}~{\em 80\/}(4), ujae127.

\end{thebibliography}
\double
\newpage

\bigskip
\begin{center}
{\large\bf Supplementary Material}
\end{center}
\setcounter{subsection}{0}
\renewcommand{\thesubsection}{S.\arabic{subsection}}
\setcounter{prop}{0}
\setcounter{lem}{0}
\setcounter{defi}{0}
\setcounter{equation}{0}
\setcounter{figure}{0}
\setcounter{table}{0}
\renewcommand{\theprop}{S\arabic{prop}}
\renewcommand{\thelem}{S\arabic{lem}}
\renewcommand{\thedefi}{S\arabic{defi}}
\renewcommand{\theequation}{S\arabic{equation}}
\renewcommand{\thefigure}{S\arabic{figure}}
\renewcommand{\thetable}{S\arabic{table}}

\subsection{Covering Number}\label{supp:back}
\begin{defi}[Covering number]\label{def:covering}
    Let $(\Omega, d)$ be a metric space with a subset $S \subseteq \Omega$. For  $\varepsilon>0$, the $\varepsilon$-covering number of $S$, denoted by $N(\varepsilon, S, d)$, is the minimal number of balls of radius $\varepsilon$ needed to cover the set $S$, i.e.,
$$N(\varepsilon, S, d):=\min \{k:\text{ There exist }y_1, \ldots, y_k \in \Omega\text{ such that }S \subseteq \bigcup_{i=1}^k B_\varepsilon(y_i)\},
$$
where $B_\varepsilon(y)=\{y' \in \Omega: d(y, y') \leq \varepsilon\}$ denotes the closed $\varepsilon$-ball centered at $y \in \Omega$.
\end{defi}

\subsection{ISOMAP}\label{supp:mani}
Consider a totally bounded metric space $(\Omega, d)$ with distance function $d: \Omega\times\Omega \mapsto [0, \infty)$, and let $\mathcal{M}\subset\Omega$ be a manifold isomorphic to a subspace of the Euclidean space. There exists a bijective representation map of the manifold, $\bm{\psi}:\mathcal{M} \mapsto \mathbb{R}^r$, where $r$ is the intrinsic dimension of the manifold $\mathcal{M}$. Denote the low-dimensional representation of any object $Y\in\mathcal{M}$ as $\bm{\psi}(Y)\in\mathbb{R}^r$.

Consider the geodesic distance on the manifold $\mathcal{M}$:
\begin{equation*}
    d_g(y_1, y_2)=\inf_{\gamma} \{L(\gamma): \gamma(0)=y_1, \gamma(1)=y_2\},
\end{equation*}
where $\gamma:[0, 1]\mapsto\mathcal{M}$ is a continuous curve defined on the manifold with the length operator
\begin{equation*}
   L(\gamma)=\sup \sum_{i=0}^{n-1}d\{\gamma(s_{i+1}),\gamma(s_i)\}. 
\end{equation*}
The supremum is taken over all partitions of the interval $[0,1]$ with arbitrary break points $0=s_0<\cdots<s_n=1$. The geodesic distance is the length of the shortest path on $\mathcal{M}$ connecting the two points and therefore is adapted to $\mathcal{M}$.

In practice, only samples of random objects $\{Y_i\}_{i=1}^n\in\mathcal{M}$ are available, while the underlying low-dimensional representations $\bm{\psi}(Y_i)$, $\, i=1, \ldots, n,$ remain unknown. Following \citet{tene:00}, we employ the ISOMAP algorithm to estimate the unknown representation map $\bm{\psi}$ that preserves the geodesic distance over $\mathcal{M}$.

To estimate the geodesic distance using $\{Y_i\}_{i=1}^n$, one can implement Dijkstra's algorithm \citep{dijk:59} for the neighborhood graph constructed using the original distance $d$. The representation map $\bm{\psi}$ is then estimated at $\{Y_i\}_{i=1}^n$ via the following optimization problem:
\begin{equation*}
\argmin_{\{\bm{\psi}(Y_1), \ldots, \bm{\psi}(Y_n)\}}\sum_{i, j=1}^n\{\|\bm{\psi}(Y_i)-\bm{\psi}(Y_j)\|-\hat{d}_g(Y_i, Y_j)\}^2,  
\end{equation*}
where $\hat{d}_g$ is the estimated geodesic distance using Dijkstra's algorithm and the minimum is taken over the vectors $\bm{\psi}(Y_i) \in \mathbb{R}^r$, $i=1, \ldots, n$. The minimizer $\hat{\bm{\psi}}(Y_i)$, $i=1, \ldots, n$ can be obtained by multidimensional scaling (MDS) \citep{cox:01}. 

\subsection{Verification of Model Assumptions}
\label{supp:prop}
\begin{prop}
	\label{prop:wass}
    The Wasserstein space $(\Omega, d_{\mathcal{W}})$ defined in Example \ref{exm:was} satisfies Assumptions \ref{itm:l1}--\ref{itm:l3}.
\end{prop}
\begin{proof}
    As discussed in Example \ref{exm:was} of \citet{mull:22:8}, the Wasserstein space $(\Omega, d_{\mathcal{W}})$ satisfies the first three parts of \ref{itm:l1}, \ref{itm:l2}, and the first two parts of \ref{itm:l3} with $\gamma_1=\gamma_2=2$. It suffices to show \ref{itm:l1} (iv) and \ref{itm:l3} (iii). For any probability distribution $y\in\Omega$, let $F_y^{-1}$ be the corresponding quantile function. Let $\langle\cdot, \cdot\rangle$, $\|\cdot\|_{L^2}$, and $d_{L^2}(\cdot, \cdot)$ be the inner product, norm, and distance on the Hilbert space $L^2(0, 1)$. For any $y\in\Omega$, the map from $y$ to $F_{y}^{-1}$ is an isometry from $\Omega$ to the subset of $L^2(0, 1)$ formed by equivalence classes of left-continuous nondecreasing functions on $(0, 1)$. The Wasserstein space $\Omega$ can thus be viewed as a subset of $L^2(0, 1)$, which has been shown to be convex and closed \citep{bigo:17}.
    
    Let
    \begin{equation}\label{eq:B:wass}
        \hat{B}(\bm{z})=\frac{1}{n}\sum_{i=1}^n\hat{w}(\hat{\bm{Z}}_i, \bm{z}, h)F_{Y_i}^{-1}.
    \end{equation}
    Since 
    \begin{align*}
        \hat{Q}_h(y, \bm{z})&=\frac{1}{n}\sum_{i=1}^n\hat{w}(\hat{\bm{Z}}_i, \bm{z}, h)d_{\mathcal{W}}^2(Y_i, y)\\&=\frac{1}{n}\sum_{i=1}^n\hat{w}(\hat{\bm{Z}}_i, \bm{z}, h)[d_{L^2}^2\{F_{Y_i}^{-1}, \hat{B}(\bm{z})\}+d_{L^2}^2\{\hat{B}(\bm{z}), F_y^{-1}\}+\\&\hspace{1.5em}2\langle F_{Y_i}^{-1}-\hat{B}(\bm{z}), \hat{B}(\bm{z})-F_y^{-1}\rangle_{L^2}]\\&=\hat{Q}_h(\hat{B}(\bm{z}), \bm{z})+d_{L^2}^2\{\hat{B}(\bm{z}), F_y^{-1}\}+\\&\hspace{1.5em}\frac{2}{n}\sum_{i=1}^n\hat{w}(\hat{\bm{Z}}_i, \bm{z}, h)\langle F_{Y_i}^{-1}-\hat{B}(\bm{z}), \hat{B}(\bm{z})-F_y^{-1}\rangle_{L^2}
    \end{align*}
    and
    \begin{align*}
        &\frac{1}{n}\sum_{i=1}^n\hat{w}(\hat{\bm{Z}}_i, \bm{z}, h)\langle F_{Y_i}^{-1}-\hat{B}(\bm{z}), \hat{B}(\bm{z})-F_y^{-1}\rangle_{L^2}\\&=\langle \frac{1}{n}\sum_{i=1}^n\hat{w}(\hat{\bm{Z}}_i, \bm{z}, h)F_{Y_i}^{-1}-\hat{B}(\bm{z}), \hat{B}(\bm{z})-F_y^{-1}\rangle_{L^2}\\&=\langle \hat{B}(\bm{z})-\hat{B}(\bm{z}), \hat{B}(\bm{z})-F_y^{-1}\rangle_{L^2}\\&=0,
    \end{align*}
    one has
    \[\hat{Q}_h(y, \bm{z})=\hat{Q}_h(\hat{B}(\bm{z}), \bm{z})+d_{L^2}^2\{\hat{B}(\bm{z}), F_y^{-1}\},\]
    whence
    \begin{equation}\label{eq:v_hat:wass}
        \hat{v}_h(\bm{z})=\argmin_{y\in\Omega}\hat{Q}_h(y, \bm{z})=\argmin_{y\in\Omega}d_{L^2}^2\{\hat{B}(\bm{z}), F_y^{-1}\}.
    \end{equation}
    Then by the convexity and closedness of the Wasserstein space, the minimizer $\hat{v}_h(\bm{z})$ exists and is unique for any $\bm{z}\in\mathbb{R}^r$. Hence \ref{itm:l1} (iv) is satisfied.
    
    In view of 
    \[\frac{1}{n}\sum_{i=1}^n\tilde{w}(\bm{Z}_i^0, \bm{z}, h)=1,\]
    one can similarly show that 
    \[\tilde{v}_h(\bm{z})=\argmin_{y\in\Omega}\tilde{Q}_h(y, \bm{z})=\argmin_{y\in\Omega}d_{L^2}^2\{\tilde{B}(\bm{z}), F_y^{-1}\},\]
    where
    \[\tilde{B}(\bm{z})=\frac{1}{n}\sum_{i=1}^n\tilde{w}(\bm{Z}_i^0, \bm{z}, h)F_{Y_i}^{-1}.\]
    Observe that $\tilde{v}_h(\bm{z})$, viewed as the best approximation of $\tilde{B}(\bm{z})$ in $\Omega$, is characterized by 
    \[\langle \tilde{B}(\bm{z})-F_{\tilde{v}_h(\bm{z})}^{-1}, F_y^{-1}-F_{\tilde{v}_h(\bm{z})}^{-1}\rangle_{L^2}\leq0,\quad\text{for all }y\in\Omega.\]
    It follows that
    \begin{align*}
        \tilde{Q}_h(y, \bm{z})&=\frac{1}{n}\sum_{i=1}^n\tilde{w}(\bm{Z}_i^0, \bm{z}, h)d_{\mathcal{W}}^2(Y_i, y)\\&=\tilde{Q}_h\{\tilde{v}_h(\bm{z}), \bm{z}\}+d_{\mathcal{W}}^2\{\tilde{v}_h(\bm{z}), y\}+\\&\hspace{1.5em}\frac{2}{n}\sum_{i=1}^n\tilde{w}(\bm{Z}_i^0, \bm{z}, h)\langle F_{Y_i}^{-1}-F_{\tilde{v}_h(\bm{z})}^{-1}, F_{\tilde{v}_h(\bm{z})}^{-1}-F_y^{-1}\rangle_{L^2}\\&=\tilde{Q}_h\{\tilde{v}_h(\bm{z}), \bm{z}\}+d_{\mathcal{W}}^2\{\tilde{v}_h(\bm{z}), y\}+2\langle \tilde{B}(\bm{z})-F_{\tilde{v}_h(\bm{z})}^{-1}, F_{\tilde{v}_h(\bm{z})}^{-1}-F_y^{-1}\rangle_{L^2}\\&\geq\tilde{Q}_h\{\tilde{v}_h(\bm{z}), \bm{z}\}+d_{\mathcal{W}}^2\{\tilde{v}_h(\bm{z}), y\}
    \end{align*}
    for all $y\in\Omega$. Therefore, \ref{itm:l3} (iii) holds for the Wasserstein space $(\Omega, d_{\mathcal{W}})$ for any $\eta_3>0$, $C_3=1$ and $\gamma_3=2$.
\end{proof}

\begin{prop}
\label{prop:mat}
    The metric space $(\Omega, d_{F})$ defined in Example \ref{exm:mat} satisfies Assumptions \ref{itm:l1}--\ref{itm:l3}, where $\Omega$ is the space of graph Laplacians of undirected weighted networks with a fixed number of nodes $m$ and bounded edge weights.
\end{prop}
\begin{proof}
    As discussed in Theorem 3 of \citet{mull:22:11}, the space of graph Laplacians $(\Omega, d_{F})$ satisfies the first three parts of \ref{itm:l1}, \ref{itm:l2}, and the first two parts of \ref{itm:l3} with $\gamma_1=\gamma_2=2$. It suffices to show \ref{itm:l1} (iv) and \ref{itm:l3} (iii). Let $\langle\cdot, \cdot\rangle_F$ and $\|\cdot\|_{F}$ be the Frobenius inner product and norm on $\Omega$. Define
    \begin{equation}\label{eq:B:mat}
        \hat{B}(\bm{z})=\frac{1}{n}\sum_{i=1}^n\hat{w}(\hat{\bm{Z}}_i, \bm{z}, h)Y_i,\quad\tilde{B}(\bm{z})=\frac{1}{n}\sum_{i=1}^n\tilde{w}(\bm{Z}_i^0, \bm{z}, h)Y_i.
    \end{equation}
    Similar to the proof of Proposition \ref{prop:wass}, one can show that
    \begin{equation}\label{eq:v_hat:mat}
        \hat{v}_h(\bm{z})=\argmin_{y\in\Omega}d_{F}^2\{\hat{B}(\bm{z}), y\},\quad\tilde{v}_h(\bm{z})=\argmin_{y\in\Omega}d_{F}^2\{\tilde{B}(\bm{z}), y\}.
    \end{equation}
    Then by the convexity and closedness of the space of graph Laplacians, the minimizer $\hat{v}_h(\bm{z})$ exists and is unique for any $\bm{z}\in\mathbb{R}^r$. Hence \ref{itm:l1} (iv) is satisfied.
    
    Observe that $\tilde{v}_h(\bm{z})$, viewed as the best approximation of $\tilde{B}(\bm{z})$ in $\Omega$, is characterized by
    \[\langle \tilde{B}(\bm{z})-\tilde{v}_h(\bm{z}), y-\tilde{v}_h(\bm{z})\rangle_{F}\leq0,\quad\text{for all }y\in\Omega.\]
    It follows that
    \begin{align*}
        \tilde{Q}_h(y, \bm{z})&=\frac{1}{n}\sum_{i=1}^n\tilde{w}(\bm{Z}_i^0, \bm{z}, h)d_{F}^2(Y_i, y)\\
         &=\tilde{Q}_h\{\tilde{v}_h(\bm{z}), \bm{z}\}+d_{F}^2\{\tilde{v}_h(\bm{z}), y\}+\\&\hspace{1.5em}\frac{2}{n}\sum_{i=1}^n\tilde{w}(\bm{Z}_i^0, \bm{z}, h)\langle Y_i-\tilde{v}_h(\bm{z}), \tilde{v}_h(\bm{z})-y\rangle_{F}\\
         &=\tilde{Q}_h\{\tilde{v}_h(\bm{z}), \bm{z}\}+d_{F}^2\{\tilde{v}_h(\bm{z}), y\}+2\langle \tilde{B}(\bm{z})-\tilde{v}_h(\bm{z}), \tilde{v}_h(\bm{z})-y\rangle_{F}\\
        &\geq\tilde{Q}_h\{\tilde{v}_h(\bm{z}), \bm{z}\}+d_{F}^2\{\tilde{v}_h(\bm{z}), y\}
    \end{align*}
    for all $y\in\Omega$. Therefore, \ref{itm:l3} (iii) holds for the space of graph Laplacians $(\Omega, d_{F})$ for any $\eta_3>0$, $C_3=1$ and $\gamma_3=2$.
\end{proof}

\begin{prop}
\label{prop:cov}
    The metric space $(\Omega, d_{F})$ defined in Example \ref{exm:cov} satisfies Assumptions \ref{itm:l1}--\ref{itm:l3}, where $\Omega$ is the space of $m$-dimensional covariance matrices with bounded diagonal entries or correlation matrices.
\end{prop}
The proof of Proposition \ref{prop:cov} is similar to that of Proposition \ref{prop:mat} and is omitted.

\subsection{Implementation of Local \f Regression}\label{supp:lfr:implement}
We describe the implementation of local \f regression as per \eqref{eq:lfrhat} for several commonly encountered metric spaces. For the Wasserstein space, the minimization problem simplifies to \eqref{eq:v_hat:wass}, as shown in the proof of Proposition \ref{prop:wass}, using standard properties of the $L^2(0,1)$ inner product. The resulting minimizer is the projection of $\hat{B}(\bm{z})$ as per \eqref{eq:B:wass} onto the Wasserstein space. Due to the convexity and closedness of the Wasserstein space, this projection is guaranteed to exist and be unique for any $\bm{z}$. The \texttt{osqp} package \citep{stel:20} can be used to compute this projection, as detailed in subsection 5.1 of \citet{zhou:23}. 

Similarly, for the space of graph Laplacians, the minimization problem in \eqref{eq:lfrhat} simplifies to \eqref{eq:v_hat:mat}, as discussed in the proof of Proposition \ref{prop:mat}. Here, the minimizer $\hat{v}_h(\bm{z})$ is the projection of $\hat{B}(\bm{z})$ as per \eqref{eq:B:mat} onto the space of graph Laplacians, with existence and uniqueness guaranteed by the convexity and closedness of this space. The \texttt{osqp} package \citep{stel:20} is also applicable for this projection, as outlined in subsection 5.1 of \citet{mull:22:11}. 

In the space of $m$-dimensional covariance matrices with bounded diagonal entries, including the special case of correlation matrices, the projection of $\hat{B}(\bm{z})$ is similarly obtained. The \texttt{nearPD} function from the \texttt{Matrix} package in R provides a practical tool for projecting onto the nearest symmetric positive semi-definite matrix.

\subsection{Proof of Theorem \ref{thm:DNN}}
To establish the convergence of the estimation error, \citet{schm:20} applied an oracle-type inequality in Theorem 2, supported by Lemma 4. In contrast, we handle dependent sub-Gaussian noise with a vanishing bias, rather than the independent Gaussian noise with mean 0 assumed by \citet{schm:20}. Furthermore, we examine both out-of-sample and in-sample errors, since the output of the deep neural networks serves as the input for the local \f regression. Understanding the in-sample error is key for analyzing the convergence rate of the local \f regression.

\begin{lem}\label{lem:sub-gaussian}
    Let $\xi_1, \ldots, \xi_M$ be sub-Gaussian random variables with mean 0 and a common parameter $C$, then, as $M \rightarrow \infty$,  \[E(\max_{l=1,\ldots,M }\xi_l^2)\lesssim \log M.\]
\end{lem}
\begin{proof}
By the sub-Gaussian assumption, we have $P(|\xi_l|\geq \sqrt{t})\leq 2 e^{-t/C^2}$. For any $u>0$, we can apply the union bound and observe the following:
\begin{align*}
    E(\max_{l=1,\ldots,M} \xi_l^2) &\leq u + \int_{u}^\infty P(\max_{l=1,\ldots,M} \xi_l^2 \geq t) dt\\
    &\leq u + \sum_{l=1}^M\int_{u}^\infty P(|\xi_l| \geq \sqrt{t}) dt\\
    &\leq u + 2M\int_{u}^\infty e^{-t/C^2}dt\\& = u + 2MC^2e^{-u/C^2}.
\end{align*}
By selecting $u\asymp\log M$, we conclude that 
\[E(\max_{l=1,\ldots,M} \xi_l^2)\lesssim \log M.\]
\end{proof}

\begin{lem}\label{lem:lemma4 extension}
Consider the $p$-variate nonparametric regression model \eqref{eq:our DNN} with unknown regression function $g_{0j}$. Let $\tilde{g}_j$ be any estimator taking values in $\mathcal{G}=\mathcal{G}(L, s, \mathbf{p}, D)$ as per \eqref{eq:DNN Class} and $N_\delta=N(\delta,\mathcal{G},\|\cdot\|_\infty)$ be the covering number of $\mathcal{G}$ using balls of size $\delta$. Define
$$\Delta_n= E[\frac{1}{n}\sum_{i=1}^n \{Z_{ij} - \tilde{g}_j(\bm{X}_i)\}^2 - \inf_{g\in\mathcal{G}}\frac{1}{n}\sum_{i=1}^n\{Z_{ij}-g(\bm{X}_i)\}^2]$$
and assume $\|g_{0j}\|_\infty\leq D $. Then for any $i=1, \ldots, n$ and $j=1, \ldots, r$, we have
\begin{equation}
\label{eq:lemma4}
        E[\{\tilde{g}_j(\bm{X}_i) - g_{0j}(\bm{X}_i)\}^2]\lesssim \inf_{g\in\mathcal{G}}E[\{g(\bm{X})-g_{0j}(\bm{X})\}^2] + \frac{\log N_\delta }{n} + \Delta_n + \delta + |u_{nj}|
\end{equation}
and
\begin{equation}
\label{eq:lemma4 predicted}
        E[\{\tilde{g}_j(\bm{X}) - g_{0j}(\bm{X})\}^2]\lesssim \inf_{g\in\mathcal{G}}E[\{g(\bm{X})-g_{0j}(\bm{X})\}^2] + \frac{\log N_\delta }{n} + \Delta_n + \delta + |u_{nj}|,
\end{equation}
with $\bm{X}$ being an independent copy of $\bm{X}_1$.
\end{lem}
\begin{proof} 
Define $\|g\|^2_n = \frac{1}{n}\sum_{k=1}^n g(\bm{X}_k)^2$. Since $E[\{\tilde{g}_j(\bm{X}_i) - g_{0j}(\bm{X}_i)\}^2]\leq 4 D^2$ and $E[\{\tilde{g}_j(\bm{X}) - g_{0j}(\bm{X})\}^2]\leq 4 D^2$, the inequality in \eqref{eq:lemma4} and \eqref{eq:lemma4 predicted} hold trivially if $\log N_\delta \geq n$. In the following, we focus on the case $\log N_\delta < n$.

Similar to Lemma 4 of \citet{schm:20}, the proof is divided into three parts which are denoted by (I)--(III).
\begin{enumerate}[label=(\Roman*)]
    \item Relate the in-sample error $E[\{\tilde{g}_j(\bm{X}_i) - g_{0j}(\bm{X}_i)\}^2]$ and predicted error $E[\{\tilde{g}_j(\bm{X}) - g_{0j}(\bm{X})\}^2]$ to $E(\|\tilde{g}_j-g_{0j}\|_n^2)$ via the inequality
\begin{equation}
\label{eq:link fitted}
    E[\{\tilde{g}_j(\bm{X}_i) - g_{0j}(\bm{X}_i)\}^2] \lesssim E(\|\tilde{g}_j-g_{0j}\|_n^2) + \frac{\log N_\delta}{n} + \delta .
\end{equation}
\begin{equation}
\label{eq:link predcited}
    E[\{\tilde{g}_j(\bm{X}) - g_{0j}(\bm{X})\}^2] \lesssim E(\|\tilde{g}_j-g_{0j}\|_n^2) + \frac{\log N_\delta}{n} + \delta .
\end{equation}

    \item For any estimator $\tilde{g}_j$ taking values in $\mathcal{G}$,
\begin{equation*}
        | E\{ \frac{1}{n}   \sum_{k=1}^n \epsilon_{kj} \pi_j(\bm{X}_1,\ldots,\bm{X}_n) \tilde{g}_j(\bm{X}_k)\}| \lesssim \sqrt{\frac{E(\|\tilde{g}_j-g_{0j}\|_n^2)  \log N_\delta}{n}} + \delta
\end{equation*}

    \item Show that
\begin{equation*}
    E(\|\tilde{g}_j-g_{0j}\|_n^2) \lesssim \inf _{g \in \mathcal{G}} E[\{g(\bm{X})-g_{0j}(\bm{X})\}^2]+\frac{\log N_\delta}{n}+\Delta_n + \delta + |u_{nj}|.
\end{equation*}

\end{enumerate}
The inequality in \eqref{eq:lemma4} and \eqref{eq:lemma4 predicted} follow from (I) and (III).

(I): Since the inequality \eqref{eq:link predcited} has been shown by Lemma 4(I) in \citet{schm:20}, we only need to show the inequality \eqref{eq:link fitted}. Given a minimal $\delta$-covering of $\mathcal{G}$, denote the centers of the balls by $g_l$. By construction, there exists a $l^*$ such that $\|\tilde{g}-g_{l^*}\|_{\infty} \leq \delta$. Without loss of generality, one can assume that $\|g_l\|_{\infty}\leq D$ by the definition of $\mathcal{G}$. It follows that 
\begin{align*}
& \bigg|E[\{\tilde{g}_j(\bm{X}_i)-g_{0j}(\bm{X}_i)\}^2] - E\{\|\tilde{g}_j-g_{0j}\|_n^2\}\bigg|\\
&= \bigg|E[\frac{1}{n}\sum_{k=1}^n\{\tilde{g}_j(\bm{X}_i)-g_{0j}(\bm{X}_i)\}^2 - \{\tilde{g}_j(\bm{X}_k)-g_{0j}(\bm{X}_k)\}^2]\bigg|\\
& = \bigg|\frac{1}{n}\sum_{k=1}^n E[\{\tilde{g}_j(\bm{X}_i)-g_{l^*}(\bm{X}_i)\}^2 + \{g_{l^*}(\bm{X}_i) - g_{0j}(\bm{X}_i)\}^2 \\
& + 2\{\tilde{g}_j(\bm{X}_i)-g_{l^*}(\bm{X}_i)\}\{g_{l^*}(\bm{X}_i) - g_{0j}(\bm{X}_i)\} - \{\tilde{g}_j(\bm{X}_k)-g_{l^*}(\bm{X}_k)\}^2 \\
& - \{g_{l^*}(\bm{X}_k) - g_{0j}(\bm{X}_k)\}^2 - 2\{\tilde{g}_j(\bm{X}_k)-g_{l^*}(\bm{X}_k)\}\{g_{l^*}(\bm{X}_k) - g_{0j}(\bm{X}_k)\}]\bigg|\\
&\lesssim E\{|\frac{1}{n} \sum_{k=1}^n f_{l^*}(\bm{X}_k, \bm{X}_i)| \} + \delta
\end{align*}
where $f_{l^*}(\bm{X}_k, \bm{X}_i)=\{g_{l^*}(\bm{X}_i)-g_{0j}(\bm{X}_i)\}^2-\{g_{l^*}(\bm{X}_k)-g_{0j}(\bm{X}_k)\}^2$. Furthermore, define $R_{l^*}=D\sqrt{n^{-1} \log N_\delta} \vee \sqrt{E[\{g_{l^*}(\bm{X}_i)-g_{0j}(\bm{X}_i)\}^2]}$. Using $\|g_{l^*}-\tilde{g}_j\|_{\infty} \leq \delta$ and the triangle inequality, we have $R_{l^*}\leq D\sqrt{n^{-1} \log N_\delta}+\sqrt{E[\{\tilde{g}_{j}(\bm{X}_i)-g_{0j}(\bm{X}_i)\}^2]} +\delta$. Similarly, we can define $f_l$ and $R_l$ in the same way with $l^*$ being replaced by $l$.

Set $U=\sqrt{E[\{\tilde{g}_{j}(\bm{X}_i)-g_{0j}(\bm{X}_i)\}^2]}$ and $T=\max_{l=1,\ldots,N_\delta}|\sum_{i=1}^n f_l(\bm{X}_k, \bm{X}_i) /(R_l D)|$, we have
\begin{align*}
    &\bigg|E[\{\tilde{g}_j(\bm{X}_i)-g_{0j}(\bm{X}_i)\}^2] - E[\|\tilde{g}_j-g_{0j}\|_n^2]\bigg| \\
    &\lesssim E\{|\frac{1}{n} \sum_{k=1}^n f_{l^*}(\bm{X}_k, \bm{X}_i)/(R_{l^*} D)| R_{l^*} \} + \delta \\
    & \lesssim \frac{1}{n}E\{T (U + \sqrt{n^{-1} \log N_\delta} + \delta) \} + \delta\\
    & \lesssim \frac{1}{n}\sqrt{E(T^2)} \sqrt{E(U^2)} + \frac{1}{n}(\sqrt{n^{-1} \log N_\delta} + \delta)E(T) + \delta,
\end{align*}
where the last inequality follows from the Cauchy-Schwarz inequality. Observe that $E[f_l(\bm{X}_k, \bm{X}_i)]=0,|f_l(\bm{X}_k, \bm{X}_i)| \leq 4 D^2$ and
\begin{align*}
\Var\{f_l(\bm{X}_k, \bm{X}_i)\} & =2  \Var[\{g_l(\bm{X}_i)-g_{0j}(\bm{X}_i)\}^2] \cdot \mathbf{1}_{k \ne i}\\
& \leq 2 E[\{g_l(\bm{X}_i)-g_{0j}(\bm{X}_i)\}^4] \\
& \leq 8 D^2 R_l^2 .
\end{align*}
Up to now, all the settings match Lemma 4(I) of \citet{schm:20}. Therefore, using the same technique, we can show that
\begin{equation*}
    E[\{\tilde{g}_j(\bm{X}_i) - g_{0j}(\bm{X}_i)\}^2] \lesssim E(\|\tilde{g}_j-g_{0j}\|_n^2) + \frac{\log N_\delta}{n} + \delta .
\end{equation*}

(II): For any estimator $\tilde{g}_j$ taking values in $\mathcal{G}$, using Jensen's inequality and Assumption \ref{itm:d1}, we have
\begin{align*}
&|E\{\sum_{k=1}^n \epsilon_{kj} \pi_{j}(\bm{X}_1,\ldots,\bm{X}_n)\{\tilde{g}_j(\bm{X}_k)-g_{l^{*}}(\bm{X}_k)\}\}| \\
&\leq C_{\pi_j}\delta \sum_{k=1}^n E(|\epsilon_{kj}|) \\
&\leq n \delta C_{\pi_j}\sqrt{E(\epsilon_{1j}^2)}\\
&\leq n \delta C_{\pi_j}.
\end{align*}

Since 
\begin{align*}
    &E\{\epsilon_{kj} \pi_j(\bm{X}_1,\ldots,\bm{X}_n) g_{0j}(\bm{X}_k)\}\\&=E[E\{\epsilon_{kj}\pi_j(\bm{X}_1,\ldots,\bm{X}_n) g_{0j}(\bm{X}_k) | \bm{X}_1, \ldots, \bm{X}_n\}]\\&=0,
\end{align*}
we also find
\begin{align}\label{eq:C.5}
& \bigg|E[\frac{1}{n} \sum_{k=1}^n \epsilon_{kj} \pi_{j}(\bm{X}_1,\ldots,\bm{X}_n) \tilde{g}_j(\bm{X}_k)]\bigg| \nonumber \\
& =\bigg|E[\frac{1}{n} \sum_{k=1}^n \epsilon_{kj}\pi_{j}(\bm{X}_1,\ldots,\bm{X}_n) \{\tilde{g}_j(\bm{X}_k)-g_{0j}(\bm{X}_k)\}]\bigg| \nonumber \\
& \leq\bigg|E[\frac{1}{n} \sum_{k=1}^n \epsilon_{kj}\pi_{j}(\bm{X}_1,\ldots,\bm{X}_n) \{\tilde{g}_j(\bm{X}_k)-g_{l^*}(\bm{X}_k) \} ]\bigg|\nonumber\\
&+ \bigg|E[\frac{1}{n} \sum_{k=1}^n \epsilon_{kj}\pi_{j}(\bm{X}_1,\ldots,\bm{X}_n) \{g_{l^*}(\bm{X}_k) - g_{0j}(\bm{X}_k)\}]\bigg| \nonumber \\
& \leq \delta C_{\pi_j}+\bigg|\frac{1}{\sqrt{n}}E[\frac{\sum_{k=1}^n \epsilon_{kj}\pi_{j}(\bm{X}_1,\ldots,\bm{X}_n) \{g_{l^*}(\bm{X}_k) - g_{0j}(\bm{X}_k)\}}{\sqrt{n}\|g_{l^*}-g_{0j}\|_n} \|g_{l^*}-g_{0j}\|_n ]\bigg| \nonumber \\
& \leq \delta C_{\pi_j}+\frac{1}{\sqrt{n}}E\{|\xi_{l^{*}}|(\|\tilde{g}_j-g_{0j}\|_n+\delta)\},
\end{align}
with
$$
\xi_l=\frac{\sum_{k=1}^n \epsilon_{kj}\pi_{j}(\bm{X}_1,\ldots,\bm{X}_n)\{g_l(\bm{X}_k)-g_{0j}(\bm{X}_k)\}}{\sqrt{n}\|g_l-g_{0j}\|_n} .
$$
Given $\bm{X}_1, \ldots, \bm{X}_n$, $\xi_l$ are sub-Gaussian with common parameter $C_{\pi_j}$. It follows from Lemma \ref{lem:sub-gaussian} that $E(\xi_{l^{*}}^2) \leq$ $E(\max_{l=1, \ldots, N_\delta} \xi_l^2) \lesssim \log N_\delta$. Using the Cauchy-Schwarz inequality,
\begin{equation}\label{eq:C.6}
    E\{|\xi_{l^{*}}|(\|\tilde{g}_j-g_{0j}\|_n+\delta)\} \lesssim \sqrt{\log N_\delta}\{\sqrt{E(\|\tilde{g}_j-g_{0j}\|_n^2)}+\delta\}.
\end{equation}
Together with \eqref{eq:C.5} and \eqref{eq:C.6}, the inequality in (II) follows.

(III): For any fixed $g \in \mathcal{G}$, it follows from the definition of $\Delta_n$ that
$$E[\frac{1}{n} \sum_{k=1}^n\{Z_{kj}-\tilde{g}_j(\bm{X}_k)\}^2] \leq E[\frac{1}{n} \sum_{k=1}^n\{Z_{kj}-g(\bm{X}_k)\}^2]+\Delta_n.$$ 
Using the inequality in (II), we have
\begin{align*}
E(\|\tilde{g}_j-g_{0j}\|_n^2) & = E[\frac{1}{n}\sum_{k=1}^n \{\tilde{g}_j(\bm{X}_k) - Z_{kj} + (\epsilon_{kj}+u_{nj})\pi_j(\bm{X}_1,\ldots,\bm{X}_n)\}^2]\\
&= E[\frac{1}{n}\sum_{k=1}^n \{\tilde{g}_j(\bm{X}_k) - Z_{kj}\}^2 + \frac{1}{n}\sum_{k=1}^n (\epsilon_{kj}+u_{nj})^2\pi_j^2(\bm{X}_1,\ldots,\bm{X}_n) \\
& + \frac{2}{n}\sum_{k=1}^n (\epsilon_{kj}+u_{nj})\pi_j(\bm{X}_1,\ldots,\bm{X}_n) \{\tilde{g}_j(\bm{X}_k) - Z_{kj}\}]\\
& \leq \Delta_n + E[\frac{1}{n}\sum_{k=1}^n \{g(\bm{X}_k) - Z_{kj}\}^2 + \frac{1}{n}\sum_{k=1}^n (\epsilon_{kj}+u_{nj})^2\pi_j^2(\bm{X}_1,\ldots,\bm{X}_n) \\
& + \frac{2}{n}\sum_{k=1}^n (\epsilon_{kj}+u_{nj})\pi_j(\bm{X}_1,\ldots,\bm{X}_n) \{\tilde{g}_j(\bm{X}_k) - Z_{kj}\}]\\
& \leq \Delta_n + E(\|g-g_{0j}\|_n^2) + E\{\frac{2}{n}\sum_{k=1}^n (\epsilon_{kj}+u_{nj})^2\pi_j^2(\bm{X}_1,\ldots,\bm{X}_n)\} \\
& + E[\frac{2}{n}\sum_{k=1}^n (\epsilon_{kj}+u_{nj})\pi_j(\bm{X}_1,\ldots,\bm{X}_n) \{\tilde{g}_j(\bm{X}_k) - Z_{kj} - g(\bm{X}_k) + g_{0j}(\bm{X}_k)\}]\\
& = \Delta_n +  E[\{g(\bm{X})-g_{0j}(\bm{X})\}^2] \\
& + E[\frac{2}{n}\sum_{k=1}^n (\epsilon_{kj}+u_{nj})\pi_j(\bm{X}_1,\ldots,\bm{X}_n) \{\tilde{g}_j(\bm{X}_k) - g(\bm{X}_k)\}]\\
& \lesssim \Delta_n + E[\{g(\bm{X})-g_{0j}(\bm{X})\}^2]\\
& + E\{\frac{1}{n} \sum_{k=1}^n \epsilon_{kj} \pi_j(\bm{X}_1,\ldots,\bm{X}_n) \tilde{g}_j(\bm{X}_k)\} + |u_{nj}| \\
& \lesssim \Delta_n +  E[\{g(\bm{X})-g_{0j}(\bm{X})\}^2]\\
& +\sqrt{\frac{E\{\|\tilde{g}_j-g_{0j}\|_n^2\} \log N_\delta}{n}} + \delta + |u_{nj}|. 
\end{align*}

Finally, applying (C.4) in \citet{schm:20}, the result follows.
\end{proof}

\begin{proof}[Proof of Theorem \ref{thm:DNN}]
Using Lemma \ref{lem:lemma4 extension} to replace Lemma 4 of \citet{schm:20}, the remaining proof follows the proof of Theorem 1 and Corollary 1 in \citet{schm:20}. Let $\tilde{g}_j = \hat{g}_j$ be the minimizer of empirical risk \eqref{eq:DNN Ln}, one has $\Delta_n = 0$. Then for any $i=1,\ldots,n$ and $j=1,\ldots,r$, we have
\begin{equation*}
    E[\{\hat{g}_j(\bm{X}_i)-g_{0j}(\bm{X}_i)\}^2] \lesssim |u_{nj}| + \kappa_n^2 \log^3 n
\end{equation*}
and
\begin{equation*}
    E[\{\hat{g}_j(\bm{X})-g_{0j}(\bm{X})\}^2] \lesssim |u_{nj}| + \kappa_n^2 \log^3 n.
\end{equation*}
The result follows.
\end{proof}

\subsection{Proof of Proposition \ref{prop:lfrpc}}
We present the elementary results of an auxiliary Lemma \ref{lem:muj tauj} and its proof, extending the well-known results from \citet{fan:96}. The key quantities of interest are $\mu_j = E\{K_h(\bm{Z}^0-\bm{z})(\bm{Z}^0-\bm{z})^{\oplus j}\}$, $\tau_j(y) = E\{K_h(\bm{Z}^0-\bm{z})(\bm{Z}^0-\bm{z})^{\oplus j}|Y=y\}$ and the estimators $\tilde\mu_j=\frac{1}{n}\sum_{i=1}^n K_h(\bm{Z}^0_i-\bm{z})(\bm{Z}^0_i-\bm{z})^{\oplus j}$, for $j=0,1,2$. These quantities are fundamental for local linear smoothing and relate to the bias and variance properties of local \f regression.

\begin{lem}\label{lem:muj tauj}
    Suppose \ref{itm:k1} and \ref{itm:p1} hold. Then,
\begin{align*}
        \mu_0 & = f_{\bm{Z}^0}(\bm{z}) K_{10} + O(h^2),\\
        \mu_1 & = h^2  K_{12}\frac{\partial f_{\bm{Z}^0}(\bm{z})}{\partial \bm{z}} + O(h^3 \mathbf{1}),\\
        \mu_2 & = h^2 f_{\bm{Z}^0}(\bm{z}) K_{12} + O(h^4 \mathbf{1}\mathbf{1}^\T),
\end{align*}
where $K_{1j} = \int_{\mathbb{R}^r} K(\bm{u}) \bm{u}^{\oplus j} d\bm{u}$ for $j=0,2$. Furthermore, we have
\[\tilde{\mu}_j=\mu_j+O_p\{(h^{2 j-r} n^{-1})^{1 / 2}\mathbf{1}^{\oplus j}\}\] and
\begin{align*}
        \tau_0(y) & = f_{\bm{Z}^0|Y}(\bm{z},y) K_{10} + O(h^2),\\
        \tau_1(y) & = h^2 K_{12} \frac{\partial f_{\bm{Z}^0|Y}(\bm{z},y)}{\partial \bm{z}} + O(h^3\mathbf{1}),\\
        \tau_2(y) & = h^2 f_{\bm{Z}^0|Y}(\bm{z},y) K_{12} + O(h^4 \mathbf{1}\mathbf{1}^\T),
\end{align*}
where the order term is uniform over $y \in \Omega$.
\end{lem}

\begin{proof}
By definition, we have
\begin{align*}
        \mu_j & = E\{K_h(\bm{Z}^0-\bm{z})(\bm{Z}^0-\bm{z})^{\oplus j}\}\\
        & = \int h^{-r}K\{H^{-1}(\bm{x}-\bm{z})\}  (\bm{x}-\bm{z})^{\oplus j} f_{\bm{Z}^0}(\bm{x}) d\bm{x},
\end{align*}
\begin{align*}
        \tau_j(y) & = E\{K_h(\bm{Z}^0-\bm{z})(\bm{Z}^0-\bm{z})^{\oplus j}|Y=y\}\\
        & = \int h^{-r}K\{H^{-1}(\bm{x}-\bm{z})\}  (\bm{x}-\bm{z})^{\oplus j} f_{\bm{Z}^0|Y}(\bm{x},y) d\bm{x}.
\end{align*}
The statements regarding $\mu_j$ and $\tau_j(y)$ follow from Assumptions \ref{itm:k1} and \ref{itm:p1} using a second-order Taylor expansion of the densities $f_{\bm{Z}^0}$ and $f_{\bm{Z}^0|Y}$,
\begin{align*}
        \mu_0 & = \int h^{-r}K\{H^{-1}(\bm{x}-\bm{z})\}f_{\bm{Z}^0}(\bm{x}) d\bm{x} \\ 
        & = \int K(\bm{u})f_{\bm{Z}^0}(\bm{z} + H\bm{u}) d\bm{u} \\
        & = f_{\bm{Z}^0}(\bm{z}) K_{10} + h (\frac{\partial f_{\bm{Z}^0}(\bm{z})}{\partial \bm{z}})^\T K_{11} + O(h^2)\\
        & = f_{\bm{Z}^0}(\bm{z}) K_{10} + O(h^2),
\end{align*}
\begin{align*}
        \mu_1 & = \int h^{-r}K\{H^{-1}(\bm{x}-\bm{z})\}  (\bm{x} - \bm{z}) f_{\bm{Z}^0}(\bm{x}) d\bm{x} \\ 
        & = h \int \bm{u} K(\bm{u}) f_{\bm{Z}^0}(\bm{z} + H\bm{u}) d\bm{u} \\
        & = h \{ f_{\bm{Z}^0}(\bm{z}) K_{11} + h  K_{12}\frac{\partial f_{\bm{Z}^0}(\bm{z})}{\partial \bm{z}} + O(h^2)\}\\
        & = h^2  K_{12}\frac{\partial f_{\bm{Z}^0}(\bm{z})}{\partial \bm{z}} + O(h^3 \mathbf{1})),
\end{align*}
\begin{align*}
        \mu_2 & = \int h^{-r}K\{H^{-1}(\bm{x}-\bm{z})\}  (\bm{x} - \bm{z})(\bm{x} - \bm{z})^\T f_{\bm{Z}^0}(\bm{x}) d\bm{x} \\ 
        & = h^2 \int \bm{u}\bm{u}^\T K(\bm{u}) f_{\bm{Z}^0}(\bm{z} + H\bm{u}) d\bm{u} \\
        & = h^2\{f_{\bm{Z}^0}(\bm{z}) K_{12} + h  \int \bm{u}\bm{u}^\T \bm{u}^\T\frac{\partial f_{\bm{Z}^0}(\bm{z})}{\partial \bm{z}} K(\bm{u})d\bm{u} + O(h^2 \mathbf{1}\mathbf{1}^\T)\}\\
        & = h^2 f_{\bm{Z}^0}(\bm{z}) K_{12} + O(h^4 \mathbf{1}\mathbf{1}^\T),
\end{align*}
\begin{align*}
        \tau_0(y) & = \int h^{-r}K\{H^{-1}(\bm{x}-\bm{z})\} f_{\bm{Z}^0|Y}(\bm{x},y) d\bm{x} \\
        & = \int K(\bm{u})f_{\bm{Z}^0|Y}(\bm{z} + H\bm{u},y) d\bm{u} \\
        & = f_{\bm{Z}^0|Y}(\bm{z},y) K_{10} + h (\frac{\partial f_{\bm{Z}^0|Y}(\bm{z},y)}{\partial \bm{z}})^\T K_{11} + O(h^2)\\
        & = f_{\bm{Z}^0|Y}(\bm{z},y) K_{10} + O(h^2),
\end{align*}
\begin{align*}
        \tau_1(y) & = \int h^{-r}K\{H^{-1}(\bm{x}-\bm{z})\}  (\bm{x}-\bm{z}) f_{\bm{Z}^0|Y}(\bm{x},y) d\bm{x} \\
        & = h \int \bm{u} K(\bm{u})f_{\bm{Z}^0|Y}(\bm{z} + H\bm{u},y) d\bm{u} \\
        & = h\{f_{\bm{Z}^0|Y}(\bm{z},y) K_{11} + h K_{12} \frac{\partial f_{\bm{Z}^0|Y}(\bm{z},y)}{\partial \bm{z}}  + O(h^2\mathbf{1})\}\\
        & = h^2 K_{12} \frac{\partial f_{\bm{Z}^0|Y}(\bm{z},y)}{\partial \bm{z}} + O(h^3\mathbf{1}),
\end{align*}
\begin{align*}
        \tau_2(y) & = \int h^{-r}K\{H^{-1}(\bm{x}-\bm{z})\}  (\bm{x} - \bm{z})(\bm{x} - \bm{z})^\T f_{\bm{Z}^0|Y}(\bm{x},y) d\bm{x} \\ 
        & = h^2 \int \bm{u}\bm{u}^\T K(\bm{u}) f_{\bm{Z}^0|Y}(\bm{z}+H\bm{u},y) d\bm{u} \\
        & = h^2\{f_{\bm{Z}^0|Y}(\bm{z},y) K_{12} + h  \int \bm{u}\bm{u}^\T \bm{u}^\T\frac{\partial f_{\bm{Z}^0|Y}(\bm{z},y)}{\partial \bm{z}} K(\bm{u})d\bm{u} + O(h^2 \mathbf{1}\mathbf{1}^\T)\}\\
        & = h^2 f_{\bm{Z}^0|Y}(\bm{z},y) K_{12} + O(h^4 \mathbf{1}\mathbf{1}^\T).
\end{align*}
Next, by definition
$$\tilde\mu_j=\frac{1}{n}\sum_{i=1}^n K_h(\bm{Z}^0_i-\bm{z})(\bm{Z}^0_i-\bm{z})^{\oplus j}.$$
Note that $E(\tilde{\mu}_j)=\mu_j$ and
\begin{equation*}
E\{K_h^2(\bm{Z}^0_i-\bm{z})(Z_{il}-z_l)^{2 j}\}=h^{2 j-r} \int K(\bm{u})^2 u_l^{2 j} f_{\bm{Z}^0}(\bm{z}+H\bm{u}) d \bm{u}=O(h^{2 j-r}),
\end{equation*}
where $Z_{il}$ is the $l$th element of the random vector $\bm{Z}^0_i$ and $z_l$ is the $l$th element of the vector $z$. It follows that $\Var(\tilde{\mu}_j)=O(h^{2 j-r} n^{-1} \mathbf{1}^{\oplus j})$, proving the result for the $\tilde{\mu}_j$.
\end{proof}

\begin{lem}\label{lem:vn-vh-op1}
    Suppose \ref{itm:k1} and \ref{itm:l1} hold and furthermore $h \to 0$, $n h^r \to \infty$. Then $d\{\tilde{v}_h(\bm{z}), v_h(\bm{z})\}=o_p(1)$.
\end{lem}

\begin{proof}
In the following,  $\rightsquigarrow$ denotes weak convergence and $l^{\infty}(\Omega)$ the space of bounded functions on $\Omega$. According to Corollary 3.2.3 in \citet{well:23} and Assumption \ref{itm:l1} (i), it is sufficient to demonstrate the convergence of $\sup_{y\in\Omega}|\tilde{Q}_h(y,\bm{z})-Q_h(y,\bm{z})|$ to zero in probability. Achieving this requires showing $Q_h(\cdot,\bm{z})-\tilde{Q}_h(\cdot,\bm{z}) \rightsquigarrow 0$ in $l^{\infty}(\Omega)$ and then applying Theorem 1.3.6 of \citet{well:23}. Finally, by Theorem 1.5.4 and 1.5.7 in \citet{well:23}, we can establish weak convergence by showing the following:

\begin{enumerate}[label=(\roman*)]
  \item $\tilde{Q}_h(y,\bm{z})-Q_h(y,\bm{z})=o_p(1)$ for all $y\in\Omega$ and\label{itm:i1}
  \item $\tilde{Q}_h(\cdot,\bm{z})-Q_h(\cdot,\bm{z})$ is asymptotically equicontinuous in probability, i.e., for all $\epsilon,\eta>0$, there exists $\delta>0$ such that
  $$\limsup_n P(\sup _{d(y_1, y_2)<\delta}|\{\tilde{Q}_h(y_1,\bm{z})-Q_h(y_1,\bm{z})\} - \{\tilde{Q}_h(y_2,\bm{z})-Q_h(y_2,\bm{z})\}|>\epsilon)<\eta.$$\label{itm:i2}
\end{enumerate}
To begin with \ref{itm:i1}, recall that 
\[w(\bm{Z}^0_i, \bm{z}, h)=\frac{1}{\sigma_0^2}K_h(\bm{Z}^0_i-\bm{z})[1-\mu_1^\T\mu_2^{-1}(\bm{Z}^0_i-\bm{z})]\]
and 
\[\tilde{w}(\bm{Z}^0_i, \bm{z}, h)=\frac{1}{\tilde\sigma_0^2} K_h(\bm{Z}^0_i-\bm{z})[1-\tilde\mu_1^\T\tilde\mu_2^{-1}(\bm{Z}^0_i-\bm{z})],\] 
where $\sigma_0^2 = \mu_0 - \mu_1^\T\mu_2^{-1}\mu_1$ and $\tilde\sigma_0^2 = \tilde\mu_0 - \tilde\mu_1^\T\tilde\mu_2^{-1}\tilde\mu_1$. Then, one observes that
\begin{align}\label{eq:Ln-Lh}
\tilde{Q}_h(y,\bm{z})-Q_h(y,\bm{z}) = &\frac{1}{n} \sum_{i=1}^n\{\tilde{w}(\bm{Z}^0_i, \bm{z}, h)-w(\bm{Z}^0_i, \bm{z}, h)\} d^2(Y_i, y)\ + \nonumber\\
&\frac{1}{n} \sum_{i=1}^n[w(\bm{Z}^0_i, \bm{z}, h) d^2(Y_i, y)-E\{w(\bm{Z}^0_i, \bm{z}, h) d^2(Y_i, y)\}].
\end{align}
Furthermore,  $\tilde{w}(\bm{Z}^0_i, \bm{z}, h)-w(\bm{Z}^0_i, \bm{z}, h)=W_{0 n} K_h(\bm{Z}^0_i-\bm{z})+K_h(\bm{Z}^0_i-\bm{z})W_{1 n}(\bm{Z}^0_i-\bm{z})$, where
\begin{equation}\label{eq:W0nW1n}
    W_{0 n}=\frac{1}{\tilde\sigma_0^2}-\frac{1}{\sigma_0^2}, \quad W_{1 n}=\frac{\tilde\mu_1^\T\tilde\mu_2^{-1}}{\tilde\sigma_0^2}-\frac{\mu_1^\T\mu_2^{-1}}{\sigma_0^2}.
\end{equation}
Using the results of Lemma \ref{lem:muj tauj} and the Sherman-Morrison formula, it follows that 
\begin{align*}
        &\tilde\mu_2^{-1}\tilde\mu_1 - \mu_2^{-1}\mu_1 \\
        =& [\mu_2+O_p\{(h^{4-r} n^{-1})^{1/2} \mathbf{1}\mathbf{1}^\T\}]^{-1}[\mu_1+O_p\{(h^{2-r} n^{-1})^{1/2} \mathbf{1}\}] - \mu_2^{-1}\mu_1 \\
        =& [\mu_2^{-1} - \frac{\mu_2^{-1}O_p\{(h^{4-r} n^{-1})^{1/2} \mathbf{1}\mathbf{1}^\T\}\mu_2^{-1}}{1 + \mathbf{1}^\T\mu_2^{-1}\mathbf{1}}]
        [\mu_1+O_p\{(h^{2-r} n^{-1})^{1/2} \mathbf{1}\}]- \mu_2^{-1}\mu_1 \\ 
        =& O_p\{(nh^{2+r})^{-1 / 2}\mathbf{1}\}.
\end{align*}
Furthermore, it is easy to get $\tilde\sigma^2_0 - \sigma_0^2=O_p\{(n h^r)^{-1 / 2}\}$. Therefore, $W_{0 n}=O_p\{(n h^r)^{-1 / 2}\}$ and $W_{1 n}=O_p\{(nh^{2+r})^{-1 / 2}\mathbf{1}\}$. Since for $j =0, 1$, we have
\begin{align*}
& E\{K_h(\bm{Z}^0_i-\bm{z})(Z_{il}-z_{l})^{j} d^2(Y_i, y)\}=O(h^j), \\
& E\{K_h^2(\bm{Z}^0_i-\bm{z})(Z_{il}-z_l)^{2 j} d^4(Y_i, y)\}=O(h^{2 j-r}),
\end{align*}
it follows that the first term in \eqref{eq:Ln-Lh} is $O_p((n h^r)^{-1 / 2})$. We find that $E\{w^2(\bm{Z}^0_i,\bm{z}, h)\}=O(h^{-r})$, so the second term in \eqref{eq:Ln-Lh} is $O_p\{(n h^r)^{-1 / 2}\}$ as well. Then, we show that $\tilde{Q}_h(y,\bm{z})-Q_h(y,\bm{z})=o_p(1)$ for any $y \in \Omega$ and any $\bm{z}\in \mathbb{R}^r$, since $n h^r \to \infty$. 

Moving on to \ref{itm:i2}, for any $y_1,y_2\in\Omega$,
\begin{align*}
&|\{\tilde{Q}_h(y_1,\bm{z})-Q_h(y_1,\bm{z})\} - \{\tilde{Q}_h(y_2,\bm{z})-Q_h(y_2,\bm{z})\}|\\
\leq & |\tilde{Q}_h(y_1,\bm{z})-\tilde{Q}_h(y_2,\bm{z})| + |Q_h(y_1,\bm{z})-Q_h(y_2,\bm{z})|\\
\leq & \frac{1}{n} \sum_{i=1}^n|\tilde{w}(\bm{Z}^0_i, \bm{z}, h)||d(Y_i, y_1)-d(Y_i, y_2)||d(Y_i, y_1)+d(Y_i, y_2)| \\
& + |E[w(\bm{Z}^0_i, \bm{z}, h)\{d(Y_i, y_1)+d(Y_i, y_2)\}\{d(Y_i, y_1)-d(Y_i, y_2)\}]|\\
\leq & 2 \text{diam}(\Omega) d(y_1,y_2) \frac{1}{n}\sum_{i=1}^n \{|\tilde{w}(\bm{Z}^0_i, \bm{z}, h)| + E|w(\bm{Z}^0_i, \bm{z}, h)|\} \\
=& O_p\{d(y_1,y_2)\},
\end{align*}
since $E\{|w(\bm{Z}^0_i, \bm{z}, h)|\}=O(1)$, $E\{w^2(\bm{Z}^0_i, \bm{z}, h)\}=O(h^{-1})$, and $n^{-1} \sum_{i=1}^n|\tilde{w}(\bm{Z}^0_i, \bm{z}, h)|=O_p(1)$. Therefore,
$$\sup_{d(y_1,y_2)<\delta}|\{\tilde{Q}_h(y_1,\bm{z})-Q_h(y_1,\bm{z})\} - \{\tilde{Q}_h(y_2,\bm{z})-Q_h(y_2,\bm{z})\}| = O_p(\delta),$$
which implies \ref{itm:i2} and then  $d\{\tilde{v}_h(\bm{z}),v_h(\bm{z})\} = o_p(1)$.
\end{proof}

\begin{proof}[Proof of Proposition \ref{prop:lfrpc}]
First w,e prove \eqref{lfrpc:bias}. Similar to the proof of Theorem 3 in \citet{mull:19:6}, we can show 
\[\frac{d F_{Y|\bm{Z}^0}(\bm{z},y)}{d F_{Y}(y)} = \frac{f_{\bm{Z}^0|Y}(\bm{z},y)}{f_{\bm{Z}^0}(\bm{z})}\]
for all $\bm{z}$ given that $f_{\bm{Z}^0}(\bm{z})>0$. Then by Lemma \ref{lem:muj tauj},
\begin{align*}
        \int w(\bm{x},\bm{z},h)d F_{\bm{Z}^0|Y}(\bm{x}|y) &= \frac{\tau_0(y) -\mu_1^\T \mu_2^{-1}\tau_1(y)}{\mu_0 -\mu_1^\T \mu_2^{-1}\mu_1} = \frac{f_{\bm{Z}^0|Y}(\bm{z},y)}{f_{\bm{Z}^0}(\bm{z})} + O(h^2), 
\end{align*}
where the error term is uniform over $y\in\Omega$. Hence,
\begin{align*}
        Q_h(\omega,\bm{z}) & = \int d^2(y,\omega)w(\bm{x},\bm{z},h)dF_{\bm{Z}^0,Y}(\bm{x},y) = \int w(\bm{x},\bm{z},h)dF_{\bm{Z}^0|Y}(\bm{x},y) d^2(y, \omega) dF_Y(y)\\
        &= \int \frac{f_{\bm{Z}^0|Y}(\bm{z},y)}{f_{\bm{Z}^0}(\bm{z})}  d^2(y, \omega) dF_Y(y) + O(h^2) = \int d^2(y, \omega) d F_{Y|\bm{Z}^0}(\bm{z},y) + O(h^2)\\
        &= Q(\omega, \bm{z}) + O(h^2),
\end{align*}
where the error term is now uniform over $\omega\in\Omega$. By the first two parts of Assumption \ref{itm:l1} and employing proof by contradiction, we then have $d\{v_h(\bm{z}), v(\bm{z})\}=o(1)$ as $h=h_n\to 0$.

Next, define $r_h=h^{-\frac{\gamma_1}{\gamma_1-1}}$ and set 
\[S_{j,h}(\bm{z})=\left\{y: 2^{j-1} < r_h d\{y, v(\bm{z})\}^{\gamma_1/2}\leq 2^{j}\right\}.\]
Let $I$ denote the indicator function. Then, for any $M>0$, there exists $a>0$ such that, for large $n$,
\begin{align*}
I[r_h d\{v_h(\bm{z}), v(\bm{z})\}^{\gamma_1/2}>2^M] & = \sum_{j \geq M} I\{v_h(\bm{z})\in S_{j,h}(\bm{z})\}\\
& \leq \sum_{j \geq M} I\big(\sup4788_{y\in S_{j,h}(\bm{z})} [Q_h\{v(\bm{z}),\bm{z}\} - Q_h(y,\bm{z})] \geq 0\big).
\end{align*}
By the first two parts of Assumption \ref{itm:l3}, for every $j$ involved in the sum,  we have for every $y\in S_{j,h}(\bm{z})$,
$$Q\{v(\bm{z}),\bm{z}\} - Q(y,\bm{z}) \lesssim -d\{v(\bm{z}),y\}^{\gamma_1}\lesssim -\frac{2^{2(j-1)}}{r_h^2}.$$
Defining  $V(y,\bm{z}) = Q_h(y,\bm{z}) - Q(y,\bm{z})$, 
\begin{align*}
    & \sup_{y\in S_{j,h}(\bm{z})}| V(y,\bm{z}) - V\{v(\bm{z}),\bm{z}\}| \\
    \geq &  \sup_{y\in S_{j,h}(\bm{z})}[ Q_h\{v(\bm{z}),\bm{z}\} - Q_h(y,\bm{z})] -\sup_{y\in S_{j,h}(\bm{z})}[ Q\{v(\bm{z}),\bm{z}\} - Q(y,\bm{z})].
\end{align*}
Hence,
\begin{align*}
I[r_h d\{v_h(\bm{z}), v(\bm{z})\}^{\gamma_1/2}>2^M] &\leq \sum_{j \geq M} I[\sup_{y\in S_{j,h}(\bm{z})} |V(y,\bm{z}) - V\{v(\bm{z}),\bm{z}\}| \geq \frac{2^{2(j-1)}}{r_h^2}]\\
&\leq \sum_{j \geq M} I\{a h^2 (\frac{2^j}{r_h})^{2/\gamma_1}\geq \frac{2^{2(j-1)}}{r_h^2}\}\\
&\leq \sum_{j \geq M} I\{4a h^2 (\frac{2^{2j/\gamma_1-2j}}{r_h^{2/\gamma_1-2}})\geq 1\}\\
& \leq 4a \sum_{j \geq M} \frac{2^{2 j(1-\gamma_1) / \gamma_1}}{r_h^{2(1-\gamma_1) / \gamma_1} h^{-2}} \leq 4a \sum_{j \geq M}\left(\frac{1}{4^{(\gamma_1-1) / \gamma_1}}\right)^j,
\end{align*}
which converges since $\gamma_1>1$. Thus, for some $M>0$, we have
$$
d\{v_h(\bm{z}), v(\bm{z})\}\leq 2^{2M/\gamma_1}  h^{2/(\gamma_1-1)}
$$
for large $n$ and hence
$$
d\{v_h(\bm{z}), v(\bm{z})\}=O(h^{2/(\gamma_1-1)})
$$
for large $n$, implying \eqref{lfrpc:bias}.

Next we show \eqref{lfrpc:variance}. Define $T_n(y,\bm{z})=\tilde{Q}_h(y,\bm{z})-Q_h(y, \bm{z})$. Letting
$$
D_i(y, \bm{z})=d^2(Y_i, y)-d^2\{Y_i, v_h(\bm{z})\},
$$
we have
\begin{align}\label{eq:Tn}
|T_n(y,\bm{z})-T_n\{v_h(\bm{z}),\bm{z}\}| \leq&\left|\frac{1}{n} \sum_{i=1}^n\{\tilde{w}(\bm{Z}^0_i, \bm{z}, h)-w(\bm{Z}^0_i, \bm{z}, h)\} D_i(y, \bm{z})\right|  \nonumber\\
& +\left|\frac{1}{n} \sum_{i=1}^n[w(\bm{Z}^0_i, \bm{z}, h) D_i(y, \bm{z})-E\{w(\bm{Z}^0_i, \bm{z}, h) D_i(y, \bm{z})\}]\right|.
\end{align}

Since $W_{0 n}$ and $W_{1 n}$ from \eqref{eq:W0nW1n} are $O_p\{(n h^r)^{-1/2}\}$ and $O_p\{(n h^{2+r})^{-1 / 2}\mathbf{1}\}$, respectively, and using the fact that $\left|D_i(y, \bm{z})\right| \leq 2 \text{diam}(\Omega) d(y, v_h(\bm{z}))$, the first term on the right-hand side of \eqref{eq:Tn} is $O_p[(n h^r)^{-1/2} d\{y, v_h(\bm{z})\}]$, where the $O_p$ term is independent of $\Omega$ and $v_h(\bm{z})$. Thus, we can define
$$
B_R=\left\{\sup _{d(y, v_h(\bm{z}))<\delta}\left|\frac{1}{n} \sum_{i=1}^n\{\tilde{w}(\bm{Z}^0_i, \bm{z}, h)-w(\bm{Z}^0_i, \bm{z}, h)\} D_i(y, \bm{z})\right| \leq R \delta(n h^r)^{-1 / 2}\right\}
$$
for $R>0$, so that $P\left(B_R^c\right) \to 0$.

Next, to control the second term on the right-hand side of \eqref{eq:Tn}, define the functions $g_y: \mathbb{R}^r \times \Omega \mapsto \mathbb{R}$ by
$$
g_y(x, w)=\frac{1}{\sigma_0^2} K_h(\bm{x}-\bm{z})\{1-\mu_1^\T\mu_2^{-1}(\bm{x}-\bm{z})\} d^2(y, w),
$$
and the corresponding function class
$$
\mathcal{F}_{n \delta}=\{g_y-g_{v_h(\bm{z})}: d\{y, v_h(\bm{z})\}<\delta\}.
$$
An envelope function for $\mathcal{F}_{n \delta}$ is
$$
F_{n \delta}(\bm{z})=\frac{2 \text{diam}(\Omega) \delta}{\sigma_0^2} K_h(\bm{z}-\bm{x})\left|1-\mu_1^\T\mu_2^{-1}(\bm{z}-\bm{x})\right|,
$$
where  $E\{F_{n \delta}^2(\bm{Z}^0)\}=O(\delta^2 h^{-r})$. According to Theorem 2.7.17 of \citet{well:23} and Assumption \ref{itm:l2}, the bracketing integral of the class $\mathcal{F}_{n\delta}$ is 
$$J_{[]}(\epsilon, \mathcal{F}_{n\delta},\|\cdot\|_{\bm{Z}^0,Y,2})=\int_0^\epsilon \sqrt{1+\log N_{[]}(t\{E(F_{n\delta}^2)\}^{1/2}, \mathcal{F}_{n\delta},\|\cdot\|_{\bm{Z}^0,Y,2})} d t = O(1), $$
where $N_{[]}(t\{E(F_{n\delta}^2)\}^{1/2}, \mathcal{F}_{n\delta},\|\cdot\|_{\bm{Z}^0,Y,2})$ denotes the minimum number of $t\{E(F_{n\delta}^2)\}^{1/2}$-brackets needed to cover the function class $\mathcal{F}_{n\delta}$ (see Definition 2.1.6 of \citet{well:23}) and $\|\cdot\|_{\bm{Z}^0,Y,2}$ denotes the norm that for any $f\in\mathcal{F}_{n\delta}$, $\|f\|_{\bm{Z}^0,Y,2}=[E\{f(\bm{Z}^0,Y)^2\}]^{1/2}$. Using these facts together with Theorems 2.14.16 of \citet{well:23} and Assumption \ref{itm:l2}, for small $\delta$,
\begin{align*}
    &E\left(\sup _{d\{y, v_h(\bm{z})\}<\delta}\left|\frac{1}{n}\sum_{i=1}^n w(\bm{Z}^0_i, \bm{z}, h) D_i(y, \bm{z})-E\{w(\bm{Z}^0_i, \bm{z}, h) D_i(y, \bm{z})\}\right|\right)\\
    &\leq\frac{J_{[]}(\epsilon, \mathcal{F}_{n\delta},\|\cdot\|_{\bm{Z}^0,Y,2})  \left[E\left\{F_{n \delta}^2(\bm{Z}^0)\right\}\right]^{1/2}}{\sqrt{n}}\\
    &=O(\delta(n h^r)^{-1 / 2}).
\end{align*}
Combining this with \eqref{eq:Tn} and the definition of $B_R$,
\begin{equation}\label{eq:ETn}
    E\left(I_{B_R} \sup _{d\{y, v_h(\bm{z})\}<\delta}\left|T_n(y,\bm{z})-T_n\{v_h(\bm{z}),\bm{z}\}\right|\right) \leq \frac{a \delta}{(n h^r)^{1 / 2}},
\end{equation}
where $I_{B_R}$ is the indicator function for the set $B_R$ and $a$ is a constant depending on $R$ and the entropy integral in Assumption \ref{itm:l2}.

To finish, set $r_n=(n h^r)^{\frac{\gamma_2}{4(\gamma_2-1)}}$ and define
$$
S_{j, n}(\bm{z})=\{y: 2^{j-1}< r_n d\{y, v_h(\bm{z})\}^{\gamma_2/2} \leq 2^{j}\} .
$$
Choosing $\eta_2$ to satisfy Assumption \ref{itm:l3} (ii) and such that Assumption \ref{itm:l2} is satisfied for any $\delta<\eta_2$, set $\eta=(\eta_2 / 2)^{\gamma_2 / 2}$. Denote events $A_{\eta_2} = \{d\{\tilde{v}_h(\bm{z}),v_h(\bm{z})\}>\eta_2/2\}$ and $C_M = \{r_n d\{\tilde{v}_h(\bm{z}),v_h(\bm{z})\}^{\gamma_2/2}>2^M)\}$. If $r_n d\{\tilde{v}_h(\bm{z}),v_h(\bm{z})\}^{\gamma_2/2}$ is larger than $2^M$ for a given integer $M$, then $\tilde{v}_h(\bm{z})$ is in one of the shells $S_{j,n}(\bm{z})$ with $j>M$ such that the supremum of the map $y\mapsto \tilde{Q}_h\{v_h(\bm{z}),\bm{z}\} - \tilde{Q}_h(y,\bm{z})$ over this shell is nonnegative by the property of $\tilde{v}_h(\bm{z})$.  Observe 
\begin{align}
    P(C_M) \leq & P(B_R^c) + P(A_{\eta_2}) + P(B_R\cap A_{\eta_2}^c\cap C_M)\nonumber\\
    \leq & P(B_R^c) + P(A_{\eta_2}) \nonumber + \sum_{\substack{j\geq M\\ 2^j\leq r_n \eta}}P\left(\left\{\sup_{y\in S_{j,n}(\bm{z})} \left| \tilde{Q}_h\{v_h(\bm{z}), \bm{z}\} - \tilde{Q}_h(y, \bm{z})\geq 0 \right|\right\}\cap B_R\right),
\end{align}
where the first term is seen to converge to 0  and the second term goes to zero by Lemma \ref{lem:vn-vh-op1}. Choosing $\eta_2>0$ small enough such that Assumption \ref{itm:l3} (ii) holds for every $d\{y,v_h(\bm{z})\}\leq \eta$ and \eqref{eq:ETn} holds for every $\delta\leq\eta$,  for every $j$ involved in the sum, we have, for every $y\in S_{j,n}(\bm{z})$,
$$Q_h(v_h(\bm{z}),\bm{z}) - Q_h(y,\bm{z})\lesssim - d\{v_h(\bm{z}),y\}^{\gamma_2} \lesssim - \frac{2^{2(j-1)}}{r_n^2}.$$
In terms of $T_n(y,\bm{z}) = \tilde{Q}_h(y,\bm{z}) - Q_h(y,\bm{z})$, by Markov's inequality, we have
\begin{align*}
    P(C_M) \leq & P(B_R^c) + P(A_{\eta_2}) \nonumber \\
    & + \sum_{\substack{j\geq M\\ 2^j\leq r_n \eta}}P\left(\left\{\sup_{y\in S_{j,n}} \left| T_n(y,\bm{z}) - T_n\{v_h(\bm{z}),\bm{z}\}\right|\gtrsim\frac{2^{2(j-1)}}{r_n^2}\right\}\cap B_R\right)\nonumber\\
    \lesssim & \sum_{\substack{j\geq M\\ 2^j\leq r_n \eta}}\frac{(2^{j}/r_n)^{2/\gamma_2} (nh^r)^{-1/2}}{2^{2j}/r_n^2} \nonumber\\
    \lesssim &\sum_{j\geq M} \frac{2^{2j(1-\gamma_2)/\gamma_2}}{r_n^{2(1-\gamma_2)/\gamma_2}\sqrt{nh^r}}\nonumber\\
     = &\sum_{j\geq M} \left(\frac{1}{4^{(\gamma_2-1)/\gamma_2}}\right)^j,
\end{align*}
which converges since $\gamma_2>1$. Hence, as $M\to \infty$,
$$P(d\{\tilde{v}_h(\bm{z}),v_h(\bm{z})\}>(\frac{2^M}{r_n})^{2/\gamma_2})\to 0.$$
Therefore,
$$d\{\tilde{v}_h(\bm{z}),v_h(\bm{z})\} = O_p\{(n h^r)^{-1/\{2(\gamma_2-1)\}}\}.$$
\end{proof}

\subsection{Proof of Proposition \ref{prop:eiv}}
\begin{lem}
\label{lem:wt}
Suppose \ref{itm:k1} holds. If $\|\hat{\bm{Z}}_i-\bm{Z}^0_i\|=O_p(\zeta_n)$ for $i=1, \ldots, n$, $\|\hat{\bm{Z}}-\bm{Z}^0\|=O_p(\zeta_n)$, $nh^{r}\rightarrow\infty$, and $h^{-r-2}\zeta_n\to0$, then the weight functions as defined in \eqref{eq:wtilde} and \eqref{eq:what} satisfy
\[\frac{1}{n}\sum_{i=1}^n\{\hat{w}(\hat{\bm{Z}}_i, \hat{\bm{Z}}, h)-\tilde{w}(\bm{Z}_i^0, \bm{Z}^0, h)\}=O_p(h^{-r-1}\zeta_n).\]
\end{lem}
\begin{proof}
Recall that 
\[\hat{w}(\hat{\bm{Z}}_i, \hat{\bm{Z}}, h)=\frac{1}{\hat{\sigma}^2}K_h(\hat{\bm{Z}}_i-\hat{\bm{Z}})[1-\hat{\mu}_1^\T\hat{\mu}_2^{-1}(\hat{\bm{Z}}_i-\hat{\bm{Z}})],\]
where 
\[\hat{\sigma}^2=\hat{\mu}_0-\hat{\mu}_1^\T\hat{\mu}_2^{-1}\hat{\mu}_1,\quad\hat{\mu}_j=\frac{1}{n}\sum_{i=1}^nK_h(\hat{\bm{Z}}_i-\hat{\bm{Z}})(\hat{\bm{Z}}_i-\hat{\bm{Z}})^{\oplus j},\, j=0, 1, 2,\]
and 
\[\tilde{w}(\bm{Z}_i^0, \bm{Z}^0, h)=\frac{1}{\tilde{\sigma}^2}K_h(\bm{Z}^0_i-\bm{Z}^0)[1-\tilde{\mu}_1^\T\tilde{\mu}_2^{-1}(\bm{Z}^0_i-\bm{Z}^0)],\]
where 
\[\tilde{\sigma}^2=\tilde{\mu}_0-\tilde{\mu}_1^\T\tilde{\mu}_2^{-1}\tilde{\mu}_1,\quad\tilde{\mu}_j=\frac{1}{n}\sum_{i=1}^nK_h(\bm{Z}^0_i-\bm{Z}^0)(\bm{Z}^0_i-\bm{Z}^0)^{\oplus j},\,  j=0, 1, 2,\]
with $K_h(\bm{z})=h^{-r}K\{H^{-1}\bm{z}\}$. 

Define
\[\rho_j(\bm{x}, \bm{z})=K_h(\bm{x}-\bm{z})(\bm{x}-\bm{z})^{\oplus j},\quad \bm{x}, \bm{z}\in\mathcal{T}^r.\]
By Assumption \ref{itm:k1}, $K(\cdot)$ is Lipschitz continuous. For any $\bm{x}_1, \bm{x}_2, \bm{z}_1, \bm{z}_2\in\mathcal{T}^r$, one has
\begin{align*}
|\rho_0(\bm{x}_1, \bm{z}_1)-\rho_0(\bm{x}_2, \bm{z}_2)|&=|K_h(\bm{x}_1-\bm{z}_1)-K_h(\bm{x}_2-\bm{z}_2)|\\&\lesssim h^{-r-1}\|\bm{x}_1-\bm{z}_1-(\bm{x}_2-\bm{z}_2)\|\\&\leq h^{-r-1}(\|\bm{x}_1-\bm{x}_2\|+\|\bm{z}_1-\bm{z}_2\|).
\end{align*}
For $j=1, 2$, it hold that
\begin{align*}
    \|\rho_j(\bm{x}_1, \bm{z}_1)-\rho_j(\bm{x}_2, \bm{z}_2)\|&=\|K_h(\bm{x}_1-\bm{z}_1)(\bm{x}_1-\bm{z}_1)^{\oplus j}-K_h(\bm{x}_2-\bm{z}_2)(\bm{x}_2-\bm{z}_2)^{\oplus j}\|\\&\leq\|K_h(\bm{x}_1-\bm{z}_1)(\bm{x}_1-\bm{z}_1)^{\oplus j}-K_h(\bm{x}_2-\bm{z}_2)(\bm{x}_1-\bm{z}_1)^{\oplus j}\|+\\&\hspace{1.4em}\|K_h(\bm{x}_2-\bm{z}_2)(\bm{x}_1-\bm{z}_1)^{\oplus j}-K_h(\bm{x}_2-\bm{z}_2)(\bm{x}_2-\bm{z}_2)^{\oplus j}\|\\&\leq\|(\bm{x}_1-\bm{z}_1)^{\oplus j}\|\cdot|K_h(\bm{x}_1-\bm{z}_1)-K_h(\bm{x}_2-\bm{z}_2)|+\\&\hspace{1.4em}K_h(\bm{x}_2-\bm{z}_2)\cdot\|(\bm{x}_1-\bm{z}_1)^{\oplus j}-(\bm{x}_2-\bm{z}_2)^{\oplus j}\|.
\end{align*}

Now using the Lipschitz continuity of $K(\cdot)$ by Assumption \ref{itm:k1} and the fact that $K_h(\bm{x}-\bm{z})=0$ if any component of $|\bm{x}-\bm{z}|$ exceeds $h$, for $j=1, 2$ one has
\begin{align*}
    \|\rho_j(\bm{x}_1, \bm{z}_1)-\rho_j(\bm{x}_2, \bm{z}_2)\|&\lesssim h^{j-r-1}\|\bm{x}_1-\bm{z}_1-(\bm{x}_2-\bm{z}_2)\|+\\&\hspace{1.3em}h^{j-1}\|\bm{x}_1-\bm{z}_1-(\bm{x}_2-\bm{z}_2)\|\\&\lesssim h^{-r-1+j}\|\bm{x}_1-\bm{z}_1-(\bm{x}_2-\bm{z}_2)\|\\&\leq h^{-r-1+j}(\|\bm{x}_1-\bm{x}_2\|+\|\bm{z}_1-\bm{z}_2\|),
\end{align*}
which implies
\begin{align*}
    \|\hat{\mu}_j-\tilde{\mu}_j\|&=\|\frac{1}{n}\sum_{i=1}^n\rho_j(\hat{\bm{Z}}_i, \hat{\bm{Z}})-\frac{1}{n}\sum_{i=1}^n\rho_j(\bm{Z}^0_i, \bm{Z}^0)\|\\&\leq\frac{1}{n}\sum_{i=1}^n\|\rho_j(\hat{\bm{Z}}_i, \hat{\bm{Z}})-\rho_j(\bm{Z}^0_i, \bm{Z}^0)\|\\&\lesssim\frac{1}{n}\sum_{i=1}^nh^{-r-1+j}(\|\hat{\bm{Z}}_i-\bm{Z}^0_i\|+\|\hat{\bm{Z}}-\bm{Z}^0\|).
\end{align*}
We conclude that $\|\hat{\mu}_j-\tilde{\mu}_j\|=O_p(h^{-r-1+j}\zeta_n)$ for $j=0, 1, 2$. Since $nh^r\to\infty$, by Lemma \ref{lem:muj tauj} one has $\tilde{\mu}_0=O_p(1), \tilde{\mu}_1=O_p(h^2)$, and $\tilde{\mu}_2=O_p(h^2)$. It follows that 
\begin{align*}
    |\hat{\sigma}^2-\tilde{\sigma}^2|&=|(\hat{\mu}_0-\hat{\mu}_1^\T\hat{\mu}_2^{-1}\hat{\mu}_1)-(\tilde{\mu}_0-\tilde{\mu}_1^\T\tilde{\mu}_2^{-1}\tilde{\mu}_1)|\\&=O_p(h^{-r-1}\zeta_n)
\end{align*}
and 
\[\|\frac{\hat{\mu}_1^\T\hat{\mu}_2^{-1}}{\hat{\sigma}^2}-\frac{\tilde{\mu}_1^\T\tilde{\mu}_2^{-1}}{\tilde{\sigma}^2}\|=O_p(h^{-r-2}\zeta_n).\]
Observe that 
\begin{align*}
    \hat{w}(\hat{\bm{Z}}_i, \hat{\bm{Z}}, h)-\tilde{w}(\bm{Z}_i^0, \bm{Z}^0, h)&=\{\hat{\sigma}^{-2}-\tilde{\sigma}^{-2}\}\rho_0(\bm{Z}^0_i, \bm{Z}^0)+\hat{\sigma}^{-2}\{\rho_0(\hat{\bm{Z}}_i, \hat{\bm{Z}})-\rho_0(\bm{Z}^0_i, \bm{Z}^0)\}-\\
    &\hspace{1.5em}\{\frac{\hat{\mu}_1^\T\hat{\mu}_2^{-1}}{\hat{\sigma}^2}-\frac{\tilde{\mu}_1^\T\tilde{\mu}_2^{-1}}{\tilde{\sigma}^2}\}\rho_1(\bm{Z}^0_i, \bm{Z}^0)-\\
    &\hspace{1.5em}\frac{\hat{\mu}_1^\T\hat{\mu}_2^{-1}}{\hat{\sigma}^2}\{\rho_1(\hat{\bm{Z}}_i, \hat{\bm{Z}})-\rho_1(\bm{Z}^0_i, \bm{Z}^0)\}
\end{align*}
and $E\{\rho_j(\bm{Z}^0_i, \bm{Z}^0)\}=O(h^j\mathbf{1}^{\oplus j}), E[\{\rho_j(\bm{Z}^0_i, \bm{Z}^0)\}^{\oplus 2}]=O(h^{2j-r}\mathbf{1}^{\oplus 2j})$ for $j=0, 1$. 
We have
\[\frac{1}{n}\sum_{i=1}^n\{\hat{w}(\hat{\bm{Z}}_i, \hat{\bm{Z}}, h)-\tilde{w}(\bm{Z}_i^0, \bm{Z}^0, h)\}=O_p(h^{-r-1}\zeta_n).\]
\end{proof}

\begin{lem}
    \label{lem:eivop1}
    Suppose \ref{itm:k1} and the last two parts of \ref{itm:l1}. If $\|\hat{\bm{Z}}_i-\bm{Z}^0_i\|=O_p(\zeta_n)$ for $i=1, \ldots, n$, $\|\hat{\bm{Z}}-\bm{Z}^0\|=O_p(\zeta_n)$, $nh^r\to\infty$, and $h^{-r-2}\zeta_n\to0$, then 
    \[d\{\hat{v}_h(\hat{\bm{Z}}), \tilde{v}_h(\bm{Z}^0)\}=o_p(1).\]
\end{lem}
\begin{proof} 
For any $y\in\Omega$, observe that
\begin{align*}
    \hat{Q}_h(y, \hat{\bm{Z}})-\tilde{Q}_h(y, \bm{Z}^0)&=\frac{1}{n}\sum_{i=1}^n\hat{w}(\hat{\bm{Z}}_i, \hat{\bm{Z}}, h)d^2(Y_i, y)-\frac{1}{n}\sum_{i=1}^n\tilde{w}(\bm{Z}_i^0, \bm{Z}^0, h)d^2(Y_i, y)\\
    &=\frac{1}{n}\sum_{i=1}^n\{\hat{w}(\hat{\bm{Z}}_i, \hat{\bm{Z}}, h)-\tilde{w}(\bm{Z}_i^0, \bm{Z}^0, h)\}d^2(Y_i, y)\\
    &=O_p(h^{-r-1}\zeta_n)
\end{align*}
by Lemma \ref{lem:wt} and the fact that $\Omega$ is totally bounded. By the last two parts of Assumption \ref{itm:l1} and employing proof by contradiction, we conclude that $d\{\hat{v}_h(\hat{\bm{Z}}), \tilde{v}_h(\bm{Z}^0)\}=o_p(1)$.
\end{proof}

\begin{proof}[Proof of Proposition 2]
Define $V_n(y)=\hat{Q}_h(y, \hat{\bm{Z}})-\tilde{Q}_h(y, \bm{Z}^0)$. Letting $D_i(y, \bm{Z}^0)=d^2(Y_i, y)-d^2\{Y_i, \tilde{v}_h(\bm{Z}^0)\}$, we have
\begin{equation}
    \label{eq:vn}
    |V_n(y)-V_n\{\tilde{v}_h(\bm{Z}^0)\}|=|\frac{1}{n}\sum_{i=1}^n\{\hat{w}(\hat{\bm{Z}}_i, \hat{\bm{Z}}, h)-\tilde{w}(\bm{Z}_i^0, \bm{Z}^0, h)\}D_i(y, \bm{Z}^0)|.
\end{equation}
By Lemma \ref{lem:wt}, and using the fact that $|D_i(y, \bm{Z}^0)|\leq2\mathrm{diam}(\Omega)d\{y, \tilde{v}_h(\bm{Z}^0)\}$, the right-hand side of \eqref{eq:vn} is $O_p[d\{y, \tilde{v}_h(\bm{Z}^0)\} h^{-r-1}\zeta_n]$. Thus, we can define
\[B_R=\{\sup_{d\{y, \tilde{v}_h(\bm{Z}^0)\}<\delta}|V_n(y)-V_n\{\tilde{v}_h(\bm{Z}^0)\}|\leq R\delta h^{-r-1}\zeta_n\}\]
for $R>0$, so that $P(B_R^C)\to0$.

Set $t_n=(h^{-r-1}\zeta_n)^{-\gamma_3/2(\gamma_3-1)}$. For each $n$, the metric space $\Omega$, excepting  the point $\tilde{v}_h(\bm{Z}^0)$,  can be partitioned into the ``shells'', 
\[S_{j, n}(\bm{Z}^0)=\{y: 2^{j-1}<t_nd\{y, \tilde{v}_h(\bm{Z}^0)\}^{\gamma_3/2}\leq2^j\},\]
with $j$ varying over the integers. Choose $\eta=(\eta_3/2)^{\gamma_3/2}>0$ small enough that Assumption \ref{itm:l3} (iii) holds and $\delta<\eta_3$. If $t_nd\{\hat{v}_h(\hat{\bm{Z}}), \tilde{v}_h(\bm{Z}^0)\}^{\gamma_3/2}$ is larger than $2^M$ for a given integer $M$, then $\hat{v}_h(\hat{\bm{Z}})$ is in one of the shells $S_{j, n}(\bm{Z}^0)$ with $j>M$. In that case the supremum of the map $y\mapsto\hat{Q}_h(\tilde{v}_h(\bm{Z}^0), \hat{\bm{Z}})-\hat{Q}_h(y, \hat{\bm{Z}})$ over this shell is nonnegative by the optimality of $\hat{v}_h(\hat{\bm{Z}})$. Thus for every $\eta>0$,
\begin{align}
    &P(t_nd\{\hat{v}_h(\hat{\bm{Z}}), \tilde{v}_h(\bm{Z}^0)\}^{\gamma_3/2}>2^M)\nonumber\\
    \leq &\,\, P(B_R^C)+P(2d\{\hat{v}_h(\hat{\bm{Z}}), \tilde{v}_h(\bm{Z}^0)\}>\eta_3)\label{eq:ptn}\\
    & + \sum_{\substack{j>M\\2^j\leq\eta t_n}}P(\{\sup_{y\in S_{j, n}(\bm{Z}^0)}[\hat{Q}_h\{\tilde{v}_h(\bm{Z}^0), \hat{\bm{Z}}\}-\hat{Q}_h(y, \hat{\bm{Z}})]\geq0\}\cap B_R),\nonumber
\end{align}
where $P(B_R^C)\to0$ as discussed previously and the second term goes to zero by Lemma \ref{lem:eivop1}. Then for every $j$ involved in the sum, by Assumption \ref{itm:l3} (iii) we have, for every $y\in S_{j, n}(\bm{Z}^0)$,
\[\tilde{Q}_h\{\tilde{v}_h(\bm{Z}^0), \bm{Z}^0\}-\tilde{Q}_h(y, \bm{Z}^0)\lesssim-\frac{2^{2(j-1)}}{t_n^2}.\]
In terms of the centered process $V_n$, by Markov's inequality the right-hand side of \eqref{eq:ptn} may be bounded by
\begin{align*}
&\sum_{\substack{j>M\\2^j\leq\eta t_n}}P(\{\sup_{y\in S_{j, n}(\bm{Z}^0)}[V_n(y)-V_n\{\tilde{v}_h(\bm{Z}^0)\}]\gtrsim\frac{2^{2(j-1)}}{t_n^2}\}\cap B_R)\\&\lesssim\sum_{j>M}\frac{t_n^2}{2^{2j}}(\frac{2^j}{t_n})^{2/\gamma_3} h^{-r-1}\zeta_n\\&=\sum_{j>M}(\frac{1}{4^{(\gamma_3-1)/\gamma_3}})^{j},
\end{align*}
which converges to zero for every $M=M_n\to\infty$ since $\gamma_3>1$. Hence
\[d\{\hat{v}_h(\hat{\bm{Z}}), \tilde{v}_h(\bm{Z}^0)\}=O_p(t_n^{-2/\gamma_3})=O_p\{(h^{-r-1}\zeta_n)^{1/(\gamma_3-1)}\}.\]
\end{proof}

\subsection{Proof of Theorem \ref{thm:pc}}
\begin{proof}[Proof of Theorem \ref{thm:pc}]
Combining Theorems, \ref{thm:DNN} and Proposition \ref{prop:lfrpc}, and \ref{prop:eiv}, it follows from the triangle inequality that
\begin{align*}
    d\{\hat{m}(\bm{X}), m(\bm{X})\} & = d\{\hat{v}_h(\hat{\bm{Z}}), v(\bm{Z}^0)\} \\
    & \leq d\{\tilde{v}_h(\bm{Z}^0), v(\bm{Z}^0)\}+d\{\hat{v}_h(\hat{\bm{Z}}), \tilde{v}_h(\bm{Z}^0)\}\\
    & = O_p\{h^{2/(\gamma_1-1)}+(nh^r)^{-1/\{2(\gamma_2-1)\}}+(h^{-r-1}\zeta_n)^{1/(\gamma_3-1)}\},
\end{align*}
where $\bm{Z}^0 = \bm{g}_0(\bm{X})$ and $\hat{\bm{Z}} = \hat{\bm{g}}(\bm{X})$.
\end{proof}

\subsection{Additional Proof}\label{subsec:remark2}
If a supremum rate of the type   $E(\sup_{i=1,\ldots,n}\|\hat{\bm{Z}}_i-\bm{Z}^0_i\|)=O(\tau_n)$ is available, one can substitute Lemma \ref{lem:wt} with Lemma \ref{lem:rmk} below. This substitution would lead to an improvement in the convergence rate outlined in Proposition \ref{prop:eiv}, resulting in the rate  $(h^{-1}\tau_n)^{1/(\gamma_3-1)}$.  An improved convergence rate would also be achieved in Theorem \ref{thm:pc}, with the new rate  $h^{2/(\gamma_1-1)}+(nh^r)^{-1/{2(\gamma_2-1)}}+(h^{-1}\tau_n)^{1/(\gamma_3-1)}$.

\begin{lem}
\label{lem:rmk}
Suppose \ref{itm:k1} and \ref{itm:p1} hold. If $E(\sup_{i=1,\ldots,n}\|\hat{\bm{Z}}_i-\bm{Z}^0_i\|)=O(\tau_n)$, $nh^{r}\rightarrow\infty$ and $h^{-r-1}\tau_n\to0$, then the weight functions as defined in \eqref{eq:wtilde} and \eqref{eq:what} satisfy
\[\frac{1}{n}\sum_{i=1}^n\{\hat{w}(\hat{\bm{Z}}_i, \hat{\bm{Z}}, h)-\tilde{w}(\bm{Z}_i^0, \bm{Z}^0, h)\}=O_p(h^{-1}\tau_n).\]
\end{lem}

\begin{proof}
Recall that in Lemma \ref{lem:wt}, for any $\bm{x}_1, \bm{x}_2, \bm{z}_1, \bm{z}_2\in\mathcal{T}^r$, one has
\begin{equation*}
    \|\rho_j(\bm{x}_1, \bm{z}_1)-\rho_j(\bm{x}_2, \bm{z}_2)\| \lesssim h^{-r-1+j}(\|\bm{x}_1-\bm{x}_2\|+\|\bm{z}_1-\bm{z}_2\|).
\end{equation*}
One thus has
\begin{align*}
    \|\hat{\mu}_j-\tilde{\mu}_j\|&=\|\frac{1}{n}\sum_{i=1}^n\rho_j(\hat{\bm{Z}}_i, \hat{\bm{Z}})-\frac{1}{n}\sum_{i=1}^n\rho_j(\bm{Z}^0_i, \bm{Z}^0)\|\\
    &\leq \frac{1}{n}\sum_{i=1}^n\|\rho_j(\hat{\bm{Z}}_i, \hat{\bm{Z}})-\rho_j(\bm{Z}^0_i, \bm{Z}^0)\|\cdot[\mathbf{I}\{\rho_j(\hat{\bm{Z}}_i, \hat{\bm{Z}})\ne 0\} + \mathbf{I}\{\rho_j(\bm{Z}^0_i, \bm{Z}^{0})\ne 0\}] \\
    &\leq\frac{1}{n h^{r+1-j}} ( \sup_{i=1,\ldots,n}\|\hat{\bm{Z}}_i-\bm{Z}_i^0\| + \|\hat{\bm{Z}}-\bm{Z}^0\|)\\
    & \hspace{5em} \sum_{i=1}^n[\mathbf{I}(\|\hat{\bm{Z}}_i-\hat{\bm{Z}}\|\leq h) + \mathbf{I}(\|\bm{Z}_i^0-\bm{Z}^0\|\leq h)]\\
    &\lesssim\frac{\tau_n}{n h^{r+1-j}}\sum_{i=1}^n [\mathbf{I}(\|\hat{\bm{Z}}_i-\hat{\bm{Z}}\|\leq h) + \mathbf{I}(\|\bm{Z}_i^0-\bm{Z}^0\|\leq h) ].
\end{align*}
Given $\bm{Z}^0$, observe 
\[\frac{1}{nh^r}\sum_{i=1}^n \mathbf{I}(\|\bm{Z}_i^0-\bm{Z}^0\|\leq h)\stackrel{P}{\rightarrow} c_1 f_{\bm{Z}^0}(\bm{Z}^0)\]
for some constant $c_1<\infty$, and that
\begin{align*}
    &\frac{1}{nh^r}\sum_{i=1}^n\mathbf{I}(\|\hat{\bm{Z}}_i-\hat{\bm{Z}}\|\leq h)\\
    \leq & \frac{1}{nh^r}\sum_{i=1}^n\mathbf{I}(\|\bm{Z}^0_i-\bm{Z}^0\| \leq 3h) + \mathbf{I}(\|\hat{\bm{Z}}_i-\bm{Z}_i^0\| \geq h) + \mathbf{I}(\|\hat{\bm{Z}} -\bm{Z}^0\|\geq h)
\end{align*}
converges to $c_2 f(\bm{Z}_i^0)$ in probability for some constant $c_2<\infty$. Here
\[\frac{1}{n h^{r}}\sum_{i=1}^n\mathbf{I}(\|\hat{\bm{Z}}_i-\bm{Z}_i^0\| \geq h)\text{ and }\frac{1}{nh^{r}}\sum_{i=1}^n\mathbf{I}(\|\hat{\bm{Z}} - \bm{Z}^0\|\geq h)\] are both $o_p(1)$ since by the Markov inequality and Theorem \ref{thm:DNN},
\begin{align*}
    &E\{\frac{1}{n h^{r}}\sum_{i=1}^n\mathbf{I}(\|\hat{\bm{Z}}_i-\bm{Z}_i^0\| \geq h)\}\leq \frac{1}{n h^{r+1}}\sum_{i=1}^n E\|\hat{\bm{Z}}_i-\bm{Z}_i^0\|\leq \frac{1}{n}\sum_{i=1}^n h^{-r-1} \tau_n \rightarrow 0,\\
    &E\{\frac{1}{n h^{r}}\sum_{i=1}^n\mathbf{I}(\|\hat{\bm{Z}} - \bm{Z}^0\| \geq h)\}\leq \frac{1}{n h^{r+1}}\sum_{i=1}^n E\|\hat{\bm{Z}} - \bm{Z}^0\|\leq \frac{1}{n}\sum_{i=1}^n h^{-r-1} \tau_n \rightarrow 0.
\end{align*}
This implies that $\|\hat{\mu}_j-\tilde{\mu}_j\|=O_p(h^{j-1}\tau_n)$ for $j=0, 1, 2$. It follows that 
\begin{equation*}
    |\hat{\sigma}^2-\tilde{\sigma}^2|=O_p(h^{-1}\tau_n) \quad\text{and}\quad \|\frac{\hat{\mu}_1^\T\hat{\mu}_2^{-1}}{\hat{\sigma}^2}-\frac{\tilde{\mu}_1^\T\tilde{\mu}_2^{-1}}{\tilde{\sigma}^2}\|=O_p(h^{-2}\tau_n).
\end{equation*}
Thus, due to $E\{\rho_j(\bm{Z}^0_i, \bm{Z}^0)\}=O(h^j\mathbf{1}^{\oplus j})$ and $E[\{\rho_j(\bm{Z}^0_i, \bm{Z}^0)\}^{\oplus 2}]=O(h^{2j-r}\mathbf{1}^{\oplus 2j})$ for $j=0, 1$, we have
\[\frac{1}{n}\sum_{i=1}^n\{\hat{w}(\hat{\bm{Z}}_i, \hat{\bm{Z}}, h)-\tilde{w}(\bm{Z}_i^0, \bm{Z}^0, h)\}=O_p(h^{-1}\tau_n).\]
\end{proof}

\subsection{Simulations for Networks}
\label{supp:network}
We consider graph Laplacians of undirected weighted networks with a fixed number of nodes $m$ and bounded edge weights as a second type of metric space-valued responses; see Example \ref{exm:mat}. The space of graph Laplacians \citep{mull:22:11} is
\begin{equation}\label{eq:laplacian}
    \Omega=\{Y=(y_{i j}): Y=Y^{\mathrm{T}} ; Y 1_m=0_m ;-W \leq y_{i j} \leq 0 \text{ for }i \neq j\},
\end{equation}
where $1_m$ and $0_m$ are the $m$-vectors of ones and zeroes, respectively. We aim to construct a generative model that produces random graph Laplacians $Y$ along with an Euclidean predictor $\bm{X}\in\mathbb{R}^9$. 

Denote the half vectorization excluding the diagonal of a symmetric and centered matrix by vech, with inverse operation $\text{vech}^{-1}$. By the symmetry and centrality of graph Laplacians as per \eqref{eq:laplacian}, every graph Laplacian $Y$ is fully determined by its upper (or lower) triangular part, which can be further vectorized into $\text{vech}(Y)$, a vector of length $d=m(m-1) / 2$. The true regression function can thus be defined as
$$m(\bm{X})=\text{vech}^{-1}[\{-E(l_1|\bm{X}), \ldots, -E(l_d|\bm{X})\}^\T].$$

To generate the random response $Y$, each entry of the random vector $(l_1, \ldots, l_d)^\T$ is first generated using beta distributions whose parameters depend on the predictor $\bm{X}$. The random response $Y$ is then generated conditional on $\bm{X}$ through an inverse half vectorization $\text{vech}^{-1}$ applied to $(-l_1, \ldots,-l_d)^{\T}$. The space of graph Laplacians $\Omega$ is not a vector space. Instead, it is a bounded, closed, and convex subset in $\mathbb{R}^{m^2}$ of dimension $m(m-1)/2$. To ensure that the random response $Y$ generated in simulations resides in $L_m$, the off-diagonal entries $l_j$, $j = 1, \ldots, d$, need to be nonpositive and bounded below as per \eqref{eq:laplacian}. To this end, we consider
\begin{align*}
    &Y=\text{vech}^{-1}(-l_1, \ldots,-l_d), \text{ where } l_j \stackrel{\mathrm{i.i.d.}}{\sim} \text{Beta}(\alpha_1, \alpha_2),\\
    &\alpha_1 = X_8\sin(\pi X_1) +(1-X_8)\cos(\pi X_2),\\
    &\alpha_2 = X_4^2 X_7+X_5^2(1-X_7),
\end{align*}
where components of the predictor $\bm{X}\in \mathbb{R}^9$ are distributed as
\begin{align*}
        &X_1 \sim U(0,1),\ X_2 \sim U(-1/2,1/2),\ X_3 \sim U(1,2)\\
        &X_4 \sim N(0,1),\ X_5 \sim N(0,1),\ X_6 \sim N(5,5)\\
        &X_7 \sim \mathrm{Ber}(0.4),\ X_8 \sim \mathrm{Ber}(0.3),\ X_9 \sim \mathrm{Ber}(0.6).
\end{align*}
The true regression function is thus
\[m(\bm{X})=\text{vech}^{-1}(-\alpha_1/(\alpha_1 + \alpha_2), \ldots, -\alpha_1/(\alpha_1 + \alpha_2)).\]

\begin{figure}[t]
\single
    \centering
    \includegraphics[width=0.8\linewidth]{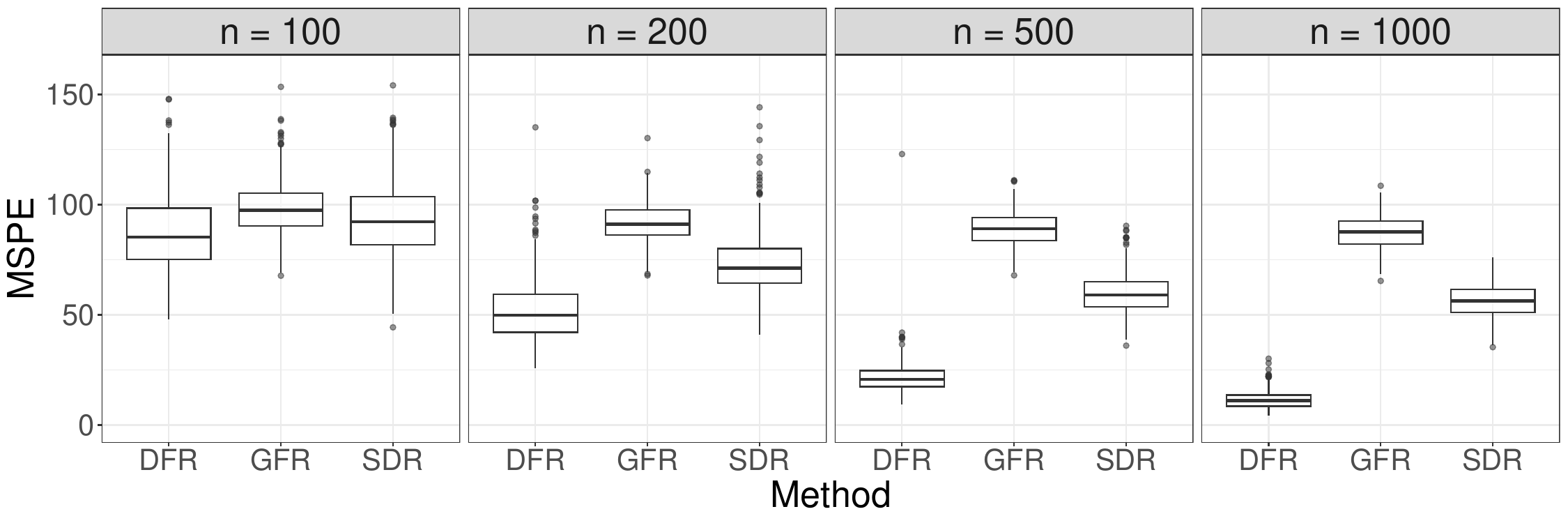}
   \caption{Boxplot of mean squared prediction errors for $Q=500$ Monte Carlo runs using deep \f regression (DFR), global \f regression (GFR) \citep{mull:22:11} and sufficient dimension reduction (SDR) \citep{zhan:21:1} for network responses.}
   \label{fig:simu2}
\end{figure}

Similar to simulations for the distributional data, we investigated sample sizes $n = 100, 200, 500, 1000$, with $Q=500$ Monte Carlo runs and set $r = 2$ as the dimension of low-dimensional representations of random objects. We compare the proposed DFR model with global \f regression (GFR) and sufficient dimension reduction (SDR). Figure \ref{fig:simu2} displays the MSPE for various sample sizes across all Monte Carlo runs, while Table \ref{tab:simu2} summarizes their averages. We observe that MSPE decreases with increasing sample size, highlighting the convergence of the proposed methods. Both boxplots and the table illustrate the superior performance of the proposed DFR model over GFR and SDR, even with small sample sizes. Despite the simulated network data not residing on a 2-dimensional manifold, the DFR model exhibits remarkable robustness as the sample size increases. Indeed, as the sample size grows, the DFR model proves to be a superior method for handling multivariate predictors and network-valued responses in this simulation example.

\begin{table}[tb]
\single
\centering
\caption{Average mean squared prediction error of deep \f regression (DFR), global \f regression (GFR) \citep{mull:22:11} and sufficient dimension reduction (SDR) \citep{zhan:21:1} for network responses.}
\label{tab:simu2}
\begin{tabular}{c|ccc}
\hline
$n$ & DFR & GFR & SDR\\
\hline
100 & 88.025 & 97.994 & 94.843\\
200 & 52.404 & 91.872 & 73.035\\
500 & 21.486 & 88.940 & 59.609\\
1000 & 11.701 & 87.416 & 56.257\\
\hline
\end{tabular}
\end{table}

\subsection{Comparison of Different Manifold Learning Methods}
\label{supp:sim:manifold}
In this section, we assess the performance of the DFR model across various mainstream manifold learning methods, including t-SNE \citep{van:08}, UMAP \citep{mcin:18}, Laplacian eigenmaps (LE) \citep{belk:03}, and diffusion maps (DM) \citep{coif:06}. The goal is to evaluate the impact of these methods on the predictive performance of the DFR model.

The simulations follow the same generative models and settings as described in subsection \ref{sec:simu:dist} of the main text. In each case, we replaced the ISOMAP step in the DFR model with one of the alternative manifold learning methods, while keeping the rest of the model unchanged. This consistent setup ensures a fair comparison of the effectiveness of each manifold learning method. 

For t-SNE, we used the \texttt{Rtsne} function from the \texttt{Rtsne} package in R and for UMAP the \texttt{umap} function from the \texttt{uwot} package in R. Both functions accept a distance matrix as input and directly output the low-dimensional representation. No modifications were made to the code, and all parameters were set to their default values. For LE and DM, no R functions natively accept a distance matrix as input, requiring us to adapt existing algorithms. For LE, we referred to the algorithm described in \citet{ghoj:23} and used the \texttt{do.llle} function from the \texttt{Rdimtools} package as a reference. For DM, we referred to the algorithm in \citet{ghoj:23} and adapted the \texttt{diffuse} function from the \texttt{diffusionMap} package. Similarly, all parameters for LE and DM were set to their default values as specified in the reference implementations.

Tables~\ref{tab:simu3} and \ref{tab:simu3:median} summarize the average mean squared prediction errors (AMSPE) and median mean squared prediction errors (MMSPE), respectively, for different manifold learning methods. Both tables show that ISOMAP consistently achieves the lowest AMSPE and MMSPE across all sample sizes. The alternative manifold learning methods did not perform as well as ISOMAP in our simulations. In Table~\ref{tab:simu3}, we observe that t-SNE and LE fail to converge when $n = 1000$. These convergence issues arise due to the presence of outliers in the simulation runs: t-SNE has 3 outliers out of 500 runs and LE has 1. Excluding these outliers, the results demonstrate convergence for all methods. Table~\ref{tab:simu3:median} further supports this observation, as MMSPE values exclude the impact of outliers, showcasing consistent performance across different methods. 

\begin{table}[tb]
\single
\centering
\caption{Average mean squared prediction error of deep \f regression using various manifold learning techniques, including ISOMAP \citep{tene:00}, t-SNE \citep{van:08}, UMAP \citep{mcin:18}, Laplacian eigenmaps (LE) \citep{belk:03} and diffusion maps (DM) \citep{coif:06}.}
\label{tab:simu3}
\begin{tabular}{c|ccccc}
\hline
$n$& ISOMAP & tSNE & UMAP & LE & DM\\
\hline
100 & 34.000 & 59.814 & 52.939 & 85.936 & 37.819\\
200 & 20.976 & 27.387 & 27.983 & 26.932 & 26.545\\
500 & 12.544 & 14.473 & 15.743 & 16.822 & 22.070\\
1000 & 7.874 & 23.845 & 12.367 & 18.194 & 16.113\\
\hline
\end{tabular}
\end{table}

\begin{table}[tb]
\single
\centering
\caption{Median mean squared prediction error of deep \f regression using various manifold learning techniques, including ISOMAP \citep{tene:00}, t-SNE \citep{van:08}, UMAP \citep{mcin:18}, Laplacian eigenmaps (LE) \citep{belk:03} and diffusion maps (DM) \citep{coif:06}.}
\label{tab:simu3:median}
\begin{tabular}{c|ccccc}
\hline
$n$ & ISOMAP & tSNE & UMAP & LE & DM\\
\hline
100 & 26.792 & 30.773 & 38.193 & 36.015 & 29.772\\
200 & 16.182 & 17.010 & 21.715 & 21.960 & 18.063\\
500 & 8.363 & 9.790 & 11.018 & 12.635 & 9.931\\
1000 & 4.674 & 6.228 & 5.998 & 8.707 & 6.298\\
\hline
\end{tabular}
\end{table}

\subsection{Robustness to Violations of Manifold Assumption}
To evaluate the robustness of DFR model under potential violations of the low-dimensional manifold assumption, we conducted additional simulations introducing varying levels of perturbation. The generative model is similar to that described in Section \ref{sec:simu:dist}, but includes a perturbation parameter $\nu$ to alter the mean of the true probability distribution, creating deviations from a perfect low-dimensional manifold.

Consider sample sizes $n=100,200,500,1000$, with $Q=500$ Monte Carlo runs. In each run, predictors $\bm{X}_i\in\mathbb{R}^{12}$ are generated independently as
\begin{align*}
        &X_{i1} \sim U(-1,0),\ X_{i2} \sim U(0,1),\ X_{i3} \sim U(1,2)\\
        &X_{i4} \sim N(1,1),\ X_{i5} \sim N(-1,1),\ X_{i6} \sim N(0,3)\\
        &X_{i7} \sim \mathrm{Bernoulli}(0.6),\ X_{i8} \sim \mathrm{Bernoulli}(0.7),\ X_{i9} \sim \mathrm{Bernoulli}(0.3)\\
        &X_{i10}\sim \mathrm{Beta}(2,2), X_{i11} \sim \mathrm{Beta}(3,2), X_{i12}\sim \mathrm{Beta}(2,3).
\end{align*}
For each $i$, the mean $\eta_i$ and standard deviation $\sigma_i$ of the Gaussian distribution are generated conditionally on $\bm{X}_i$ as
\begin{align*}
         &\eta_i|\bm{X}_i\sim N(\mu, \nu^2) \text{ where }\\
         &\mu=3 X_{i10} + 3 X_{i8}\{\cos (\pi X_{i1})+\sin(\pi X_{i2})\}  + X_{i7}(X_{i4}+X_{i5}) ,\\
         & \nu \in \{0,0.5,1,1.5,2\},\\
         &\sigma_i= 3X_{i11} + 2 X_{i8}\{ \sin (\pi X_{i1}) +\cos(\pi X_{i2})\} + X_{i7}(X_{i4}^2 + X_{i5}^2).
\end{align*}
The corresponding distributional response is constructed as $Y_i = \eta_i + \sigma_i \Phi^{-1}$. Increasing the perturbation parameter $\nu$ leads to greater deviation from the low-dimensional manifold, with $\nu=0$ representing a perfect manifold structure. To mimic real-world scenarios where direct observations of the underlying distribution are unavailable, we sampled 100 observations $\{y_{ij}\}_{j=1}^{100}$ from each distribution $Y_i$ and use the empirical measure in lieu of the true response \citep{zhou:23}.

We compared the proposed deep \f regression (DFR) against global \f regression (GFR) \citep{mull:19:6} and sufficient dimension reduction (SDR) \citep{zhan:21:1}. The quality of the estimated regression function was assessed using the average mean squared prediction error (AMSPE), averaged over $Q=500$ Monte Carlo runs. Table \ref{tab:simu5} summarizes the AMSPEs across varying perturbation levels and sample sizes. We observe that DFR is quite robust to perturbations and consistently outperforms GFR and SDR across different perturbation levels, especially for larger sample sizes. The only exception is the case of $n=100$ and $\nu=2$ where DFR is slightly inferior to GFR. As $n$ increases, DFR effectively detects the true signal and mitigates the impact of the perturbation, demonstrating its capability to de-noise the data.

\begin{table}[tb]
\single
\centering
\caption{Average mean squared prediction error of deep \f regression (DFR), global \f regression (GFR) \citep{mull:19:6} and sufficient dimension reduction (SDR) \citep{zhan:21:1} for distributional responses with different perturbation levels $\nu$.}
\label{tab:simu5}
\begin{tabular}{c|c|ccc}
\hline
$n$ & $\nu$ & DFR & GFR & SDR \\
\hline
\multirow{5}{*}{100} & 0.0 & 5.737 & 8.025 & 8.127\\
& 0.5 & 5.961 & 8.058 & 8.172\\
& 1.0 & 6.621 & 8.167 & 8.419\\
& 1.5 & 7.848 & 8.351 & 8.858\\
& 2.0 & 9.376 & 8.613 & 9.557\\
\hline
\multirow{5}{*}{200} & 0.0 & 3.281 & 7.376 & 6.433\\
& 0.5 & 3.428 & 7.386 & 6.437\\
& 1.0 & 3.935 & 7.434 & 6.530\\
& 1.5 & 4.767 & 7.517 & 6.701\\
& 2.0 & 5.715 & 7.635 & 6.975\\
\hline
\multirow{5}{*}{500} & 0.0 & 1.716 & 6.967 & 5.547\\
& 0.5 & 1.921 & 6.977 & 5.550\\
& 1.0 & 2.461 & 6.997 & 5.576\\
& 1.5 & 3.118 & 7.032 & 5.628\\
& 2.0 & 4.044 & 7.079 & 5.714\\
\hline
\multirow{5}{*}{1000} & 0.0 & 1.322 & 7.029 & 5.329\\
& 0.5 & 1.450 & 7.033 & 5.337\\
& 1.0 & 1.867 & 7.044 & 5.339\\
& 1.5 & 2.497 & 7.061 & 5.353\\
& 2.0 & 3.383 & 7.086 & 5.385\\
\hline
\end{tabular}
\end{table}

\subsection{Choice of Hyperparameters}
\label{supp:hyper}
The hyperparameters for ISOMAP and deep neural networks can be selected using a grid search over the candidate values listed in Table \ref{tab:hyper}. The optimal combination of hyperparameters is chosen to minimize the mean squared prediction error for the validation dataset.

\begin{table}[t]
    \single
    \centering
    \caption{Hyperparameter settings.}
    \label{tab:hyper}
    \begin{tabular}{lcccc}
    \hline
    Number of neighbors & 10 & 20 & 30 & 50 \\
    Number of layers    &  3 & 4  & 5  & 6  \\
    Number of neurons   &  8 & 16 & 32 & 64 \\
    Dropout rate        & 0.1& 0.2 & 0.3 & 0.4 \\
    Learning rate       & 0.0001 & 0.0005& 0.001& 0.005\\
    \hline
    \end{tabular}
\end{table}

\subsection{Data Application for Human Mortality Data}
\label{supp:mortality}
We illustrate the proposed method using the age-at-death distributions of 162 countries in 2015, obtained from United Nations databases (\url{http://data.un.org/}) and UN World Population Prospects 2022 Databases (\url{https://population.un.org/wpp/downloads}). For each country and age group, the life table reports the number of deaths aggregated in five-year intervals. These can be viewed as histograms with equal bin widths of five years, representing the number of deaths per age interval. We applied the \texttt{frechet} package \citep{chen:20} to obtain densities by local linear smoothing of the histograms, followed by standardization via trapezoid integration.

The age-at-death distribution in each country is influenced by various factors, such as economics, healthcare systems, and social and environmental factors. To investigate the impact of these factors, we considered nine predictors, as used in \citet{zhan:21:1}, including demographic, economic and environmental factors in 2015; see Table \ref{tab:pre_mor}.

\begin{table}[t]
    \single
    \centering
    \caption{Predictors of human mortality data, quantified as age-at-death distributions for various countries.}
\begin{tabular}{p{0.15\linewidth} | p{0.31\linewidth} | p{0.45\linewidth}}
\hline 
Category & Variables & Explanation \\
\hline 
\multirow{5}{*}{Demography} & 1. population Density & Population per square Kilometer \\
\cline{2-3}
& \multirow{2}{*}{2. Sex Ratio} & number of males per 100 females in the population \\
\cline{2-3}
& \multirow{2}{*}{3. Mean Childbearing Age} & average age of mothers at the birth of their children \\
\hline
\multirow{10}{*}{Economics} & 4. GDP & gross domestic product per capita \\
\cline{2-3}
& \multirow{3}{*}{5. GVA by Agriculture} & percentage of agriculture, hunting, forestry, and fishing activities of gross value added \\
\cline{2-3}
 & \multirow{2}{*}{6. CPI} & consumer price index treating 2010 as the base year\\
\cline{2-3}
& \multirow{2}{*}{7. Unemployment Rate} & percentage of unemployed people in the labor force\\ 
\cline{2-3}
& \multirow{2}{*}{8. Health Expenditure} & percentage of expenditure on health of GDP\\
\hline
Environment & 9. Arable Land & percentage of total land area \\
\hline 
\end{tabular}
    \label{tab:pre_mor}
\end{table}

The MSPE was calculated through leave-one-out cross-validation. From Table \ref{tab:mor}, the proposed method achieves the most accurate prediction results, with a 16\% and 4\% improvement compared to GFR and SDR, respectively. These results demonstrate that despite the data not residing on a two-dimensional manifold and having a small sample size, the proposed method is superior and robust. 

\begin{table}[t]
\single
\centering
\caption{Average mean squared prediction error of deep \f regression (DFR), global \f regression (GFR) \citep{mull:19:6} and sufficient dimension reduction (SDR) \citep{zhan:21:1} for human mortality data.}
\label{tab:mor}
\begin{tabular}{ccc}
\hline
DFR & GFR & SDR\\
\hline
26.377 & 31.322 & 27.602\\
\hline
\end{tabular}
\end{table}

Additionally, our proposed method also allows for the investigation of the effects of predictors. Here, we further investigate the effect of GDP and health expenditure on mortality distributions while holding other predictors constant at their median values. Figure \ref{fig:mor} (left) shows that the fitted age-at-death densities shift towards the right as GDP increases, indicating an increase in life expectancy. Additionally, the probability of death before age 5 declines with increasing GDP, suggesting a lower infant mortality rate. Similarly, as the percentage of health expenditure in GDP increases, the age-at-death densities also shift rightward, indicating improved life expectancy. Conversely, as the percentage of agricultural GVA increases, the age-at-death densities shift leftward, indicating a decrease in life expectancy. This trend is consistent with the observation that as a country develops, GDP increases and the agricultural share of GVA commonly decreases.

\begin{figure}[t]
    \single
    \centering
    \begin{subfigure}[]{0.33\textwidth}
        \centering
        \includegraphics[width=\linewidth]{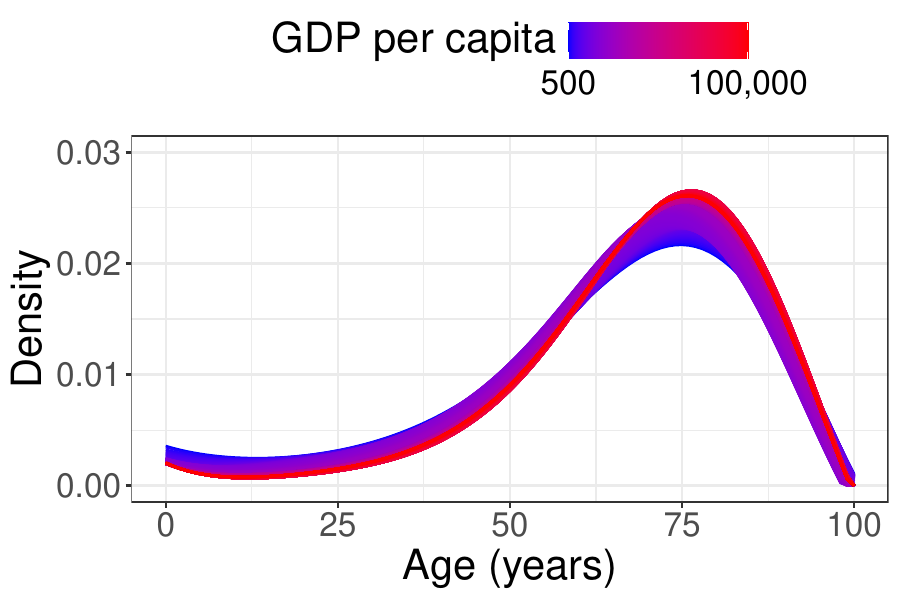}
    \end{subfigure}%
    \begin{subfigure}[]{0.33\textwidth}
        \centering
        \includegraphics[width=\linewidth]{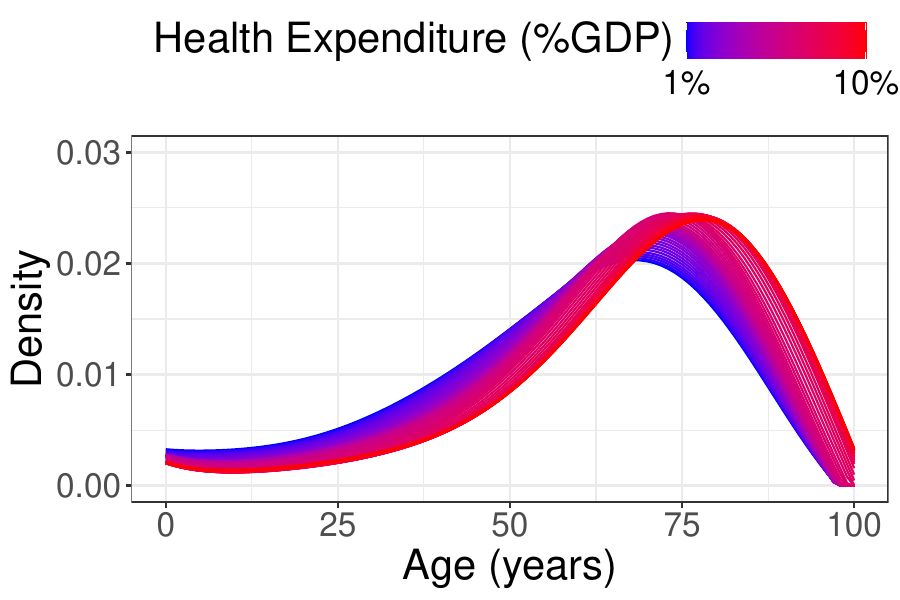}
    \end{subfigure}
    \begin{subfigure}[]{0.33\textwidth}
        \centering
        \includegraphics[width=\linewidth]{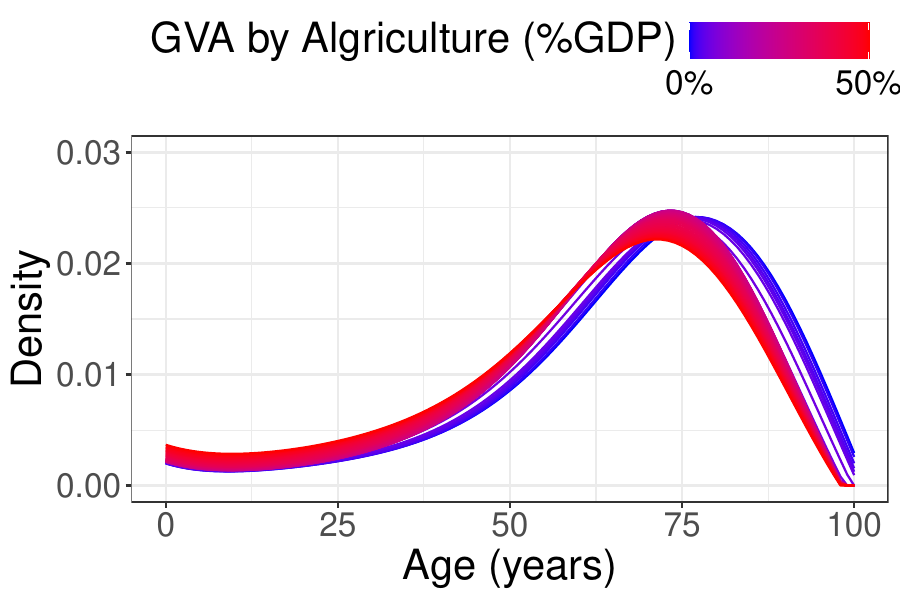}
    \end{subfigure}
    \caption{Age-at-death densities at different levels of GDP (left), health expenditure (middle) and agricultural GVA (right).}
    \label{fig:mor}
\end{figure}
\end{document}